\documentclass[aps,prx,amsmath,amssymb,groupedaddress,twocolumn,superscriptaddress,notitlepage,nofootinbib,longbibliography]{revtex4-2}
\pdfoutput=1

\usepackage[final]{graphicx}
\usepackage{times,bbm,amsmath,amssymb,amsthm}
\usepackage{epsfig,color}
\usepackage[dvipsnames]{xcolor}
\usepackage[colorlinks=true,citecolor=blue,linkcolor=blue]{hyperref}
\hypersetup{
    allcolors  = {blue},
}

\usepackage{cleveref}

\usepackage{float}
\usepackage[caption = false]{subfig}
\usepackage{verbatim}
\usepackage[greek,english]{babel}
\usepackage{thumbpdf,enumerate}
\usepackage{booktabs}
\usepackage{sidecap}
\usepackage[scaled=.8]{couriers}    
\usepackage{pstricks}
\usepackage{multirow}
\usepackage{placeins}
\usepackage{relsize}  
\usepackage{pst-grad,bm}
\usepackage{epigraph}
\usepackage{gensymb}
\usepackage{longtable}
\usepackage{booktabs}
\usepackage{gensymb}
\usepackage{enumitem}
\usepackage{amssymb}
\usepackage{soul}
\usepackage{ulem} 
\usepackage{acronym}
\usepackage{physics}
\usepackage{tikz}

\usepackage{todonotes}

\graphicspath{{Images/}}

\begin{document}

\author{Tommaso Francalanci}
\affiliation{Dipartimento di Fisica, Sapienza Universit\`{a} di Roma, Piazzale Aldo Moro 5, I-00185 Roma, Italy}

\author{Nicol\`o Spagnolo}
\affiliation{Dipartimento di Fisica, Sapienza Universit\`{a} di Roma, Piazzale Aldo Moro 5, I-00185 Roma, Italy}

\author{Mario Sigalotti}
\affiliation{Laboratoire Jacques-Louis Lions (LJLL), Inria, Sorbonne Universit\'e, Universit\'e de Paris, CNRS, Paris, France,}

\author{Eliott Z. Mamon}
\affiliation{Laboratoire d’Informatique de Paris 6, CNRS, Sorbonne Universit\'e, 4 Place Jussieu, 75005 Paris, France}

\author{Ulysse Chabaud}
\affiliation{DIENS, \'Ecole normale sup\'erieure, PSL University, CNRS, INRIA, 45 rue d’Ulm, Paris, 75005, France}

\author{Fabio Sciarrino}
\email{fabio.sciarrino@uniroma1.it}
\affiliation{Dipartimento di Fisica, Sapienza Universit\`{a} di Roma, Piazzale Aldo Moro 5, I-00185 Roma, Italy}

\title{Local controllability of heralded quantum linear optics}

\begin{abstract}
Photonic linear optical networks provide a versatile platform for quantum information processing and quantum state engineering. However, the set of states that can be generated using passive linear optics alone is fundamentally constrained by bosonic symmetries. Heralding, based on conditional measurements on auxiliary modes, is a widely used technique to overcome these limitations and effectively enlarge the set of accessible states. Despite the widespread use of heralding, it is often unclear how specific ancillary resources impact the overall reachability of the target space. In this work, we investigate the local controllability of photonic states in linear optical networks by analyzing the rank of the Jacobian of the output state with respect to the underlying unitary circuit, which provides a quantitative measure of the dimension of the accessible tangent space at a given configuration. Our analysis ranges from passive linear optics to heralded linear optics, where auxiliary resources and conditional measurements are included. Within this framework, we quantify how different resources enlarge the locally accessible state space beyond that of passive linear optics and determine the resources required for the Jacobian rank to reach its maximal value, thereby achieving full local controllability. As maximal local rank is a necessary condition for global reachability, our framework offers a systematic tool to assess and compare the accessible state space of measurement-based photonic architectures, and to establish practical criteria for the resources needed in high-dimensional quantum state engineering.
\end{abstract}

\maketitle
\section*{Introduction}

In recent years, photonic schemes for quantum computing have attracted increasing attention \cite{Kok2007}. A key milestone in this field is the seminal scheme of Knill, Laflamme, and Milburn (KLM) \cite{Knill2001}, which established that linear optics, single photons, photon-number measurements, and feed-forward are sufficient for universal quantum computation. Shortly after, several works clarified the underlying mechanism: conditional measurements on auxiliary modes generate effective photon–photon nonlinearities required for universal quantum computation, even though the underlying evolution remains strictly linear \cite{Lapaire2003, Scheel2003}. In parallel, linear optics has also played a central role in the study of quantum computational complexity. In particular, the Boson Sampling model \cite{Aaronson2011, Brod2019}, where one samples from the output distribution of indistinguishable photons evolving through a linear interferometer, was introduced as a candidate framework to demonstrate quantum computational speedup with linear optics. While standard Boson Sampling is based on passive linear optics (PLO) and photon counting, more recent works have explored the role of measurement-induced nonlinearities in this context, incorporating auxiliary photons and conditional detections to implement effective nonlinear transformations \cite{Chabaud2021, Spagnolo2023}. New architectures based on intermediate measurements and feed-forward have already been proposed, most notably Adaptive Linear Optics (ALO) \cite{Chabaud2021, Hoch2025} and State Injection (SI) \cite{Monbroussou2025, Monbroussou2025B}. These approaches aim to bridge the gap between Boson Sampling and universal quantum computing by progressively enhancing the computational power of linear optical systems through measurement-induced nonlinearities. This motivates the introduction of heralded linear optics (HLO), a general framework in which a global linear-optical evolution is followed by conditional measurements on a subset of modes, resulting in a heralded pure state on the remaining modes. In fact, at the level of fixed-outcome pure-state generation, the above architectures can be described within the HLO framework (see Appendix~\ref{app:other schemes}).

Beyond gate-based approaches \cite{Knill2002}, several works have directly addressed the problem of quantum state generation with heralding. A general framework for heralded state preparation, with applications to path-entangled states, was introduced in \cite{VanMeter2007}, where the problem is formulated as a system of polynomial equations in the entries of an effective linear transformation. This algebraic perspective has recently been extended using tools from computational algebraic geometry, such as the \textit{NulLA} algorithm, to decide feasibility without explicitly constructing solutions \cite{Singh2026}. For the simple case of two photons in two modes, Aniello \textit{et al.} \cite{Aniello2006} showed that adding a single ancilla mode is sufficient to reach any target state. Building on this result, de Gliniasty \textit{et al.} \cite{deGliniasty2024} systematically characterized the accessible state space for two-photon systems, deriving explicit criteria for reachability as a function of the available resources. More generally, in multimode and multi-photon regimes, it has been shown that target state preparation can be addressed by protocols that supplement multiport interferometers with measurement-based nonlinearities \cite{Kopylov2025, Aralov2026}. Within this framework, exact state preparation of arbitrary states with fixed photon number can be achieved through the injection of auxiliary photons and projective measurements \cite{Aralov2026}. While these methods establish valuable sufficient conditions for exact state preparation, the resulting resource requirements are tied to specific constructive procedures and may not be optimal. This leaves open the fundamental question of identifying the minimal physical resources strictly necessary for state engineering. On the foundational side, complementary theoretical efforts have investigated the geometry of quantum states under linear optical transformations, identifying invariants, characterizing the orbits generated by passive unitaries, and describing the reachable state space \cite{Migdal2014, Escartin2019, Parellada2024, Mamon2025, Draux2025, Kopylov2025, Rodari2025}. These studies make explicit the limitations of passive linear optics, where due to bosonic symmetries, only a restricted subset of Fock space can be explored. 

A major application of architectures based on heralding is the generation of multipartite entanglement \cite{Forbes2025}. Within this framework, several protocols have been developed to optimize the preparation of Bell-like states \cite{Stanisic2017, Gubarev2020, Fldzhyan2021, Fldzhyan2023}, scalable high-dimensional GHZ states \cite{Paesani2021, Chin2024}, and symmetric qudit Dicke states \cite{Kang2026}. These results highlight the trade-off between success probability and resource requirements, and emphasize the crucial role of auxiliary photons and modes in enabling the preparation of otherwise inaccessible states. It is worth noting that, historically, a large fraction of photonic experiments for entanglement generation have bypassed the strict requirements of heralding by relying instead on post-selection schemes, where all photons are destructively measured and the target state is identified only within the collected measurement data \cite{Bouwmeester1999, Pan2001, Lu2007, Pont2024}. While post-selection offers experimental simplicity, the resulting states are consumed upon measurement and cannot be preserved for subsequent quantum processing \cite{Forbes2025}. This motivates our focus on heralded linear optics, where the target state remains available in the unmeasured data modes for subsequent processing. Moreover, beyond entanglement generation, heralding finds broad applications across quantum information processing, including photonic quantum computing \cite{Paesani2021}, controllability in photonic quantum machine learning \cite{Carolan2020, Gan2022, Maring2024, Wang2025}, and bosonic simulators \cite{Yalouz2021}.

While it is well understood that heralding enlarges the set of reachable states, a systematic theory relating the available auxiliary resources (modes and photons) to the dimension of the accessible manifold is still missing. Here, we address this gap by investigating the structure and local dimension of the set of accessible states with HLO. This approach was introduced in \cite{Mamon2025} to study the local controllability of output states in passive linear optics, and we extend the analysis to scenarios with auxiliary resources and heralding. The approach is based on studying the rank of the Jacobian of the output state, which quantifies how the state varies as a function of the underlying linear optical circuit. Our main contributions are as follows. First, we derive analytical expressions for HLO output states and their infinitesimal variations under perturbations of the underlying circuit. Second, we provide general upper bounds on the dimension of the locally accessible state space within HLO, together with analytic lower bounds in specific regimes. Then, we formulate a conjecture for the generic value of the Jacobian rank, supported by extensive numerical evidence. Finally, building on \cite{deGliniasty2024} and \cite{Aralov2026}, we analyze the special two-photon regime, where our general conjecture does not apply due to the presence of additional symmetries. 

The remainder of the paper is organized as follows. Section~\ref{sec: background} introduces the Jacobian formalism for passive linear optics and HLO. The main results on the HLO Jacobian rank and the two-photon analysis are presented in Section~\ref{sec: results}, while Section~\ref{sec:global} discusses the implications for global reachability, i.e., the preparation of arbitrary heralded states.
 
\section{Background}\label{sec: background}

\subsection{Passive linear optics}
 
We consider here the passive linear optical (PLO) evolution of $n$ photons propagating through $m$ optical modes. The corresponding Hilbert space is the $n$-photon Fock space $\mathcal{F}_n^m = \text{span}(\{\ket{\bm{n}}=\ket{n_1, \dots, n_m} : n_j \geq 0,\sum_j n_j = n\}),$
whose dimension is
\begin{equation}
D = \dim \mathcal{F}_n^m = \binom{m+n-1}{n}.
\end{equation}
To describe the dynamics within this space, we introduce the creation and annihilation operators $a_j^\dagger$ and $a_j$, which are defined via $a_j^\dagger \ket{\bm{n}}=\sqrt{n_j+1} \ket{n_1,\dots,n_j+1,\dots,n_m}$ and
$a_j \ket{\bm{n}}=\sqrt{n_j} \ket{n_1,\dots,n_j-1,\dots,n_m}$.

As shown in Fig.~\ref{fig: Schemes1} (a), a PLO circuit acting on $m$ modes is described 
by an interferometer $U \in \mathcal{U}(m)$, where $\mathcal{U}(m)$ denotes the unitary group on $\mathbb{C}^m$.
In this work, we follow the convention where the element $U_{ij}$ represents the transition amplitude from input mode $i$ to output mode $j$ \cite{Brod2019}.
The corresponding $n$-photon PLO evolution is then governed by the unitary $\Gamma_n(U)$ acting on $\mathcal{F}_n^m$, where $\Gamma_n$ is the $n$-photon PLO representation \cite{Aniello-ExploringRepresentation-2006,Aaronson2011,Garcia-Escartin-MultiplePhoton-2019}. 
Under this representation, the creation operators transform via the so-called \textit{Bogoliubov transformations} \cite{Kok2007}:
\begin{equation}\label{eq:Bogoliubov}
\Gamma(U)\, a_i^\dagger\, \Gamma(U)^\dagger = \sum_{j=1}^m U_{ij}\, a_j^\dagger,
\end{equation}
where $\Gamma(U) = \bigoplus_{n=0}^\infty \Gamma_n(U)$ is the unitary representation of the interferometer on the full Fock space.

Considering an input Fock state $\ket{\bm{s}} = \ket{n_1, \dots, n_m}$, with $\sum_j n_j = n$, the corresponding output state resulting from PLO evolution is:
\begin{equation}
\ket{\psi(U,\bm{s})} = \Gamma_n(U)\ket{\bm{s}} \in \mathcal{F}_n^m.
\end{equation}
To characterize how the output state changes under small variations of the interferometer $U$, we study the \textit{Jacobian} of the map $U \mapsto \ket{\psi(U,\bm{s})}$ at a given point $U$, consists of the collection of directional derivatives $\delta_G \ket{\psi(U,\bm{s})}$ along a basis of movement directions $G$ around $U$.
As the domain $\mathcal{U}(m)$ is a \textit{Lie group}, these directions formally belong to the associated \textit{Lie algebra} $\mathfrak{u}(m)$ consisting of anti-Hermitian matrices \cite{Hall2015}:
\begin{equation}
\mathfrak{u}(m) = \{ G \in \,\mathbb{C}^{m \times m} : G^\dagger = -G \}.
\end{equation}
The unitaries $e^{tG} U$ for $G \in \mathfrak{u}(m)$ and small $t$ thus encode interferometers close to $U$.
A basis of the Lie algebra $\mathfrak{u}(m)$ can be taken as the generators $G^x_{jk}$, $G^y_{jk}$ ($1 \leq j < k \leq m$) and $G^z_{j}$ ($1 \leq j \leq m$), given by
\begin{align}
G^x_{jk} &= \frac{-i}{\sqrt{2}}(\ket{j}\bra{k} + \ket{k}\bra{j}),\\
G^y_{jk} &= \frac{1}{\sqrt{2}}(\ket{j}\bra{k} - \ket{k}\bra{j}),\\
G^z_{j} &= \ket{j}\bra{j}.
\end{align}
For each generator $G \in \mathfrak{u}(m)$, the infinitesimal variation of the output state is
\begin{equation}
\delta_G \ket{\psi(U,\bm{s})} = 
\left.\frac{d}{dt} \Gamma_n(e^{t G} U)\ket{\bm{s}}\right|_{t=0}.
\end{equation}
By stacking the real and imaginary parts of these tangent vectors as independent columns, we obtain the real Jacobian matrix $J(U,\bm{s})$, whose rank quantifies the number of locally independent directions accessible in the real space $\mathbb{R}^{2D}$.

The rank of the Jacobian matrix, $\mathrm{rank} \ J(U,\bm{s})$, corresponds to the local dimension of the orbit generated by the action of the linear optical group $\mathcal{U}(m)$ on the initial state $\ket{\bm{s}}$. Equivalently, this quantity measures the number of independent directions in state space that can be explored by infinitesimal variations of the interferometer parameters, and thus provides a notion of controllability of the output state, quantifying the number of independent directions in state space which can be accessed by small perturbations of the unitary interferometer.

The set of states reachable by PLO
\begin{equation}
\mathcal{O}_{\bm{s}} = \{ \Gamma_n(U)\ket{\bm{s}} : U \in \mathcal{U}(m) \}
\end{equation}
forms a smooth manifold \cite{Mamon2025}, which implies that its local dimension is the same at all its points. Therefore, the rank of the Jacobian does not depend on the choice of $U$ and can be equivalently computed at the identity:
\begin{equation}
\mathrm{rank} \ J(U,\bm{s}) = \mathrm{rank} \ J(\mathbb{I},\bm{s}) \equiv R_{\rm PLO}(\bm{s}).
\end{equation}

Evaluating the tangent space at the identity reveals that off-diagonal generators contribute non-trivially only when at least one of the two corresponding modes is populated, since otherwise their action annihilates the Fock state. Denoting by $m_e$ the number of empty modes in $\ket{\bm{s}}$, the number of active pairs is
\begin{equation}
N = \binom{m}{2} - \binom{m_e}{2}.
\end{equation}
Each active pair contributes two real directions, corresponding to $G^x_{jk}$ and $G^y_{jk}$, while the diagonal generators $\{G^z_j\}$ contribute only one additional direction associated with the global phase (which, in this work, is treated as an explicit direction of variation and is not factored out).
Hence, the Jacobian rank for passive linear optics is
\begin{equation}\label{eq: Lo_rank}
R_{\rm PLO}(\bm{s}) = 2N + 1 = m(m-1) - m_e(m_e-1) + 1.
\end{equation}
This expression coincides with the orbit-dimension result derived in \cite{Mamon2025}, where a general Lie-algebraic proof is provided.

The maximal attainable value, obtained when at most one mode is empty, is
\begin{equation}
R_{\rm PLO}^{\max}(m) = m(m-1) + 1.
\end{equation}
Since the real dimension of the state space (the unit sphere in $\mathcal{F}_n^m$) is $2D - 1$, with $D = \binom{m+n-1}{n}$, it follows that $R_{\rm PLO}^{\max}(m) \ll 2D - 1$ for large values of $n$ and $m$. 
This gap highlights a fundamental limitation \cite{Moyano-Fernandez-LinearOptics-2017,Monbroussou2025,Mamon2025} of passive linear optics: 
the accessible manifold is confined to a low-dimensional subset of the full Fock space.

\subsection{Heralded linear optics}
\begin{figure}[t]
\centering
\includegraphics[width=0.95\columnwidth]{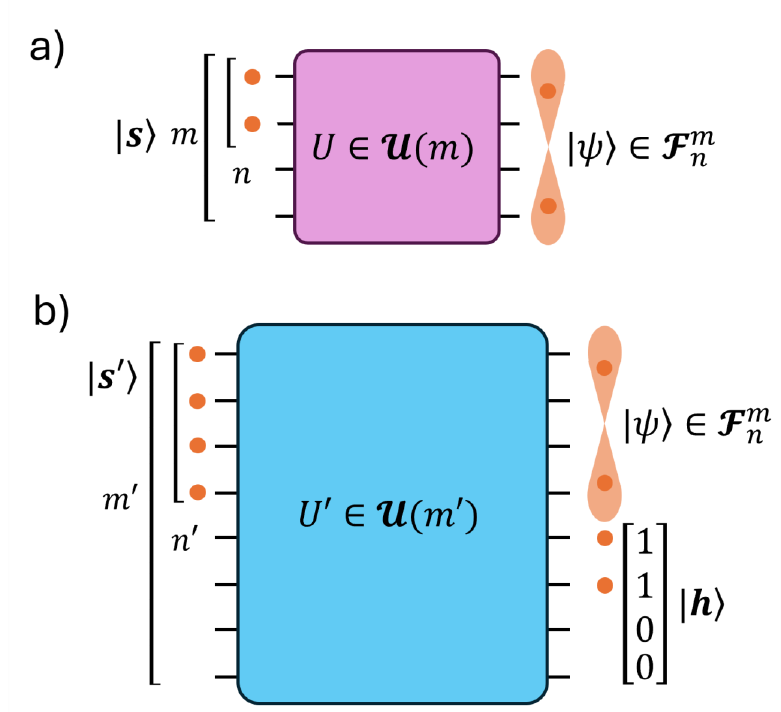}
\caption{
  Schematics of different settings:  
    (a) Passive linear optics, with $n$ photons in $m$ modes evolving through a circuit described by $U \in \mathcal{U}(m)$   
    (b) Heralded linear optics, where an enlarged input state comprising $n'$ photons in $m'$ modes is considered, evolving through a circuit described by $U' \in \mathcal{U}(m')$, followed by projection onto a fixed heralding state $\ket{\bm{h}}$ for heralding.
}
\label{fig: Schemes1}
\end{figure}
In the heralded linear optics framework, illustrated in Fig.~\ref{fig: Schemes1} (b), the standard PLO setting is extended by introducing an enlarged input state evolving through an expanded interferometer, followed by a conditional measurement on a subset of the modes.

The input state is taken to be a Fock state with $n'$ photons distributed over $m'$ modes,
\begin{equation}\label{eq:input}
\ket{\bm{s}'} = \ket{s'_1,\dots,s'_{m'}}, 
\qquad 
\sum_{i=1}^{m'} s'_i = n'.
\end{equation}
The state evolves under a global linear optical transformation described by a unitary $U' \in \mathcal{U}(m')$ acting on all modes.

After the evolution, a projective measurement is performed on a subset of modes, described by the projector
\begin{equation}
\Pi_{\bm{h}} = 
\mathbb{I}_{\mathcal{D}} \otimes \ket{\bm{h}}\!\bra{\bm{h}}_{\mathcal{H}},
\end{equation}
where $\ket{\bm{h}}$ is a Fock state with $r$ photons in $k$ modes specifying the heralding outcome on the set of $k$ measured modes $\mathcal{H}$. The remaining $m'-k$ data modes $\mathcal{D}$ define the output subsystem. In the general HLO setting, we fix the global parameters $(m', n')$ and the heralding configuration $(k, r)$, so that the output state lies in the Fock space of $n = n'-r$ photons in $m = m'-k$ modes. When comparing with PLO, however, it is often convenient to fix the output data parameters $(m,n)$ and regard $(k,r)$ (and thus $(m',n')$) as additional resources.  Different heralding outcomes $\bm{h}$ generally lead to distinct sets of reachable output states. If two outcomes $\bm{h}$ and $\bm{h}'$ are related by a permutation of the measured modes, they correspond to physically equivalent situations, since such permutations can be implemented by linear optical unitaries acting solely on the measured subsystem. Accordingly, the structure of the reachable states depends on $\bm{h}$ only through its occupation-number equivalence class, defined up to permutations of the heralded modes.

The (unnormalized) state obtained after heralding is given by
\begin{equation}
\ket{\Phi(U',\bm{s}',\bm{h})} = \Pi_{\bm{h}} \, \Gamma_{n'}(U') \, \ket{\bm{s}'} .
\end{equation}
Starting from this expression, the normalized heralded output state, restricted to the data modes only, can be obtained as
\begin{equation}\label{eq:normalized_heralded_output}
\ket{\psi(U',\bm{s'},\bm{h})}
=
\frac{\ket{\phi(U',\bm{s'},\bm{h})}}
{\sqrt{p(U',\bm{s'},\bm{h})}},
\end{equation}
where
\begin{equation}
\ket{\phi(U',\bm{s'},\bm{h})}
=
\bigl( \mathbb{I}_{\mathcal{D}} \otimes \bra{\bm{h}}_{\mathcal{H}}  \bigr)
\Gamma_{n'}(U') \ket{\bm{s}'}
\end{equation}
is the unnormalized heralded state restricted to the data modes, and
\begin{equation}
p(U',\bm{s'},\bm{h}) = \langle \phi(U',\bm{s'},\bm{h}) |
\phi(U',\bm{s'},\bm{h}) \rangle
\end{equation}
is the success probability of the heralding event (which is assumed to be non-zero in Eq. \eqref{eq:normalized_heralded_output}).

The normalized state $\ket{\psi}$ can be identified with a point on the unit sphere
$S^{2D-1} \subset \mathbb{R}^{2D}$, where
\begin{equation}
D=\dim\mathcal{F}_n^m = \binom{m+n-1}{n}
\end{equation}
is the dimension of the $n$-photon Fock space of the $m$ output modes, with $n=n'-r$. Since the map $U' \mapsto \ket{\psi(U',\bm{s'},\bm{h})}$ is complex-analytic in the entries of $U'$, its real and imaginary parts define a real-analytic map (see Appendix \ref{app: rank} for a formal definition of real-analyticity)
\begin{equation}
\tilde{\psi} :
\mathcal{X} \longrightarrow S^{2D-1} \subset \mathbb{R}^{2D},
\end{equation}
with domain
\begin{equation}
\mathcal{X} = \{ U' \in \mathcal{U}(m') :
p(U',\bm{s'},\bm{h}) > 0 \}.
\end{equation}

To generalize the Jacobian formalism to the heralded setting,
we consider infinitesimal variations of the normalized output state generated by
$G \in \mathfrak{u}(m')$:
\begin{equation}
\delta_G \ket{\psi(U',\bm{s'},\bm{h})} = \left.\frac{d}{dt}
\ket{\psi(e^{tG}U',\bm{s'},\bm{h})}\right|_{t=0}.
\end{equation}
This defines the real Jacobian matrix
\begin{equation}
J(U',\bm{s'},\bm{h})
\in \mathbb{R}^{2D \times M},
\qquad
M=\dim(\mathfrak{u}(m'))=(m')^2 ,
\end{equation}
whose rank quantifies the dimension of the tangent space to $S^{2D-1}$ spanned by infinitesimal unitary variations around $U'$.
By construction,
\begin{equation}
\mathrm{rank}\, J(U',\bm{s'},\bm{h}) \le 2D-1,
\end{equation}
where the subtraction of one accounts for normalization of the output state. 

Unlike the passive linear optical case, where the output states (under the unitary group action) form a manifold and the Jacobian rank is constant, the presence of the non-unitary projection associated with the heralding outcome $\bm{h}$ breaks this group structure. As a consequence, the Jacobian rank may depend on the specific unitary $U'$. For instance, if the unitary $U'$ does not couple the measured modes $\mathcal{H}$ with the remaining modes $\mathcal{D}$, the heralding operation does not affect the evolution of the output subsystem, and the rank reduces to the passive linear optical value. Conversely, for generic unitaries that couple the measured and unmeasured modes, the rank typically attains a maximal value determined solely by the input and heralding configuration $(\bm{s'},\bm{h})$.
Indeed, as shown in Appendix~\ref{app: rank}, since the map $\tilde{\psi}$ is real-analytic, its Jacobian rank is constant and maximal almost everywhere:
\begin{equation}
\mathrm{rank}\, J(U',\bm{s'},\bm{h})
=
R_{\rm HLO}(\bm{s'},\bm{h})
\quad \text{for almost all } U' \in \mathcal{X},
\end{equation}
where we define the \emph{generic HLO rank} as
\begin{equation}
R_{\rm HLO}(\bm{s'},\bm{h})
=
\max_{U' \in \mathcal{X}}
\mathrm{rank}\, J(U',\bm{s'},\bm{h}) .
\end{equation}

\section{Results}\label{sec: results} 

In this section, we present our main results on the structure and controllability of the output state space accessible via heralded linear optics (HLO). Our approach is based on an explicit characterization of the Jacobian of the HLO map and on the analysis of its generic rank. We begin by deriving in Section~\ref{sec: analytic} a closed-form expression for the heralded output state and for its derivatives with respect to the interferometer parameters, which provides a complete description of the Jacobian. Building on this formulation, we establish in Section~\ref{sec:upper bound} a general upper bound on the generic HLO rank as a function of $(m',n',k,r,k_{\rm occ})$, with $k_{\rm occ}$ denoting the number of non-empty modes in the heralding pattern, obtained from geometric constraints on submatrices of unitary matrices. In Section~\ref{sec:lower bounds}, we then complement this result with analytic lower bounds in specific regimes. In particular, for $n'\le k-k_{\rm occ}$, we provide constructive arguments establishing explicit lower bounds on the Jacobian rank both in the regime $\{r=0, m\geq n\}$ and in the regime $\{r>0,\, m>n',\, n>2\}$. These results naturally lead to Conjecture \ref{conj:HLO_rank}, predicting the generic value of the HLO rank: we conjecture that whenever either $r=0$ or $n>2$, and for output states where $m \geq n$, the upper bound is generically saturated. This conjecture, discussed in Section~\ref{sec:conjecture}, is supported by extensive numerical evidence. 
In Section~\ref{sec:two_photon_bound} we further analyze the special case of two photons ($n=2$), where the general conjecture does not apply. In this regime, building on \cite{deGliniasty2024} and \cite{Aralov2026}, we identify configurations with $m'=m+1$ modes and $r=n'-2$ photons in a single heralding mode that are sufficient to achieve the maximal allowed local rank, confirming from a geometric perspective that additional auxiliary modes are not generically required.

\subsection{Analytic structure of the HLO map}\label{sec: analytic}

To clarify how the Jacobian rank of the heralded state depends on the interferometer parameters, we provide in this section an explicit analytic expression for the heralded state and its derivatives. 

Consider an input state $\ket{\bm{s}'} = \prod_{i \in \mathcal{I}} \hat a_i^\dagger \ket{0}$, where $\mathcal{I}$ is the set of $n'$ occupied input modes. From now on, we restrict to collision-free input states, where each input mode contains at most one photon ($s'_i \in \{0,1\}$), which are both experimentally relevant and simplify the analytic treatment. 
Up to permutation, we can assume that $\mathcal{I}=\{1,\dots,n'\}$.
Under the transformation $U' \in \mathcal{U}(m')$, 
according to \eqref{eq:Bogoliubov},
the state before measurement evolves into:
\begin{equation}\label{eq:expanded_state}
\ket{\Psi} = \sum_{j_1,\dots,j_{n'}=1}^{m'} \bigg( \prod_{i\in\mathcal{I}} U'_{i,j_i} \bigg) \hat a_{j_1}^\dagger\dots \hat a_{j_{n'}}^\dagger \ket{0},
\end{equation}
where $j_i$ identifies the output mode of the photon originating from input $i$.

Heralding is implemented by projecting the $k$ modes in $\mathcal{H}$ onto a fixed pattern $\ket{\bm{h}}=\ket{r_1,...,r_k}$, with $\sum_{l=1}^kr_l = r$. Let $\mathcal{J} \subset \{1, \dots, m'\}$ denote the set of allowed output modes, consisting of all unmeasured modes $\mathcal{D}$ plus any measured mode $l \in \mathcal{H}$ that is occupied in the heralding pattern ($r_l > 0$). The unnormalized heralded state $\ket{\Phi}$ can be written by restricting the sum in Eq.~\eqref{eq:expanded_state} to configurations that satisfy the heralding constraint:
\begin{equation}
\label{eq: heralded_state}
\ket{\Phi} = \sum_{\substack{j_1,\dots,j_{n'} \in \mathcal{J} \\ \sum_{i \in \mathcal{I}} \delta_{j_i,l} = r_l \ \forall l\in \mathcal{H}}} \bigg( \prod_{i \in \mathcal{I}} U'_{i,j_i} \bigg) \hat a_{j_1}^\dagger \dots \hat a_{j_{n'}}^\dagger \ket{0}.
\end{equation}
This form highlights that the state depends on $U'$ only through the submatrix $V = U'[\mathcal{I}, \mathcal{J}] \in \mathbb{C}^{p \times q}$, where $p=n'$ and $q = m+k_{\rm occ}$. Specifically, for a given output Fock state $\ket{\bm{n}} \in \mathcal{F}_n^m$, the amplitude $c_{\bm{n}} = \langle \bm{n}, \bm{h} | \Phi \rangle$ is given by:
\begin{equation}
c_{\bm{n}} = \frac{\text{Perm}(U'_{\bm{s}, (\bm{n},\bm{h})})}{\sqrt{
\prod_j n_j! \prod_l r_l!}},
\end{equation}
where 
$\text{Perm}(A)$ denotes the permanent of a square matrix $A$ and $U'_{\bm{s}, (\bm{n},\bm{h})}$ is an $n' \times n'$ matrix constructed from $U'$ by selecting and repeating rows and columns according to the input occupation $\bm{s}'$ and the concatenated output and heralding patterns $(\bm{n}, \bm{h})$, respectively. Since these rows and columns belong to the sets $\mathcal{I}$ and $\mathcal{J}$, all matrix elements used to compute the amplitudes $c_{\bm{n}}$ are contained within the submatrix $V$.

We now derive an explicit expression for the variation of the heralded state under an infinitesimal transformation of the interferometer matrix $U'$. Let $G \in \mathfrak{u}(m')$ and consider the one-parameter family of unitaries $U'(t) = e^{tG} U'$. The corresponding variation of the (unnormalized) heralded state is defined as 
\begin{equation} 
\delta_G \ket{\Phi} = \left.\frac{d}{dt}\ket{\Phi(U'(t))}\right|_{t=0}. \end{equation} 
Using the explicit expansion of Eq.~\eqref{eq: heralded_state}, the dependence of $\ket{\Phi}$ on $U'$ appears only through products of matrix elements $\prod_{i \in \mathcal{I}} U'_{i,j_i}$. Differentiating such products with respect to $t$ yields
\begin{equation}
\frac{d}{dt} \prod_{i\in\mathcal{I}} U'_{i,j_i}(t) = \sum_{k\in\mathcal{I}} \bigg( \frac{d}{dt}U'_{k,j_k}(t) \prod_{i\neq k} U'_{i,j_i}(t) \bigg). \end{equation}
The derivative of each matrix element is
\begin{equation} \left. \frac{d}{dt} U'_{ij}(t) \right|_{t=0} = \left. \frac{d}{dt}(e^{tG}U')_{ij}\right|_{t=0} = (GU')_{ij}. 
\end{equation} 
Substituting this into the expansion of the heralded state yields \begin{equation}\label{eq: heralded derivative} 
\begin{aligned} 
&\delta_G \ket{\Phi} =\\
&= \sum_{\substack{j_1,\dots,j_{n'} \in \mathcal{J} \\ \sum_{i \in \mathcal{I}} \delta_{j_i,l} = r_l \ \forall l\in \mathcal{H}}} \sum_{k\in\mathcal{I}} \bigg( (GU')_{k,j_k} \prod_{i\neq k} U'_{i,j_i} \bigg) \hat a_{j_1}^\dagger \cdots \hat a_{j_{n'}}^\dagger \ket{0}. \end{aligned} 
\end{equation} 
This expression shows that the derivative of the heralded state can be obtained from Eq.~\eqref{eq: heralded_state} by replacing one factor $U'_{k,j_k}$ of the product $\prod_{i\in\mathcal{I}} U'_{i,j_i}$ with $(GU')_{k,j_k}$ while leaving all other factors unchanged (and summing over $k$). In particular, the derivative of the amplitude $c_{\bm{n}}$ with respect to a generator $G$ follows from the linearity of the permanent with respect to its rows:
\begin{equation}
\delta_G c_{\bm{n}} = \frac{1}{\mathcal{N}} \sum_{k=1}^{n'} \text{Perm}({U'}_k^{(G)})
\end{equation}
where ${U'}_k^{(G)}$ is the matrix $U'_{\bm{s}, (\bm{n},\bm{h})}$ where the $k$-th row has been replaced by the $k$-th row of the matrix $(GU')_{\bm{s}, (\bm{n},\bm{h})}$ restricted to the same row and column indices, and $\mathcal{N}$ is the normalization factor {\small$ \sqrt{\prod_j n_j! \prod_l r_l!}$}.

Since the heralding pattern $\ket{\bm{h}}$ is fixed, the post-measurement state can be written as $\ket{\Phi} = \ket{\phi}_{\mathcal{D}} \otimes \ket{\bm{h}}_{\mathcal{H}}$, where $\ket{\phi}$ denotes the state of the output subsystem $\mathcal{D}$. Fixing the Fock basis of $\mathcal{F}_n^m$, we identify $\ket{\phi}$ with a vector in $\mathbb{C}^D$, so that the map $U' \mapsto \ket{\phi}$ can be regarded as a complex-analytic map into $\mathbb{C}^D$. For the Jacobian rank, we are interested in the normalized state $\ket{\psi} = \ket{\phi}/\|\phi\|$. The variation $\delta_G \ket{\psi}$ is obtained by projecting the unnormalized variation onto the tangent space of the unit sphere:
\begin{equation}\label{eq: derivative normalization}
\delta_G \ket{\psi} = \frac{1}{\|\phi\|} \bigg( \delta_G \ket{\phi} - \mathrm{Re} \left[ \langle \psi | \delta_G \phi \rangle \right] \ket{\psi} \bigg).
\end{equation}

The real Jacobian $J(U', \bm{s}', \bm{h})$ is then constructed by expressing each variation $\delta_G \ket{\psi}$ in the Fock basis, and collecting the real and imaginary parts of the resulting coefficient vectors.

\subsection{Upper bound on the generic HLO rank}\label{sec:upper bound}

We now present a series of results describing how different input states and heralding patterns affect the generic Jacobian rank in heralded linear optics (HLO).

We begin with an upper bound on the
value of the HLO rank as a function of the input state and heralding configuration.

\newtheorem{theorem}{Theorem}
\begin{theorem}[Upper bound on the generic HLO rank]\label{thm:upper_bound}

Let $\ket{\bm{s'}} = |1^{\otimes n'} 0^{\otimes (m'-n')}\rangle$ be the input state,
and let the heralding projector be specified by a pattern $\ket{\bm{h}}$
with $r$ photons in $k$ modes.

Denote by $k_{\rm occ}$ the number of occupied modes in $\ket{\bm{h}}$, and
define
\begin{equation}
m = m' - k, \qquad n = n' - r,
\end{equation}
together with
\begin{equation}
p = n', \qquad q = m + k_{\rm occ}, \qquad P = n' + k_{\rm occ},
\end{equation}
and
\begin{equation}
\Omega = p-(m'-q) =  n' - (k - k_{\rm occ}).
\end{equation}

Then, the generic HLO  rank satisfies
\begin{equation}
R_{\rm HLO}(\bm{s'},\bm{h})
\;\le\;
\min\!\big(2D-1,\; f(m',n',k,k_{\rm occ})\big),
\end{equation}
where $D = \binom{m+n-1}{n}$ and
\begin{equation}
f(m',n',k,k_{\rm occ}) =
\begin{cases}
R_{\rm free}, & \Omega \le 0, \\[2mm]
\min\!\big(R_{\rm free}, R_{\rm sat}\big), & \Omega > 0,
\end{cases}
\end{equation}
with
\begin{equation}
\begin{aligned}
R_{\rm free} &= 2pq - 2P + 1, \\
R_{\rm sat}  &= 2pq - \Omega^2 -P + 1.
\end{aligned}
\end{equation}

\end{theorem}

The proof is given in Appendix~\ref{app:submatrix}. We summarize here the key ideas underlying the bound.

First, the heralded state $|\Phi\rangle$ depends on the interferometer only through the $p \times q$ submatrix $V$ connecting the occupied input modes to the output modes and the occupied modes in the heralding pattern. The complex entries of $V$ are subject to rescaling symmetries: phase shifts applied to input modes (rows) or to occupied heralding modes (columns) result only in an overall complex rescaling of $|\Phi\rangle$, which reduces to a global phase shift after normalization. As a consequence, $2P$ real parameters, with $P = n' + k_{\rm occ}$, effectively collapse to a single direction, leading to a reduction of $2P - 1$ in the Jacobian rank.

Second, the submatrix $V$ is constrained by its embedding into a larger unitary matrix $U'$. When $V$ is large enough to satisfy $\Omega > 0$, the orthogonality constraints of the global unitary become active, further restricting the number of independent parameters in the submatrix. This gives rise to the saturation regime $R_{\rm sat}$, where the rank is limited by unitarity constraints.

Finally, the rank is bounded above by $2D - 1$, the real dimension of the unit sphere in the $n$-photon, $m$-mode Hilbert space, corresponding to the case where the HLO configuration provides enough degrees of freedom to generically explore the entire state space locally.

\subsection{Lower bounds on the generic HLO rank in $\Omega \leq 0$ regimes}\label{sec:lower bounds}

We now derive lower bounds on the generic HLO Jacobian rank in the regime $\Omega \le 0$, where the submatrix $V$ admits a fully contractive embedding and its local coordinates are unconstrained. The bounds are derived for input states of the form $\ket{s'} = \ket{\bm{s}}_{\mathcal{D}}\otimes \ket{\bm{h}}_{\mathcal{H}}$, where $\ket{\bm{s}} = |1^{\otimes n} 0^{\otimes (m-n)}\rangle$ and $\ket{\bm{h}} = |1^{\otimes r} 0^{\otimes (k-r)}\rangle$.

\begin{theorem}[Lower bounds for $\Omega \le 0$]\label{thm:lower bounds}

For heralding patterns $\ket{\bm{h}} = \ket{\bm{0}} \equiv |0^{\otimes k}\rangle$ with no photons ($r=0$), $\Omega = n-k \leq 0$, and $n \leq m$, the generic HLO rank satisfies
\begin{equation}
R_{\rm{HLO}}(\bm{s'}, \bm{0}) \ge 
\min(2n(m-1) + 1,2D-1),
\end{equation}
where $D = \binom{m+n-1}{n}$.
For heralding patterns of the form $\ket{\bm{h}} = |1^{\otimes r} 0^{\otimes (k-r)}\rangle$ with $r>0$, $\Omega = n-k +2 r\leq0$, $n>2$ and $n+r \le m$, the HLO Jacobian rank instead satisfies
\begin{equation}
R_{\rm{HLO}}(\bm{s'}, \bm{h}) \ge 
\min(2m(n+r) - 1,2D-1).
\end{equation}
\end{theorem}

The bounds are obtained by identifying independent variations of the heralded state generated by variations of the submatrix $V$, evaluated at suitably chosen reference points $V^*$ where the data and heralding sectors decouple in a controlled way. The proof relies on a combinatorial analysis of the supports of the derivatives of the output state in the Fock basis, and on the absence of non-trivial linear dependencies among these variations in the considered settings. The detailed proof, including the explicit structure of derivatives and the combinatorial independence arguments, is given in Appendix~\ref{app:lower bounds}.

Notably, in specific regimes, the lower bounds established by Theorem \ref{thm:lower bounds} coincide with the upper bounds predicted by Theorem \ref{thm:upper_bound}, leading to the following exact characterization of the generic HLO rank:

\newtheorem{corollary}{Corollary}
\begin{corollary}[Exact HLO rank in $\Omega \leq 0$ regimes]
Under the same assumptions as in Theorem \ref{thm:lower bounds}, the generic HLO rank is exactly determined in the following configurations: 
For heralding patterns $\ket{\bm{h}} = \ket{\bm{0}}$ with no photons ($r=0$), $\Omega = n-k \leq 0$, and $n \leq m$:
\begin{equation}
R_{\rm{HLO}}(\bm{s'}, \bm{0}) = \min(2n(m-1) + 1, 2D-1).
\end{equation}
For heralding patterns with one heralded photon ($r=1$), $\Omega = n-k+2 \leq0$, $n>2$ and $n+1 \le m$:
\begin{equation}
R_{\rm{HLO}}(\bm{s'}, \bm{h}) = \min(2m(n+1) - 1,  2D-1).
\end{equation}
\end{corollary}

\newtheorem{conjecture}{Conjecture}
\subsection{Generic HLO rank conjecture}\label{sec:conjecture}

We now formulate a conjecture describing how different input states and heralding patterns affect the generic Jacobian rank in heralded linear optics (HLO), based on extensive numerical evidence.

\begin{conjecture}[Generic HLO rank for $n>2$]\label{conj:HLO_rank}

Under the same assumptions and notation as in Theorem~\ref{thm:upper_bound}, assume either $n = n'-r > 2$ (with arbitrary heralding pattern) or $r=0$ (with arbitrary $n$ and $k$).
Assume furthermore that the output state satisfies $n \le m$, where $m=m'-k$.

Then the upper bound given in Theorem~\ref{thm:upper_bound} is saturated, i.e.,
\begin{equation}
R_{\rm HLO}(\bm{s'}, \bm{h}) =
\min\!\big[2D-1,\,
f(m',n',k,k_{\rm occ})\big].
\end{equation}
\end{conjecture}

The conjecture is analytically proven in the regimes where the lower bounds derived in Section~\ref{sec:lower bounds} coincide with the upper bound of Theorem~\ref{thm:upper_bound}, and has been numerically tested for all configurations up to $n'=7$ and $m'=12$. In particular, to validate the conjecture, a numerical evaluation of the Jacobian rank was performed by computing the state derivatives from Eqs.~\eqref{eq: heralded derivative} and \eqref{eq: derivative normalization} at Haar-random points $U'$. The procedure was repeated for different input states of the form $\ket{\bm{s'}} = | 1^{\otimes n'} 0^{\otimes (m'-n')}\rangle$, corresponding to $n'$ photons distributed one per mode over $m'$ modes, and for different heralding patterns. The numerical results consistently support the conjecture in the regime $n \le m$. 

In the complementary regime $n > m$, which lies outside the scope of Conjecture~\ref{conj:HLO_rank}, the upper bound is still observed to be saturated in the majority of the tested configurations, with isolated exceptions. Within the explored parameter range, the only observed exception occurs for output states with $n=4$ photons in $m=3$ modes, obtained from an input state with $n'=5$ and $m' \geq 7$, and a heralding pattern with $r=1$ and $k \geq 4$. In this case, the upper bound predicts a rank 
\begin{equation}
    R_{\rm free} = 2pq - 2P + 1 = 29.
\end{equation}
However, numerical evaluation consistently yields a rank of $27$, revealing a systematic deficit of two real degrees of freedom. This discrepancy suggests the emergence in these configurations of non-trivial algebraic constraints that go beyond the rescaling symmetries and unitarity arguments considered here.

An immediate consequence of Conjecture~\ref{conj:HLO_rank} is a sufficient condition for achieving full controllability in configurations with $n>2$ and $n \leq m$. 

\begin{corollary}[Sufficient condition for maximal HLO rank]\label{cor:1}
 
Assuming that Conjecture~\ref{conj:HLO_rank} holds true, 
let $\ket{\bm{s'}} = \ket{1^{\otimes n'} 0^{\otimes (m'-n')}}$ be the input state and let the heralding projector be specified by the state $\ket{\bm{h}}$ with $r$ photons in $k$ modes, with $k_{\rm occ}$ occupied modes, yielding an output state with $n=n'-r>2$ photons in $m = m'-k \geq n$ modes. A sufficient condition for achieving the maximal HLO rank $R_{\rm HLO}(\bm{s'},\bm{h}) = 2D-1$, with $D=\binom{m+n-1}{n}$, is that the input state and the heralding pattern satisfy
\begin{equation}\label{eq:suff_condition_new}
f(m',n',k,k_{\rm occ}) \ \ge \ 2D-1.
\end{equation}
\end{corollary}

This corollary provides a practical guideline to identify the minimal input state required to saturate local controllability for a given target output space. To test this prescription, we examined a range of data outputs with $n>2$ photons, $n \leq m$, and corresponding Fock space dimensions up to $D = 70$. For each output, input configurations of the form $\ket{\bm{s'}} = \ket{1^{\otimes n'} 0^{\otimes (m'-n')}}$ and heralding configurations of the form $\ket{\bm{h}} = \ket{1^{\otimes r} 0^{\otimes (k-r)}}$, with $r = n'-n$ and $k = m'-m$, were chosen to satisfy the sufficient condition of Eq.~\eqref{eq:suff_condition_new}. We focus on configurations where $k_{\rm occ} = r$ (each auxiliary photon occupies a distinct heralding mode), as the rank is observed to increase with the number of occupied heralding modes. The corresponding HLO generic rank was then evaluated numerically at Haar-random points $U'$. To provide a systematic comparison across different inputs, the selection of minimal resources was performed using a score function that penalizes first the total number of photons $n'$ and then the total number of modes $m'$. This ensures that the reported configurations correspond to the smallest resources sufficient to saturate the rank. The results are summarized in Table~\ref{tab: min-aux-results}, which lists the target data outputs, the chosen input $(m',n')$, and the resulting HLO rank. In all cases, the HLO rank achieves the maximal value $2D-1$, providing a numerical indication of the validity of the corollary.

In summary, the corollary implies the following: when the input state is sufficiently large and appropriately populated, generic HLO dynamics attain the full tangent space dimension $2D-1$. This guarantees \emph{local controllability}, i.e., the ability to explore all infinitesimal directions around a given output state. However, to establish that this local property implies the ability to prepare (even approximately) \emph{any} target state in the output Hilbert space, one must consider the global structure of the reachable manifold. This is discussed in Section~\ref{sec:global}.

\begin{table}[t]
\centering

\begin{tabular}{c c c c}
\hline\hline
$(m,n)$ & $(m',n')$ & $D$ & $R_{\rm HLO}(\bm{s',\bm{h}})$ \\
\hline
$(3,3)$ & $(5,4)$ & $10$ & $19$ \\
$(4,3)$ & $(8,5)$ & $20$ & $39$ \\
$(4,4)$ & $(10,7)$ & $35$ & $69$ \\
$(5,3)$ & $(10,6)$ & $35$ & $69$ \\
$(6,3)$ & $(13,7)$  & $56$ &  $111$ \\
$(5,4)$ & $(14,9)$ & $70$ &  $139$ \\
\hline\hline
\end{tabular}
\caption{Minimal input state resources required to satisfy the condition of Corollary~\ref{cor:1} (Eq.~\ref{eq:suff_condition_new}) assuming Conjecture~\ref{conj:HLO_rank}, and corresponding numerically evaluated HLO ranks for selected input states with varying $n$ and $n\leq m$. Notably, in each case the HLO rank saturates the maximum value $2D-1$.}
\label{tab: min-aux-results}
\end{table}

\subsection{The two-photon case ($n=2$)}\label{sec:two_photon_bound}

We stress that Conjecture~\ref{conj:HLO_rank} explicitly excludes the case $n=2$ with arbitrary heralding patterns. In this regime, the structure of the accessible states is fundamentally constrained by the algebraic properties of two-photon states, which are in one-to-one correspondence with complex symmetric matrices. Indeed, any heralded two-photon output state can be written as
\begin{equation}\label{eq:two_photon}
\ket{\psi} = \sum_{i,j=1}^m S_{ij} \hat a_i^\dagger \hat a_j^\dagger \ket{0},
\end{equation}
where $S$ is an $m \times m$ complex symmetric matrix. As shown in~\cite{deGliniasty2024}, such states necessarily satisfy the constraint $\mathrm{rank}(S) \le n'$, where $n'$ is the number of input photons.

It is important to distinguish between global reachability and the local perspective adopted in this work. In~\cite{deGliniasty2024}, it is shown that the condition $n' \ge m$ is sufficient for global reachability of arbitrary two-photon states, provided that an unrestricted number of auxiliary vacuum modes is available. Moreover, in \cite{Aralov2026}, an analytical construction for exact state preparation with minimal resources is provided, showing that any two-photon state in $m$ modes can be generated using $n'=m$ input photons and a single auxiliary mode ($m'=m+1$), by conditioning on the detection of $m-2$ photons in the auxiliary port.

Our Jacobian-rank analysis provides a local geometric characterization consistent with these results. Since $S$ is a $m\times m$ matrix, the bound $\mathrm{rank}(S)\le n'$ imposes no algebraic restriction for $n' \ge m$. Since, moreover, the number of independent complex parameters of a full-rank $m \times m$ symmetric matrix is $m(m+1)/2$, which coincides with the complex dimension $D$ of the output Hilbert space in the case $n=2$, the representation of Eq.\ \eqref{eq:two_photon} is compatible with the fact that two-photon output states can cover a set of full dimension in the output manifold $S^{2D-1}$. For $n' < m$, the input photon number imposes a rank constraint on the state's amplitude matrix, leading to the bound
\begin{equation}\label{eq:two_photon_bound}
R_{\rm HLO} \le 2D-1 - (m-n')(m-n'+1).
\end{equation}
In this case, the bound corresponds to the number of real degrees of freedom of a symmetric complex matrix with rank at most $n'$, reduced by one to account for the overall normalization of the state.

Our numerical results show that these bounds are already saturated within a minimal configuration in terms of auxiliary modes, consistent with the resource-efficient construction proposed in \cite{Aralov2026}. Specifically, for any $m$, and any number of input photons $n' \geq 2$, we find that a configuration with a single auxiliary mode ($m' = m+1$) and a heralding pattern $\ket{\bm{h}} = \ket{n'-2}$ (where all $r = n'-2$ auxiliary photons are collected in the single auxiliary mode) is sufficient to saturate the available degrees of freedom provided by $n'$.  In this configuration, which is minimal with respect to the number of modes, a transition in the role of the input photons $n'$ emerges:
\begin{itemize}[leftmargin=*]
\item If $n' \ge m$, the Jacobian rank attains the maximal value $2D-1$.
\item If $n' < m$, the Jacobian rank saturates the upper bound in Eq.~\eqref{eq:two_photon_bound}.
\end{itemize}
Increasing the number of auxiliary vacuum modes beyond $m'=m+1$ does not lead to any further increase in the Jacobian rank. In particular, the single additional empty auxiliary mode is needed in the $m=2$ case to overcome the geometric constraints of linear optics (as already observed in~\cite{Aniello2006}).  

To summarize, the inapplicability of Conjecture~\ref{conj:HLO_rank} for $n=2$ reflects a genuine structural transition. While for $n>2$ heralding generically unlocks all available degrees of freedom up to the bound predicted by Theorem \ref{thm:upper_bound}, the two-photon case is uniquely restricted by the rank of the underlying symmetric matrix $S$.

\section{Global Reachability}\label{sec:global}

The focus of this section is global reachability for the HLO map, i.e., the ability to approximate any target pure state arbitrarily well with some HLO circuit. To this end, we consider the real normalized HLO map $\tilde{\psi}: \mathcal{X} \to S^{2D-1}$, which assigns to each unitary $U'\in\mathcal{U}(m')$ the real and imaginary parts of the corresponding normalized data state:
\begin{equation}
    \tilde{\psi}(U') = \bigl( \mathrm{Re}\ket{\psi(U',\bm{s'},\bm{h})}, \   \mathrm{Im}\ket{\psi(U',\bm{s'},\bm{h})} \bigr).
\end{equation}
In the following, we restrict our attention to configurations where the HLO map satisfies local full controllability, meaning that its generic Jacobian rank saturates the dimension of the output sphere, $R_{\mathrm{HLO}}(\bm{s'},\bm{h}) = 2D-1$. Under this condition, the Jacobian rank of the map is generically constant and equal to $2D-1$ almost everywhere on the open domain $\mathcal{X}$ where the success probability is nonzero.

If $\tilde{\psi}$ were a smooth submersion on the entire compact manifold $\mathcal{U}(m')$ (i.e., if it had no singular points and maximal rank everywhere), its image $\mathcal{Y} = \tilde{\psi}(\mathcal{U}(m'))$ would coincide with the whole sphere $S^{2D-1}$, which would imply full global reachability (see Appendix \ref{app: global}). In practice, however, there may be a nonempty singular locus,
\begin{equation}
    \mathcal{Z}=\mathcal{U}(m')\setminus\mathcal{X} =\{U'\in\mathcal{U}(m'): p(U',\bm{s'},\bm{h})=0\},
\end{equation} 
corresponding to unitaries for which the heralding probability vanishes.
Moreover, the critical locus
\begin{equation} \mathcal{C}=\{U'\in\mathcal{X}:\mathrm{rank}\,J(U')<R_{\mathrm{HLO}}(\bm{s'},\bm{h})\} \end{equation}
identifying points where the Jacobian rank drops below its generic value, is known to be nonempty for the HLO map. A notable example is the case of unitaries that do not couple the heralded and data modes, where, as previously noticed, the rank reduces to the lower value characteristic of passive linear optics. As discussed in Appendix \ref{app: global}, these loci can in general obstruct global surjectivity creating boundaries in the image of the map. Thus, although $\tilde{\psi}$ is a submersion almost everywhere, a rigorous proof that maximal rank implies global reachability remains an open question. We therefore formulate the following conjecture:

\begin{conjecture}[Global reachability for maximal rank HLO map]\label{conj: 3}
Suppose the HLO map $\tilde{\psi}: \mathcal{X}\to S^{2D-1}$ attains maximal rank almost everywhere on $\mathcal{X}$, i.e.:
\begin{equation}
R_{\rm HLO}(\bm{s'},\bm{h}) = 2D-1.
\end{equation}
Then its image $\mathcal{Y} = \tilde{\psi}(\mathcal{X})$ is dense in $S^{2D-1}$.
\end{conjecture}

Since normalized pure states in the $D$-dimensional complex Fock space correspond to points on the real sphere $S^{2D-1}$, the validity of this conjecture implies that any target state can be approximated arbitrarily well by some HLO circuit, provided that the resource thresholds identified by Corollary \ref{cor:1} (assuming Conjecture \ref{conj:HLO_rank}) are met. A stronger version of this conjecture would be to assume that the image coincides with the whole sphere $S^{2D-1}$, meaning that any target state can be prepared exactly by some HLO circuit. We note that, while in many cases these two statements coincide \cite{Boscain-Approximate-2015}, this is not guaranteed for the HLO map.

The problem of exact state preparation has been recently addressed in \cite{Aralov2026}, where it is shown that any target state with a fixed number of photons can be synthesized given a sufficient number of auxiliary photons and projective measurements. However, identifying exactly the minimal resource requirement remains an open question. Our Jacobian-rank analysis provides a complementary, geometric perspective by identifying the threshold at which the HLO map becomes locally fully controllable, i.e., when the Jacobian rank saturates to $2D-1$. Notably, this threshold can be reached with fewer photons than those required by the currently known constructive schemes. For instance, while the exact preparation of a general 3-photon, 3-mode state is guaranteed in \cite{Aralov2026} using $15$ photons in total, we observe that the Jacobian rank already saturates for the same output configuration with as little as $n'=4$ photons. This suggests that, if Conjecture~\ref{conj: 3} holds, the threshold for global reachability may be attained at significantly lower resource levels than those required by current constructive protocols, leaving room for substantial improvements in existing resource bounds.

Moreover, while Conjecture \ref{conj: 3} remains an open question, the Jacobian-rank criterion still provides a powerful target-independent no-go result. When the generic rank is not maximal ($R_{\rm HLO} < 2D-1$), the accessible set $\mathcal{A}$ is a subset of the image of a smooth map with rank bounded by $R_{\rm HLO}$. By Sard’s theorem \cite{Sard1942}, such an image has measure zero in $S^{2D-1}$, which means that in that case exact state preparation is impossible for almost all target states.
This yields a necessary condition for generic state preparation that is stricter than the fundamental lower bound $n' \ge \max(n, m)$ established in \cite{Aralov2026}. Indeed, as shown in Table \ref{tab: min-aux-results}, the number of photons $n'$ required for the Jacobian rank to reach its maximal value is consistently higher than this lower bound, reflecting additional geometric constraints imposed by the structure of the HLO map.

It is important to note, however, that while this geometric approach establishes the existence of a full-measure set of unreachable states when rank is not maximized, it does not directly identify which specific targets are excluded. The rank condition identifies an obstruction based on the global counting of degrees of freedom, but does not yield state-specific no-go results beyond those already known for restricted settings such as passive linear optics \cite{Migdal2014, Escartin2019, Parellada2024}. Moreover, a measure-zero image does not, in principle, exclude density.

To gain further insight into the accessibility of the state space, we investigate the reachability of Haar-random target states through numerical optimization. As detailed in Appendix \ref{app: global numerical}, we study the minimum infidelity achieved for a large set of random targets as a function of the available HLO resources. Our numerical experiments demonstrate that when $R_{\rm HLO} < 2D-1$, the optimization consistently plateaus at a finite distance from the target, suggesting that the accessible set is not only of measure zero but also fails to be dense. In contrast, when the rank is maximal, the infidelity exhibits a sustained exponential decay with the number of optimization steps, providing numerical support for Conjecture~\ref{conj: 3}.

Finally, in Appendix \ref{app: frame potential}, we further explore numerically the global properties of the HLO map by computing the frame potentials of the generated state ensembles \cite{Ambainis2007,Mele2024}. This analysis confirms that the enhanced local controllability granted by HLO translates into increasing expressivity. Specifically, we observe that the output distribution of HLO circuits can approximate high-order $t$-designs more closely than PLO circuits, and the approximation improves with increased ancillary and heralding resources.

\section{Conclusion}

In this work, we extended the Jacobian-rank approach proposed in \cite{Mamon2025} to characterize local controllability in heralded linear optical networks. By analyzing the tangent space generated by infinitesimal variations of the underlying unitary circuit, we characterized how the introduction of enlarged input states together with projective measurements expands the locally accessible state space. In particular, we formulated upper and lower bounds on the Jacobian rank attainable in heralded linear optics, addressing the general case with a conjecture based on analytic results in specific regimes and extensive numerical simulations. These results provide a new, target-independent characterization of local controllability in heralded schemes, and yield explicit bounds on the resources necessary to achieve maximal local controllability. In particular, our results cover configurations with \(n>2\) photons, for which limited results were previously available.

It is instructive to contrast the present Jacobian-rank approach with algebraic feasibility methods such as \textit{NulLA} \cite{Singh2026}. While algebraic methods provide powerful state-specific no-go certificates, they do not yield information on the dimension or geometry of the reachable state space, nor do they provide target-independent statements on controllability. By contrast, the Jacobian-rank approach directly characterizes the dimension of the locally accessible manifold and provides a scalable, target-independent criterion to assess controllability. Moreover, a rank strictly smaller than the maximal value immediately implies an obstruction to generic state preparation, which is not directly accessible via such infeasibility methods. It is also natural to compare our findings with recently proposed algebraic synthesis protocols \cite{Kopylov2025, Aralov2026}. While these works provide sufficient conditions for global reachability, demonstrating that any state can be prepared exactly given enough auxiliary photons, their resource requirements may not be optimal in general. In contrast, our geometric analysis identifies a regime of full local controllability for the HLO map, characterized by the saturation of the Jacobian rank. Notably, we find that this regime is reached at resource levels lower than those required by existing constructive protocols for exact state preparation. Furthermore, while algebraic protocols such as the one proposed in \cite{Aralov2026} rely on the sequential and approximate application of many photon-addition modules, our framework leverages the full expressive power of a single, global unitary. The fact that the Jacobian rank saturates with fewer resources suggests that state preparation may be more efficiently addressed through global optimization of the interferometer than via modular algebraic synthesis, thus reducing the experimental overhead.

Throughout this work, a state is considered reachable whenever it can be generated with nonzero heralding probability, irrespective of the magnitude of that probability. A quantitative study of the scalings of these probabilities therefore remains an interesting follow-up direction. Another promising direction concerns the direct characterization of controllability in photonic protocols designed to go beyond passive linear optics. Indeed, as discussed in Appendix~\ref{app:other schemes}, linear optical architectures involving intermediate measurements can be recast within the heralded linear optics framework at the level of single measurement branches. This shows that the present approach can be directly applied to assess local controllability in a broad class of photonic protocols, including adaptive and measurement-based schemes. 

Overall, this work establishes a framework to assess controllability and resource requirements in high-dimensional photonic state engineering, and provides a quantitative tool to characterize and compare photonic architectures incorporating intermediate measurements or heralding. We expect these results to offer useful guidance in the design of scalable photonic protocols and to stimulate further investigation of the interplay between measurement, resources, and controllability in quantum optical systems.

We provide the Python source code for the numerical evaluation of the HLO Jacobian rank on GitHub \cite{HloRepo}.

\section*{Acknowledgements}

This work is supported by the ERC Advanced Grant QUBOSS (QUantum advantage via non-linear BOSon Sampling, Grant Agreement No. 884676), the ICSC--Centro Nazionale di Ricerca in High Performance Computing, Big Data and Quantum Computing, funded by European Union--NextGenerationEU, and the European Union’s Horizon Europe research and innovation program under EPIQUE Project (Grant Agreement No. 101135288). EZM is supported by the grant ANR-22-PNCQ-0002.

%\bibliography{main}

\begin{thebibliography}{57}%
\makeatletter
\providecommand \@ifxundefined [1]{%
 \@ifx{#1\undefined}
}%
\providecommand \@ifnum [1]{%
 \ifnum #1\expandafter \@firstoftwo
 \else \expandafter \@secondoftwo
 \fi
}%
\providecommand \@ifx [1]{%
 \ifx #1\expandafter \@firstoftwo
 \else \expandafter \@secondoftwo
 \fi
}%
\providecommand \natexlab [1]{#1}%
\providecommand \enquote  [1]{``#1''}%
\providecommand \bibnamefont  [1]{#1}%
\providecommand \bibfnamefont [1]{#1}%
\providecommand \citenamefont [1]{#1}%
\providecommand \href@noop [0]{\@secondoftwo}%
\providecommand \href [0]{\begingroup \@sanitize@url \@href}%
\providecommand \@href[1]{\@@startlink{#1}\@@href}%
\providecommand \@@href[1]{\endgroup#1\@@endlink}%
\providecommand \@sanitize@url [0]{\catcode `\\12\catcode `\$12\catcode
  `\&12\catcode `\#12\catcode `\^12\catcode `\_12\catcode `\%12\relax}%
\providecommand \@@startlink[1]{}%
\providecommand \@@endlink[0]{}%
\providecommand \url  [0]{\begingroup\@sanitize@url \@url }%
\providecommand \@url [1]{\endgroup\@href {#1}{\urlprefix }}%
\providecommand \urlprefix  [0]{URL }%
\providecommand \Eprint [0]{\href }%
\providecommand \doibase [0]{https://doi.org/}%
\providecommand \selectlanguage [0]{\@gobble}%
\providecommand \bibinfo  [0]{\@secondoftwo}%
\providecommand \bibfield  [0]{\@secondoftwo}%
\providecommand \translation [1]{[#1]}%
\providecommand \BibitemOpen [0]{}%
\providecommand \bibitemStop [0]{}%
\providecommand \bibitemNoStop [0]{.\EOS\space}%
\providecommand \EOS [0]{\spacefactor3000\relax}%
\providecommand \BibitemShut  [1]{\csname bibitem#1\endcsname}%
\let\auto@bib@innerbib\@empty
%</preamble>
\bibitem [{\citenamefont {Kok}\ \emph {et~al.}(2007)\citenamefont {Kok},
  \citenamefont {Munro}, \citenamefont {Nemoto}, \citenamefont {Ralph},
  \citenamefont {Dowling},\ and\ \citenamefont {Milburn}}]{Kok2007}%
  \BibitemOpen
  \bibfield  {author} {\bibinfo {author} {\bibfnamefont {P.}~\bibnamefont
  {Kok}}, \bibinfo {author} {\bibfnamefont {W.~J.}\ \bibnamefont {Munro}},
  \bibinfo {author} {\bibfnamefont {K.}~\bibnamefont {Nemoto}}, \bibinfo
  {author} {\bibfnamefont {T.~C.}\ \bibnamefont {Ralph}}, \bibinfo {author}
  {\bibfnamefont {J.~P.}\ \bibnamefont {Dowling}},\ and\ \bibinfo {author}
  {\bibfnamefont {G.~J.}\ \bibnamefont {Milburn}},\ }\bibfield  {title}
  {\bibinfo {title} {Linear optical quantum computing with photonic qubits},\
  }\href {https://doi.org/10.1103/revmodphys.79.135} {\bibfield  {journal}
  {\bibinfo  {journal} {Reviews of Modern Physics}\ }\textbf {\bibinfo {volume}
  {79}},\ \bibinfo {pages} {135–174} (\bibinfo {year} {2007})}\BibitemShut
  {NoStop}%
\bibitem [{\citenamefont {Knill}\ \emph {et~al.}(2001)\citenamefont {Knill},
  \citenamefont {Laflamme},\ and\ \citenamefont {Milburn}}]{Knill2001}%
  \BibitemOpen
  \bibfield  {author} {\bibinfo {author} {\bibfnamefont {E.}~\bibnamefont
  {Knill}}, \bibinfo {author} {\bibfnamefont {R.}~\bibnamefont {Laflamme}},\
  and\ \bibinfo {author} {\bibfnamefont {G.~J.}\ \bibnamefont {Milburn}},\
  }\bibfield  {title} {\bibinfo {title} {A scheme for efficient quantum
  computation with linear optics},\ }\href {https://doi.org/10.1038/35051009}
  {\bibfield  {journal} {\bibinfo  {journal} {Nature}\ }\textbf {\bibinfo
  {volume} {409}},\ \bibinfo {pages} {46–52} (\bibinfo {year}
  {2001})}\BibitemShut {NoStop}%
\bibitem [{\citenamefont {Lapaire}\ \emph {et~al.}(2003)\citenamefont
  {Lapaire}, \citenamefont {Kok}, \citenamefont {Dowling},\ and\ \citenamefont
  {Sipe}}]{Lapaire2003}%
  \BibitemOpen
  \bibfield  {author} {\bibinfo {author} {\bibfnamefont {G.~G.}\ \bibnamefont
  {Lapaire}}, \bibinfo {author} {\bibfnamefont {P.}~\bibnamefont {Kok}},
  \bibinfo {author} {\bibfnamefont {J.~P.}\ \bibnamefont {Dowling}},\ and\
  \bibinfo {author} {\bibfnamefont {J.~E.}\ \bibnamefont {Sipe}},\ }\bibfield
  {title} {\bibinfo {title} {Conditional linear-optical measurement schemes
  generate effective photon nonlinearities},\ }\href
  {https://doi.org/10.1103/physreva.68.042314} {\bibfield  {journal} {\bibinfo
  {journal} {Physical Review A}\ }\textbf {\bibinfo {volume} {68}},\ \bibinfo
  {pages} {042314} (\bibinfo {year} {2003})}\BibitemShut {NoStop}%
\bibitem [{\citenamefont {Scheel}\ \emph {et~al.}(2003)\citenamefont {Scheel},
  \citenamefont {Nemoto}, \citenamefont {Munro},\ and\ \citenamefont
  {Knight}}]{Scheel2003}%
  \BibitemOpen
  \bibfield  {author} {\bibinfo {author} {\bibfnamefont {S.}~\bibnamefont
  {Scheel}}, \bibinfo {author} {\bibfnamefont {K.}~\bibnamefont {Nemoto}},
  \bibinfo {author} {\bibfnamefont {W.~J.}\ \bibnamefont {Munro}},\ and\
  \bibinfo {author} {\bibfnamefont {P.~L.}\ \bibnamefont {Knight}},\ }\bibfield
   {title} {\bibinfo {title} {Measurement-induced nonlinearity in linear
  optics},\ }\href {https://doi.org/10.1103/physreva.68.032310} {\bibfield
  {journal} {\bibinfo  {journal} {Physical Review A}\ }\textbf {\bibinfo
  {volume} {68}},\ \bibinfo {pages} {032310} (\bibinfo {year}
  {2003})}\BibitemShut {NoStop}%
\bibitem [{\citenamefont {Aaronson}\ and\ \citenamefont
  {Arkhipov}(2011)}]{Aaronson2011}%
  \BibitemOpen
  \bibfield  {author} {\bibinfo {author} {\bibfnamefont {S.}~\bibnamefont
  {Aaronson}}\ and\ \bibinfo {author} {\bibfnamefont {A.}~\bibnamefont
  {Arkhipov}},\ }\bibfield  {title} {\bibinfo {title} {The computational
  complexity of linear optics},\ }in\ \href
  {https://doi.org/10.1145/1993636.1993682} {\emph {\bibinfo {booktitle}
  {Proceedings of the forty-third annual ACM symposium on Theory of
  computing}}},\ \bibinfo {series and number} {STOC’11}\ (\bibinfo
  {publisher} {ACM},\ \bibinfo {year} {2011})\ p.\ \bibinfo {pages}
  {333–342}\BibitemShut {NoStop}%
\bibitem [{\citenamefont {Brod}\ \emph {et~al.}(2019)\citenamefont {Brod},
  \citenamefont {Galv{\~a}o}, \citenamefont {Crespi}, \citenamefont {Osellame},
  \citenamefont {Spagnolo},\ and\ \citenamefont {Sciarrino}}]{Brod2019}%
  \BibitemOpen
  \bibfield  {author} {\bibinfo {author} {\bibfnamefont {D.~J.}\ \bibnamefont
  {Brod}}, \bibinfo {author} {\bibfnamefont {E.~F.}\ \bibnamefont
  {Galv{\~a}o}}, \bibinfo {author} {\bibfnamefont {A.}~\bibnamefont {Crespi}},
  \bibinfo {author} {\bibfnamefont {R.}~\bibnamefont {Osellame}}, \bibinfo
  {author} {\bibfnamefont {N.}~\bibnamefont {Spagnolo}},\ and\ \bibinfo
  {author} {\bibfnamefont {F.}~\bibnamefont {Sciarrino}},\ }\bibfield  {title}
  {\bibinfo {title} {{Photonic implementation of boson sampling: a review}},\
  }\href {https://doi.org/10.1117/1.AP.1.3.034001} {\bibfield  {journal}
  {\bibinfo  {journal} {Advanced Photonics}\ }\textbf {\bibinfo {volume} {1}},\
  \bibinfo {pages} {034001} (\bibinfo {year} {2019})}\BibitemShut {NoStop}%
\bibitem [{\citenamefont {Chabaud}\ \emph {et~al.}(2021)\citenamefont
  {Chabaud}, \citenamefont {Markham},\ and\ \citenamefont
  {Sohbi}}]{Chabaud2021}%
  \BibitemOpen
  \bibfield  {author} {\bibinfo {author} {\bibfnamefont {U.}~\bibnamefont
  {Chabaud}}, \bibinfo {author} {\bibfnamefont {D.}~\bibnamefont {Markham}},\
  and\ \bibinfo {author} {\bibfnamefont {A.}~\bibnamefont {Sohbi}},\ }\bibfield
   {title} {\bibinfo {title} {Quantum machine learning with adaptive linear
  optics},\ }\href {https://doi.org/10.22331/q-2021-07-05-496} {\bibfield
  {journal} {\bibinfo  {journal} {Quantum}\ }\textbf {\bibinfo {volume} {5}},\
  \bibinfo {pages} {496} (\bibinfo {year} {2021})}\BibitemShut {NoStop}%
\bibitem [{\citenamefont {Spagnolo}\ \emph {et~al.}(2023)\citenamefont
  {Spagnolo}, \citenamefont {Brod}, \citenamefont {Galvão},\ and\
  \citenamefont {Sciarrino}}]{Spagnolo2023}%
  \BibitemOpen
  \bibfield  {author} {\bibinfo {author} {\bibfnamefont {N.}~\bibnamefont
  {Spagnolo}}, \bibinfo {author} {\bibfnamefont {D.~J.}\ \bibnamefont {Brod}},
  \bibinfo {author} {\bibfnamefont {E.~F.}\ \bibnamefont {Galvão}},\ and\
  \bibinfo {author} {\bibfnamefont {F.}~\bibnamefont {Sciarrino}},\ }\bibfield
  {title} {\bibinfo {title} {Non-linear boson sampling},\ }\href
  {https://doi.org/10.1038/s41534-023-00676-x} {\bibfield  {journal} {\bibinfo
  {journal} {npj Quantum Information}\ }\textbf {\bibinfo {volume} {9}},\
  \bibinfo {pages} {3} (\bibinfo {year} {2023})}\BibitemShut {NoStop}%
\bibitem [{\citenamefont {Hoch}\ \emph {et~al.}(2025)\citenamefont {Hoch},
  \citenamefont {Caruccio}, \citenamefont {Rodari}, \citenamefont
  {Francalanci}, \citenamefont {Suprano}, \citenamefont {Giordani},
  \citenamefont {Carvacho}, \citenamefont {Spagnolo}, \citenamefont {Koudia},
  \citenamefont {Proietti}, \citenamefont {Liorni}, \citenamefont {Cerocchi},
  \citenamefont {Albiero}, \citenamefont {Di~Giano}, \citenamefont {Gardina},
  \citenamefont {Ceccarelli}, \citenamefont {Corrielli}, \citenamefont
  {Chabaud}, \citenamefont {Osellame}, \citenamefont {Dispenza},\ and\
  \citenamefont {Sciarrino}}]{Hoch2025}%
  \BibitemOpen
  \bibfield  {author} {\bibinfo {author} {\bibfnamefont {F.}~\bibnamefont
  {Hoch}}, \bibinfo {author} {\bibfnamefont {E.}~\bibnamefont {Caruccio}},
  \bibinfo {author} {\bibfnamefont {G.}~\bibnamefont {Rodari}}, \bibinfo
  {author} {\bibfnamefont {T.}~\bibnamefont {Francalanci}}, \bibinfo {author}
  {\bibfnamefont {A.}~\bibnamefont {Suprano}}, \bibinfo {author} {\bibfnamefont
  {T.}~\bibnamefont {Giordani}}, \bibinfo {author} {\bibfnamefont
  {G.}~\bibnamefont {Carvacho}}, \bibinfo {author} {\bibfnamefont
  {N.}~\bibnamefont {Spagnolo}}, \bibinfo {author} {\bibfnamefont
  {S.}~\bibnamefont {Koudia}}, \bibinfo {author} {\bibfnamefont
  {M.}~\bibnamefont {Proietti}}, \bibinfo {author} {\bibfnamefont
  {C.}~\bibnamefont {Liorni}}, \bibinfo {author} {\bibfnamefont
  {F.}~\bibnamefont {Cerocchi}}, \bibinfo {author} {\bibfnamefont
  {R.}~\bibnamefont {Albiero}}, \bibinfo {author} {\bibfnamefont
  {N.}~\bibnamefont {Di~Giano}}, \bibinfo {author} {\bibfnamefont
  {M.}~\bibnamefont {Gardina}}, \bibinfo {author} {\bibfnamefont
  {F.}~\bibnamefont {Ceccarelli}}, \bibinfo {author} {\bibfnamefont
  {G.}~\bibnamefont {Corrielli}}, \bibinfo {author} {\bibfnamefont
  {U.}~\bibnamefont {Chabaud}}, \bibinfo {author} {\bibfnamefont
  {R.}~\bibnamefont {Osellame}}, \bibinfo {author} {\bibfnamefont
  {M.}~\bibnamefont {Dispenza}},\ and\ \bibinfo {author} {\bibfnamefont
  {F.}~\bibnamefont {Sciarrino}},\ }\bibfield  {title} {\bibinfo {title}
  {Quantum machine learning with adaptive boson sampling via post-selection},\
  }\href {https://doi.org/10.1038/s41467-025-55877-z} {\bibfield  {journal}
  {\bibinfo  {journal} {Nature Communications}\ }\textbf {\bibinfo {volume}
  {16}},\ \bibinfo {pages} {902} (\bibinfo {year} {2025})}\BibitemShut
  {NoStop}%
\bibitem [{\citenamefont {Monbroussou}\ \emph
  {et~al.}(2025{\natexlab{a}})\citenamefont {Monbroussou}, \citenamefont
  {Mamon}, \citenamefont {Thomas}, \citenamefont {Yacoub}, \citenamefont
  {Chabaud},\ and\ \citenamefont {Kashefi}}]{Monbroussou2025}%
  \BibitemOpen
  \bibfield  {author} {\bibinfo {author} {\bibfnamefont {L.}~\bibnamefont
  {Monbroussou}}, \bibinfo {author} {\bibfnamefont {E.~Z.}\ \bibnamefont
  {Mamon}}, \bibinfo {author} {\bibfnamefont {H.}~\bibnamefont {Thomas}},
  \bibinfo {author} {\bibfnamefont {V.}~\bibnamefont {Yacoub}}, \bibinfo
  {author} {\bibfnamefont {U.}~\bibnamefont {Chabaud}},\ and\ \bibinfo {author}
  {\bibfnamefont {E.}~\bibnamefont {Kashefi}},\ }\bibfield  {title} {\bibinfo
  {title} {Toward quantum advantage with photonic state injection},\ }\href
  {https://doi.org/10.1103/physrevresearch.7.033051} {\bibfield  {journal}
  {\bibinfo  {journal} {Physical Review Research}\ }\textbf {\bibinfo {volume}
  {7}},\ \bibinfo {pages} {033051} (\bibinfo {year}
  {2025}{\natexlab{a}})}\BibitemShut {NoStop}%
\bibitem [{\citenamefont {Monbroussou}\ \emph
  {et~al.}(2025{\natexlab{b}})\citenamefont {Monbroussou}, \citenamefont
  {Polacchi}, \citenamefont {Yacoub}, \citenamefont {Caruccio}, \citenamefont
  {Rodari}, \citenamefont {Hoch}, \citenamefont {Carvacho}, \citenamefont
  {Spagnolo}, \citenamefont {Giordani}, \citenamefont {Bossi}, \citenamefont
  {Rajan}, \citenamefont {Di~Giano}, \citenamefont {Albiero}, \citenamefont
  {Ceccarelli}, \citenamefont {Osellame}, \citenamefont {Kashefi},\ and\
  \citenamefont {Sciarrino}}]{Monbroussou2025B}%
  \BibitemOpen
  \bibfield  {author} {\bibinfo {author} {\bibfnamefont {L.}~\bibnamefont
  {Monbroussou}}, \bibinfo {author} {\bibfnamefont {B.}~\bibnamefont
  {Polacchi}}, \bibinfo {author} {\bibfnamefont {V.}~\bibnamefont {Yacoub}},
  \bibinfo {author} {\bibfnamefont {E.}~\bibnamefont {Caruccio}}, \bibinfo
  {author} {\bibfnamefont {G.}~\bibnamefont {Rodari}}, \bibinfo {author}
  {\bibfnamefont {F.}~\bibnamefont {Hoch}}, \bibinfo {author} {\bibfnamefont
  {G.}~\bibnamefont {Carvacho}}, \bibinfo {author} {\bibfnamefont
  {N.}~\bibnamefont {Spagnolo}}, \bibinfo {author} {\bibfnamefont
  {T.}~\bibnamefont {Giordani}}, \bibinfo {author} {\bibfnamefont
  {M.}~\bibnamefont {Bossi}}, \bibinfo {author} {\bibfnamefont
  {A.}~\bibnamefont {Rajan}}, \bibinfo {author} {\bibfnamefont
  {N.}~\bibnamefont {Di~Giano}}, \bibinfo {author} {\bibfnamefont
  {R.}~\bibnamefont {Albiero}}, \bibinfo {author} {\bibfnamefont
  {F.}~\bibnamefont {Ceccarelli}}, \bibinfo {author} {\bibfnamefont
  {R.}~\bibnamefont {Osellame}}, \bibinfo {author} {\bibfnamefont
  {E.}~\bibnamefont {Kashefi}},\ and\ \bibinfo {author} {\bibfnamefont
  {F.}~\bibnamefont {Sciarrino}},\ }\bibfield  {title} {\bibinfo {title}
  {Photonic quantum convolutional neural networks with adaptive state
  injection},\ }\href {https://doi.org/10.1117/1.ap.7.6.066012} {\bibfield
  {journal} {\bibinfo  {journal} {Advanced Photonics}\ }\textbf {\bibinfo
  {volume} {7}},\ \bibinfo {pages} {066012} (\bibinfo {year}
  {2025}{\natexlab{b}})}\BibitemShut {NoStop}%
\bibitem [{\citenamefont {Knill}(2002)}]{Knill2002}%
  \BibitemOpen
  \bibfield  {author} {\bibinfo {author} {\bibfnamefont {E.}~\bibnamefont
  {Knill}},\ }\bibfield  {title} {\bibinfo {title} {Quantum gates using linear
  optics and postselection},\ }\href
  {https://doi.org/10.1103/physreva.66.052306} {\bibfield  {journal} {\bibinfo
  {journal} {Physical Review A}\ }\textbf {\bibinfo {volume} {66}},\ \bibinfo
  {pages} {052306} (\bibinfo {year} {2002})}\BibitemShut {NoStop}%
\bibitem [{\citenamefont {VanMeter}\ \emph {et~al.}(2007)\citenamefont
  {VanMeter}, \citenamefont {Lougovski}, \citenamefont {Uskov}, \citenamefont
  {Kieling}, \citenamefont {Eisert},\ and\ \citenamefont
  {Dowling}}]{VanMeter2007}%
  \BibitemOpen
  \bibfield  {author} {\bibinfo {author} {\bibfnamefont {N.~M.}\ \bibnamefont
  {VanMeter}}, \bibinfo {author} {\bibfnamefont {P.}~\bibnamefont {Lougovski}},
  \bibinfo {author} {\bibfnamefont {D.~B.}\ \bibnamefont {Uskov}}, \bibinfo
  {author} {\bibfnamefont {K.}~\bibnamefont {Kieling}}, \bibinfo {author}
  {\bibfnamefont {J.}~\bibnamefont {Eisert}},\ and\ \bibinfo {author}
  {\bibfnamefont {J.~P.}\ \bibnamefont {Dowling}},\ }\bibfield  {title}
  {\bibinfo {title} {General linear-optical quantum state generation scheme:
  Applications to maximally path-entangled states},\ }\href
  {https://doi.org/10.1103/physreva.76.063808} {\bibfield  {journal} {\bibinfo
  {journal} {Physical Review A}\ }\textbf {\bibinfo {volume} {76}},\ \bibinfo
  {pages} {063808} (\bibinfo {year} {2007})}\BibitemShut {NoStop}%
\bibitem [{\citenamefont {Singh}\ \emph {et~al.}(2026)\citenamefont {Singh},
  \citenamefont {Marshman}, \citenamefont {Villegas-Aguilar}, \citenamefont
  {Eisert},\ and\ \citenamefont {Tischler}}]{Singh2026}%
  \BibitemOpen
  \bibfield  {author} {\bibinfo {author} {\bibfnamefont {D.}~\bibnamefont
  {Singh}}, \bibinfo {author} {\bibfnamefont {R.~J.}\ \bibnamefont {Marshman}},
  \bibinfo {author} {\bibfnamefont {L.}~\bibnamefont {Villegas-Aguilar}},
  \bibinfo {author} {\bibfnamefont {J.}~\bibnamefont {Eisert}},\ and\ \bibinfo
  {author} {\bibfnamefont {N.}~\bibnamefont {Tischler}},\ }\href
  {https://doi.org/10.48550/arXiv.2602.09495} {\bibinfo {title} {Rigorous no-go
  theorems for heralded linear-optical state generation tasks}} (\bibinfo
  {year} {2026}),\ \Eprint {https://arxiv.org/abs/2602.09495} {arXiv:2602.09495
  [quant-ph]} \BibitemShut {NoStop}%
\bibitem [{\citenamefont {Aniello}\ \emph
  {et~al.}(2006{\natexlab{a}})\citenamefont {Aniello}, \citenamefont {Lupo},
  \citenamefont {Napolitano},\ and\ \citenamefont {Paris}}]{Aniello2006}%
  \BibitemOpen
  \bibfield  {author} {\bibinfo {author} {\bibfnamefont {P.}~\bibnamefont
  {Aniello}}, \bibinfo {author} {\bibfnamefont {C.}~\bibnamefont {Lupo}},
  \bibinfo {author} {\bibfnamefont {M.}~\bibnamefont {Napolitano}},\ and\
  \bibinfo {author} {\bibfnamefont {M.~G.}\ \bibnamefont {Paris}},\ }\bibfield
  {title} {\bibinfo {title} {Engineering multiphoton states for linear optics
  computation},\ }\href {https://doi.org/10.1140/epjd/e2006-00259-y} {\bibfield
   {journal} {\bibinfo  {journal} {The European Physical Journal D}\ }\textbf
  {\bibinfo {volume} {41}},\ \bibinfo {pages} {579–587} (\bibinfo {year}
  {2006}{\natexlab{a}})}\BibitemShut {NoStop}%
\bibitem [{\citenamefont {de~Gliniasty}\ \emph {et~al.}(2024)\citenamefont
  {de~Gliniasty}, \citenamefont {Bagourd}, \citenamefont {Draux},\ and\
  \citenamefont {Bourdoncle}}]{deGliniasty2024}%
  \BibitemOpen
  \bibfield  {author} {\bibinfo {author} {\bibfnamefont {G.}~\bibnamefont
  {de~Gliniasty}}, \bibinfo {author} {\bibfnamefont {P.}~\bibnamefont
  {Bagourd}}, \bibinfo {author} {\bibfnamefont {S.}~\bibnamefont {Draux}},\
  and\ \bibinfo {author} {\bibfnamefont {B.}~\bibnamefont {Bourdoncle}},\
  }\href {https://doi.org/10.48550/arXiv.2405.01395} {\bibinfo {title} {Simple
  rules for two-photon state preparation with linear optics}} (\bibinfo {year}
  {2024}),\ \Eprint {https://arxiv.org/abs/2405.01395} {arXiv:2405.01395
  [quant-ph]} \BibitemShut {NoStop}%
\bibitem [{\citenamefont {Kopylov}\ \emph {et~al.}(2025)\citenamefont
  {Kopylov}, \citenamefont {Offen}, \citenamefont {Ares}, \citenamefont
  {Wembe}, \citenamefont {Ober-Bl\"{o}baum}, \citenamefont {Meier},
  \citenamefont {Sharapova},\ and\ \citenamefont {Sperling}}]{Kopylov2025}%
  \BibitemOpen
  \bibfield  {author} {\bibinfo {author} {\bibfnamefont {D.~A.}\ \bibnamefont
  {Kopylov}}, \bibinfo {author} {\bibfnamefont {C.}~\bibnamefont {Offen}},
  \bibinfo {author} {\bibfnamefont {L.}~\bibnamefont {Ares}}, \bibinfo {author}
  {\bibfnamefont {B.}~\bibnamefont {Wembe}}, \bibinfo {author} {\bibfnamefont
  {S.}~\bibnamefont {Ober-Bl\"{o}baum}}, \bibinfo {author} {\bibfnamefont
  {T.}~\bibnamefont {Meier}}, \bibinfo {author} {\bibfnamefont {P.~R.}\
  \bibnamefont {Sharapova}},\ and\ \bibinfo {author} {\bibfnamefont
  {J.}~\bibnamefont {Sperling}},\ }\bibfield  {title} {\bibinfo {title}
  {Multiphoton, multimode state classification for nonlinear optical
  circuits},\ }\href {https://doi.org/10.1103/sv6z-v1gk} {\bibfield  {journal}
  {\bibinfo  {journal} {Physical Review Research}\ }\textbf {\bibinfo {volume}
  {7}},\ \bibinfo {pages} {033062} (\bibinfo {year} {2025})}\BibitemShut
  {NoStop}%
\bibitem [{\citenamefont {Aralov}\ \emph {et~al.}(2026)\citenamefont {Aralov},
  \citenamefont {Gillet}, \citenamefont {Nguyen}, \citenamefont {Cosentino},
  \citenamefont {Walschaers},\ and\ \citenamefont {Frigerio}}]{Aralov2026}%
  \BibitemOpen
  \bibfield  {author} {\bibinfo {author} {\bibfnamefont {A.}~\bibnamefont
  {Aralov}}, \bibinfo {author} {\bibfnamefont {E.}~\bibnamefont {Gillet}},
  \bibinfo {author} {\bibfnamefont {V.}~\bibnamefont {Nguyen}}, \bibinfo
  {author} {\bibfnamefont {A.}~\bibnamefont {Cosentino}}, \bibinfo {author}
  {\bibfnamefont {M.}~\bibnamefont {Walschaers}},\ and\ \bibinfo {author}
  {\bibfnamefont {M.}~\bibnamefont {Frigerio}},\ }\bibfield  {title} {\bibinfo
  {title} {Photon catalysis for general multimode multi-photon quantum state
  preparation},\ }\href {https://doi.org/10.1103/ktc9-9rjb} {\bibfield
  {journal} {\bibinfo  {journal} {PRX Quantum}\ }\textbf {\bibinfo {volume}
  {7}},\ \bibinfo {pages} {020323} (\bibinfo {year} {2026})}\BibitemShut
  {NoStop}%
\bibitem [{\citenamefont {Migdał}\ \emph {et~al.}(2014)\citenamefont
  {Migdał}, \citenamefont {Rodríguez-Laguna}, \citenamefont {Oszmaniec},\
  and\ \citenamefont {Lewenstein}}]{Migdal2014}%
  \BibitemOpen
  \bibfield  {author} {\bibinfo {author} {\bibfnamefont {P.}~\bibnamefont
  {Migdał}}, \bibinfo {author} {\bibfnamefont {J.}~\bibnamefont
  {Rodríguez-Laguna}}, \bibinfo {author} {\bibfnamefont {M.}~\bibnamefont
  {Oszmaniec}},\ and\ \bibinfo {author} {\bibfnamefont {M.}~\bibnamefont
  {Lewenstein}},\ }\bibfield  {title} {\bibinfo {title} {Multiphoton states
  related via linear optics},\ }\href
  {https://doi.org/10.1103/physreva.89.062329} {\bibfield  {journal} {\bibinfo
  {journal} {Physical Review A}\ }\textbf {\bibinfo {volume} {89}},\ \bibinfo
  {pages} {062329} (\bibinfo {year} {2014})}\BibitemShut {NoStop}%
\bibitem [{\citenamefont {Garcia-Escartin}\ \emph {et~al.}(2019)\citenamefont
  {Garcia-Escartin}, \citenamefont {Gimeno},\ and\ \citenamefont
  {Moyano-Fernández}}]{Escartin2019}%
  \BibitemOpen
  \bibfield  {author} {\bibinfo {author} {\bibfnamefont {J.~C.}\ \bibnamefont
  {Garcia-Escartin}}, \bibinfo {author} {\bibfnamefont {V.}~\bibnamefont
  {Gimeno}},\ and\ \bibinfo {author} {\bibfnamefont {J.~J.}\ \bibnamefont
  {Moyano-Fernández}},\ }\bibfield  {title} {\bibinfo {title} {Method to
  determine which quantum operations can be realized with linear optics with a
  constructive implementation recipe},\ }\href
  {https://doi.org/10.1103/physreva.100.022301} {\bibfield  {journal} {\bibinfo
   {journal} {Physical Review A}\ }\textbf {\bibinfo {volume} {100}},\ \bibinfo
  {pages} {022301} (\bibinfo {year} {2019})}\BibitemShut {NoStop}%
\bibitem [{\citenamefont {Parellada}\ \emph {et~al.}(2024)\citenamefont
  {Parellada}, \citenamefont {Garcia}, \citenamefont {Moyano-Fernández},\ and\
  \citenamefont {Garcia-Escartin}}]{Parellada2024}%
  \BibitemOpen
  \bibfield  {author} {\bibinfo {author} {\bibfnamefont {P.~V.}\ \bibnamefont
  {Parellada}}, \bibinfo {author} {\bibfnamefont {V.~G.~i.}\ \bibnamefont
  {Garcia}}, \bibinfo {author} {\bibfnamefont {J.~J.}\ \bibnamefont
  {Moyano-Fernández}},\ and\ \bibinfo {author} {\bibfnamefont {J.~C.}\
  \bibnamefont {Garcia-Escartin}},\ }\href
  {https://doi.org/10.48550/arXiv.2409.12223} {\bibinfo {title} {Lie algebraic
  invariants in quantum linear optics}} (\bibinfo {year} {2024}),\ \Eprint
  {https://arxiv.org/abs/2409.12223} {arXiv:2409.12223 [quant-ph]} \BibitemShut
  {NoStop}%
\bibitem [{\citenamefont {Mamon}(2025)}]{Mamon2025}%
  \BibitemOpen
  \bibfield  {author} {\bibinfo {author} {\bibfnamefont {E.~Z.}\ \bibnamefont
  {Mamon}},\ }\href {https://doi.org/10.48550/arXiv.2506.07995} {\bibinfo
  {title} {Orbit dimensions in linear and gaussian quantum optics}} (\bibinfo
  {year} {2025}),\ \Eprint {https://arxiv.org/abs/2506.07995} {arXiv:2506.07995
  [quant-ph]} \BibitemShut {NoStop}%
\bibitem [{\citenamefont {Draux}\ \emph {et~al.}(2025)\citenamefont {Draux},
  \citenamefont {Perdrix}, \citenamefont {Jeandel},\ and\ \citenamefont
  {Mansfield}}]{Draux2025}%
  \BibitemOpen
  \bibfield  {author} {\bibinfo {author} {\bibfnamefont {S.}~\bibnamefont
  {Draux}}, \bibinfo {author} {\bibfnamefont {S.}~\bibnamefont {Perdrix}},
  \bibinfo {author} {\bibfnamefont {E.}~\bibnamefont {Jeandel}},\ and\ \bibinfo
  {author} {\bibfnamefont {S.}~\bibnamefont {Mansfield}},\ }\href
  {https://doi.org/10.48550/arXiv.2509.02211} {\bibinfo {title} {Invariants in
  linear optics}} (\bibinfo {year} {2025}),\ \Eprint
  {https://arxiv.org/abs/2509.02211} {arXiv:2509.02211 [quant-ph]} \BibitemShut
  {NoStop}%
\bibitem [{\citenamefont {Rodari}\ \emph {et~al.}(2025)\citenamefont {Rodari},
  \citenamefont {Francalanci}, \citenamefont {Caruccio}, \citenamefont {Hoch},
  \citenamefont {Carvacho}, \citenamefont {Giordani}, \citenamefont {Spagnolo},
  \citenamefont {Albiero}, \citenamefont {Di~Giano}, \citenamefont
  {Ceccarelli}, \citenamefont {Corrielli}, \citenamefont {Crespi},
  \citenamefont {Osellame}, \citenamefont {Chabaud},\ and\ \citenamefont
  {Sciarrino}}]{Rodari2025}%
  \BibitemOpen
  \bibfield  {author} {\bibinfo {author} {\bibfnamefont {G.}~\bibnamefont
  {Rodari}}, \bibinfo {author} {\bibfnamefont {T.}~\bibnamefont {Francalanci}},
  \bibinfo {author} {\bibfnamefont {E.}~\bibnamefont {Caruccio}}, \bibinfo
  {author} {\bibfnamefont {F.}~\bibnamefont {Hoch}}, \bibinfo {author}
  {\bibfnamefont {G.}~\bibnamefont {Carvacho}}, \bibinfo {author}
  {\bibfnamefont {T.}~\bibnamefont {Giordani}}, \bibinfo {author}
  {\bibfnamefont {N.}~\bibnamefont {Spagnolo}}, \bibinfo {author}
  {\bibfnamefont {R.}~\bibnamefont {Albiero}}, \bibinfo {author} {\bibfnamefont
  {N.}~\bibnamefont {Di~Giano}}, \bibinfo {author} {\bibfnamefont
  {F.}~\bibnamefont {Ceccarelli}}, \bibinfo {author} {\bibfnamefont
  {G.}~\bibnamefont {Corrielli}}, \bibinfo {author} {\bibfnamefont
  {A.}~\bibnamefont {Crespi}}, \bibinfo {author} {\bibfnamefont
  {R.}~\bibnamefont {Osellame}}, \bibinfo {author} {\bibfnamefont
  {U.}~\bibnamefont {Chabaud}},\ and\ \bibinfo {author} {\bibfnamefont
  {F.}~\bibnamefont {Sciarrino}},\ }\bibfield  {title} {\bibinfo {title}
  {Observation of lie algebraic invariants in quantum linear optics},\ }\href
  {https://doi.org/10.1103/7961-hg2q} {\bibfield  {journal} {\bibinfo
  {journal} {Physical Review Research}\ }\textbf {\bibinfo {volume} {7}},\
  \bibinfo {pages} {043325} (\bibinfo {year} {2025})}\BibitemShut {NoStop}%
\bibitem [{\citenamefont {Forbes}\ \emph {et~al.}(2025)\citenamefont {Forbes},
  \citenamefont {Ghafari}, \citenamefont {Deacon}, \citenamefont {Singh},
  \citenamefont {Lavie}, \citenamefont {Yard}, \citenamefont {Shaw},
  \citenamefont {Laing},\ and\ \citenamefont {Tischler}}]{Forbes2025}%
  \BibitemOpen
  \bibfield  {author} {\bibinfo {author} {\bibfnamefont {I.}~\bibnamefont
  {Forbes}}, \bibinfo {author} {\bibfnamefont {F.}~\bibnamefont {Ghafari}},
  \bibinfo {author} {\bibfnamefont {E.~C.~R.}\ \bibnamefont {Deacon}}, \bibinfo
  {author} {\bibfnamefont {S.~P.}\ \bibnamefont {Singh}}, \bibinfo {author}
  {\bibfnamefont {E.}~\bibnamefont {Lavie}}, \bibinfo {author} {\bibfnamefont
  {P.}~\bibnamefont {Yard}}, \bibinfo {author} {\bibfnamefont {R.~D.}\
  \bibnamefont {Shaw}}, \bibinfo {author} {\bibfnamefont {A.}~\bibnamefont
  {Laing}},\ and\ \bibinfo {author} {\bibfnamefont {N.}~\bibnamefont
  {Tischler}},\ }\bibfield  {title} {\bibinfo {title} {Heralded generation of
  entanglement with photons},\ }\href
  {https://doi.org/10.1088/1361-6633/adf85e} {\bibfield  {journal} {\bibinfo
  {journal} {Reports on Progress in Physics}\ }\textbf {\bibinfo {volume}
  {88}},\ \bibinfo {pages} {086002} (\bibinfo {year} {2025})}\BibitemShut
  {NoStop}%
\bibitem [{\citenamefont {Stanisic}\ \emph {et~al.}(2017)\citenamefont
  {Stanisic}, \citenamefont {Linden}, \citenamefont {Montanaro},\ and\
  \citenamefont {Turner}}]{Stanisic2017}%
  \BibitemOpen
  \bibfield  {author} {\bibinfo {author} {\bibfnamefont {S.}~\bibnamefont
  {Stanisic}}, \bibinfo {author} {\bibfnamefont {N.}~\bibnamefont {Linden}},
  \bibinfo {author} {\bibfnamefont {A.}~\bibnamefont {Montanaro}},\ and\
  \bibinfo {author} {\bibfnamefont {P.~S.}\ \bibnamefont {Turner}},\ }\bibfield
   {title} {\bibinfo {title} {Generating entanglement with linear optics},\
  }\href {https://doi.org/10.1103/physreva.96.043861} {\bibfield  {journal}
  {\bibinfo  {journal} {Physical Review A}\ }\textbf {\bibinfo {volume} {96}},\
  \bibinfo {pages} {043861} (\bibinfo {year} {2017})}\BibitemShut {NoStop}%
\bibitem [{\citenamefont {Gubarev}\ \emph {et~al.}(2020)\citenamefont
  {Gubarev}, \citenamefont {Dyakonov}, \citenamefont {Saygin}, \citenamefont
  {Struchalin}, \citenamefont {Straupe},\ and\ \citenamefont
  {Kulik}}]{Gubarev2020}%
  \BibitemOpen
  \bibfield  {author} {\bibinfo {author} {\bibfnamefont {F.~V.}\ \bibnamefont
  {Gubarev}}, \bibinfo {author} {\bibfnamefont {I.~V.}\ \bibnamefont
  {Dyakonov}}, \bibinfo {author} {\bibfnamefont {M.~Y.}\ \bibnamefont
  {Saygin}}, \bibinfo {author} {\bibfnamefont {G.~I.}\ \bibnamefont
  {Struchalin}}, \bibinfo {author} {\bibfnamefont {S.~S.}\ \bibnamefont
  {Straupe}},\ and\ \bibinfo {author} {\bibfnamefont {S.~P.}\ \bibnamefont
  {Kulik}},\ }\bibfield  {title} {\bibinfo {title} {Improved heralded schemes
  to generate entangled states from single photons},\ }\href
  {https://doi.org/10.1103/physreva.102.012604} {\bibfield  {journal} {\bibinfo
   {journal} {Physical Review A}\ }\textbf {\bibinfo {volume} {102}},\ \bibinfo
  {pages} {012604} (\bibinfo {year} {2020})}\BibitemShut {NoStop}%
\bibitem [{\citenamefont {Fldzhyan}\ \emph {et~al.}(2021)\citenamefont
  {Fldzhyan}, \citenamefont {Saygin},\ and\ \citenamefont
  {Kulik}}]{Fldzhyan2021}%
  \BibitemOpen
  \bibfield  {author} {\bibinfo {author} {\bibfnamefont {S.~A.}\ \bibnamefont
  {Fldzhyan}}, \bibinfo {author} {\bibfnamefont {M.~Y.}\ \bibnamefont
  {Saygin}},\ and\ \bibinfo {author} {\bibfnamefont {S.~P.}\ \bibnamefont
  {Kulik}},\ }\bibfield  {title} {\bibinfo {title} {Compact linear optical
  scheme for bell state generation},\ }\href
  {https://doi.org/10.1103/physrevresearch.3.043031} {\bibfield  {journal}
  {\bibinfo  {journal} {Physical Review Research}\ }\textbf {\bibinfo {volume}
  {3}},\ \bibinfo {pages} {043031} (\bibinfo {year} {2021})}\BibitemShut
  {NoStop}%
\bibitem [{\citenamefont {Fldzhyan}\ \emph {et~al.}(2023)\citenamefont
  {Fldzhyan}, \citenamefont {Saygin},\ and\ \citenamefont
  {Kulik}}]{Fldzhyan2023}%
  \BibitemOpen
  \bibfield  {author} {\bibinfo {author} {\bibfnamefont {S.~A.}\ \bibnamefont
  {Fldzhyan}}, \bibinfo {author} {\bibfnamefont {M.~Y.}\ \bibnamefont
  {Saygin}},\ and\ \bibinfo {author} {\bibfnamefont {S.~P.}\ \bibnamefont
  {Kulik}},\ }\bibfield  {title} {\bibinfo {title} {Programmable heralded
  linear optical generation of two-qubit states},\ }\href
  {https://doi.org/10.1103/physrevapplied.20.054030} {\bibfield  {journal}
  {\bibinfo  {journal} {Physical Review Applied}\ }\textbf {\bibinfo {volume}
  {20}},\ \bibinfo {pages} {054030} (\bibinfo {year} {2023})}\BibitemShut
  {NoStop}%
\bibitem [{\citenamefont {Paesani}\ \emph {et~al.}(2021)\citenamefont
  {Paesani}, \citenamefont {Bulmer}, \citenamefont {Jones}, \citenamefont
  {Santagati},\ and\ \citenamefont {Laing}}]{Paesani2021}%
  \BibitemOpen
  \bibfield  {author} {\bibinfo {author} {\bibfnamefont {S.}~\bibnamefont
  {Paesani}}, \bibinfo {author} {\bibfnamefont {J.~F.}\ \bibnamefont {Bulmer}},
  \bibinfo {author} {\bibfnamefont {A.~E.}\ \bibnamefont {Jones}}, \bibinfo
  {author} {\bibfnamefont {R.}~\bibnamefont {Santagati}},\ and\ \bibinfo
  {author} {\bibfnamefont {A.}~\bibnamefont {Laing}},\ }\bibfield  {title}
  {\bibinfo {title} {Scheme for universal high-dimensional quantum computation
  with linear optics},\ }\href {https://doi.org/10.1103/physrevlett.126.230504}
  {\bibfield  {journal} {\bibinfo  {journal} {Physical Review Letters}\
  }\textbf {\bibinfo {volume} {126}},\ \bibinfo {pages} {230504} (\bibinfo
  {year} {2021})}\BibitemShut {NoStop}%
\bibitem [{\citenamefont {Chin}\ \emph {et~al.}(2024)\citenamefont {Chin},
  \citenamefont {Ryu},\ and\ \citenamefont {Kim}}]{Chin2024}%
  \BibitemOpen
  \bibfield  {author} {\bibinfo {author} {\bibfnamefont {S.}~\bibnamefont
  {Chin}}, \bibinfo {author} {\bibfnamefont {J.}~\bibnamefont {Ryu}},\ and\
  \bibinfo {author} {\bibfnamefont {Y.-S.}\ \bibnamefont {Kim}},\ }\bibfield
  {title} {\bibinfo {title} {Exponentially enhanced scheme for the heralded
  qudit greenberger-horne-zeilinger state in linear optics},\ }\href
  {https://doi.org/10.1103/physrevlett.133.253601} {\bibfield  {journal}
  {\bibinfo  {journal} {Physical Review Letters}\ }\textbf {\bibinfo {volume}
  {133}},\ \bibinfo {pages} {253601} (\bibinfo {year} {2024})}\BibitemShut
  {NoStop}%
\bibitem [{\citenamefont {Kang}\ \emph {et~al.}(2025)\citenamefont {Kang},
  \citenamefont {Kim}, \citenamefont {Munro}, \citenamefont {Chin},\ and\
  \citenamefont {Huh}}]{Kang2026}%
  \BibitemOpen
  \bibfield  {author} {\bibinfo {author} {\bibfnamefont {M.}~\bibnamefont
  {Kang}}, \bibinfo {author} {\bibfnamefont {J.}~\bibnamefont {Kim}}, \bibinfo
  {author} {\bibfnamefont {W.~J.}\ \bibnamefont {Munro}}, \bibinfo {author}
  {\bibfnamefont {S.}~\bibnamefont {Chin}},\ and\ \bibinfo {author}
  {\bibfnamefont {J.}~\bibnamefont {Huh}},\ }\href
  {https://doi.org/10.48550/arXiv.2512.20881} {\bibinfo {title} {Heralded
  linear optical generation of dicke states}} (\bibinfo {year} {2025}),\
  \Eprint {https://arxiv.org/abs/2512.20881} {arXiv:2512.20881 [quant-ph]}
  \BibitemShut {NoStop}%
\bibitem [{\citenamefont {Bouwmeester}\ \emph {et~al.}(1999)\citenamefont
  {Bouwmeester}, \citenamefont {Pan}, \citenamefont {Daniell}, \citenamefont
  {Weinfurter},\ and\ \citenamefont {Zeilinger}}]{Bouwmeester1999}%
  \BibitemOpen
  \bibfield  {author} {\bibinfo {author} {\bibfnamefont {D.}~\bibnamefont
  {Bouwmeester}}, \bibinfo {author} {\bibfnamefont {J.-W.}\ \bibnamefont
  {Pan}}, \bibinfo {author} {\bibfnamefont {M.}~\bibnamefont {Daniell}},
  \bibinfo {author} {\bibfnamefont {H.}~\bibnamefont {Weinfurter}},\ and\
  \bibinfo {author} {\bibfnamefont {A.}~\bibnamefont {Zeilinger}},\ }\bibfield
  {title} {\bibinfo {title} {Observation of three-photon
  greenberger-horne-zeilinger entanglement},\ }\href
  {https://doi.org/10.1103/PhysRevLett.82.1345} {\bibfield  {journal} {\bibinfo
   {journal} {Physical Review Letters}\ }\textbf {\bibinfo {volume} {82}},\
  \bibinfo {pages} {1345} (\bibinfo {year} {1999})}\BibitemShut {NoStop}%
\bibitem [{\citenamefont {Pan}\ \emph {et~al.}(2001)\citenamefont {Pan},
  \citenamefont {Daniell}, \citenamefont {Gasparoni}, \citenamefont {Weihs},\
  and\ \citenamefont {Zeilinger}}]{Pan2001}%
  \BibitemOpen
  \bibfield  {author} {\bibinfo {author} {\bibfnamefont {J.-W.}\ \bibnamefont
  {Pan}}, \bibinfo {author} {\bibfnamefont {M.}~\bibnamefont {Daniell}},
  \bibinfo {author} {\bibfnamefont {S.}~\bibnamefont {Gasparoni}}, \bibinfo
  {author} {\bibfnamefont {G.}~\bibnamefont {Weihs}},\ and\ \bibinfo {author}
  {\bibfnamefont {A.}~\bibnamefont {Zeilinger}},\ }\bibfield  {title} {\bibinfo
  {title} {Experimental demonstration of four-photon entanglement and
  high-fidelity teleportation},\ }\href
  {https://doi.org/10.1103/PhysRevLett.86.4435} {\bibfield  {journal} {\bibinfo
   {journal} {Physical Review Letters}\ }\textbf {\bibinfo {volume} {86}},\
  \bibinfo {pages} {4435} (\bibinfo {year} {2001})}\BibitemShut {NoStop}%
\bibitem [{\citenamefont {Lu}\ \emph {et~al.}(2007)\citenamefont {Lu},
  \citenamefont {Zhou}, \citenamefont {Gühne}, \citenamefont {Gao},
  \citenamefont {Zhang}, \citenamefont {Yuan}, \citenamefont {Goebel},
  \citenamefont {Yang},\ and\ \citenamefont {Pan}}]{Lu2007}%
  \BibitemOpen
  \bibfield  {author} {\bibinfo {author} {\bibfnamefont {C.-Y.}\ \bibnamefont
  {Lu}}, \bibinfo {author} {\bibfnamefont {X.-Q.}\ \bibnamefont {Zhou}},
  \bibinfo {author} {\bibfnamefont {O.}~\bibnamefont {Gühne}}, \bibinfo
  {author} {\bibfnamefont {W.-B.~B.}\ \bibnamefont {Gao}}, \bibinfo {author}
  {\bibfnamefont {J.}~\bibnamefont {Zhang}}, \bibinfo {author} {\bibfnamefont
  {Z.-S.}\ \bibnamefont {Yuan}}, \bibinfo {author} {\bibfnamefont
  {A.}~\bibnamefont {Goebel}}, \bibinfo {author} {\bibfnamefont
  {T.}~\bibnamefont {Yang}},\ and\ \bibinfo {author} {\bibfnamefont {J.-W.}\
  \bibnamefont {Pan}},\ }\bibfield  {title} {\bibinfo {title} {Experimental
  entanglement of six photons in graph states},\ }\href
  {https://doi.org/10.1038/nphys507} {\bibfield  {journal} {\bibinfo  {journal}
  {Nature Physics}\ }\textbf {\bibinfo {volume} {3}},\ \bibinfo {pages} {91}
  (\bibinfo {year} {2007})}\BibitemShut {NoStop}%
\bibitem [{\citenamefont {Pont}\ \emph {et~al.}(2024)\citenamefont {Pont},
  \citenamefont {Corrielli}, \citenamefont {Fyrillas}, \citenamefont {Agresti},
  \citenamefont {Carvacho}, \citenamefont {Maring}, \citenamefont {Emeriau},
  \citenamefont {Ceccarelli}, \citenamefont {Albiero}, \citenamefont
  {Dias~Ferreira}, \citenamefont {Somaschi}, \citenamefont {Senellart},
  \citenamefont {Sagnes}, \citenamefont {Morassi}, \citenamefont {Lemaître},
  \citenamefont {Senellart}, \citenamefont {Sciarrino}, \citenamefont
  {Liscidini}, \citenamefont {Belabas},\ and\ \citenamefont
  {Osellame}}]{Pont2024}%
  \BibitemOpen
  \bibfield  {author} {\bibinfo {author} {\bibfnamefont {G.}~\bibnamefont
  {Pont}}, \bibinfo {author} {\bibfnamefont {G.}~\bibnamefont {Corrielli}},
  \bibinfo {author} {\bibfnamefont {A.}~\bibnamefont {Fyrillas}}, \bibinfo
  {author} {\bibfnamefont {I.}~\bibnamefont {Agresti}}, \bibinfo {author}
  {\bibfnamefont {G.}~\bibnamefont {Carvacho}}, \bibinfo {author}
  {\bibfnamefont {N.}~\bibnamefont {Maring}}, \bibinfo {author} {\bibfnamefont
  {P.-E.}\ \bibnamefont {Emeriau}}, \bibinfo {author} {\bibfnamefont
  {F.}~\bibnamefont {Ceccarelli}}, \bibinfo {author} {\bibfnamefont
  {R.}~\bibnamefont {Albiero}}, \bibinfo {author} {\bibfnamefont {P.~H.}\
  \bibnamefont {Dias~Ferreira}}, \bibinfo {author} {\bibfnamefont
  {N.}~\bibnamefont {Somaschi}}, \bibinfo {author} {\bibfnamefont
  {J.}~\bibnamefont {Senellart}}, \bibinfo {author} {\bibfnamefont
  {I.}~\bibnamefont {Sagnes}}, \bibinfo {author} {\bibfnamefont
  {M.}~\bibnamefont {Morassi}}, \bibinfo {author} {\bibfnamefont
  {A.}~\bibnamefont {Lemaître}}, \bibinfo {author} {\bibfnamefont
  {P.}~\bibnamefont {Senellart}}, \bibinfo {author} {\bibfnamefont
  {F.}~\bibnamefont {Sciarrino}}, \bibinfo {author} {\bibfnamefont
  {M.}~\bibnamefont {Liscidini}}, \bibinfo {author} {\bibfnamefont
  {N.}~\bibnamefont {Belabas}},\ and\ \bibinfo {author} {\bibfnamefont
  {R.}~\bibnamefont {Osellame}},\ }\bibfield  {title} {\bibinfo {title}
  {High-fidelity four-photon ghz states on chip},\ }\href
  {https://doi.org/10.1038/s41534-024-00830-z} {\bibfield  {journal} {\bibinfo
  {journal} {npj Quantum Information}\ }\textbf {\bibinfo {volume} {10}},\
  \bibinfo {pages} {40} (\bibinfo {year} {2024})}\BibitemShut {NoStop}%
\bibitem [{\citenamefont {Carolan}\ \emph {et~al.}(2020)\citenamefont
  {Carolan}, \citenamefont {Mohseni}, \citenamefont {Olson}, \citenamefont
  {Prabhu}, \citenamefont {Chen}, \citenamefont {Bunandar}, \citenamefont
  {Niu}, \citenamefont {Harris}, \citenamefont {Wong}, \citenamefont
  {Hochberg}, \citenamefont {Lloyd},\ and\ \citenamefont
  {Englund}}]{Carolan2020}%
  \BibitemOpen
  \bibfield  {author} {\bibinfo {author} {\bibfnamefont {J.}~\bibnamefont
  {Carolan}}, \bibinfo {author} {\bibfnamefont {M.}~\bibnamefont {Mohseni}},
  \bibinfo {author} {\bibfnamefont {J.~P.}\ \bibnamefont {Olson}}, \bibinfo
  {author} {\bibfnamefont {M.}~\bibnamefont {Prabhu}}, \bibinfo {author}
  {\bibfnamefont {C.}~\bibnamefont {Chen}}, \bibinfo {author} {\bibfnamefont
  {D.}~\bibnamefont {Bunandar}}, \bibinfo {author} {\bibfnamefont {M.~Y.}\
  \bibnamefont {Niu}}, \bibinfo {author} {\bibfnamefont {N.~C.}\ \bibnamefont
  {Harris}}, \bibinfo {author} {\bibfnamefont {F.~N.~C.}\ \bibnamefont {Wong}},
  \bibinfo {author} {\bibfnamefont {M.}~\bibnamefont {Hochberg}}, \bibinfo
  {author} {\bibfnamefont {S.}~\bibnamefont {Lloyd}},\ and\ \bibinfo {author}
  {\bibfnamefont {D.}~\bibnamefont {Englund}},\ }\bibfield  {title} {\bibinfo
  {title} {Variational quantum unsampling on a quantum photonic processor},\
  }\href {https://doi.org/10.1038/s41567-019-0747-6} {\bibfield  {journal}
  {\bibinfo  {journal} {Nature Physics}\ }\textbf {\bibinfo {volume} {16}},\
  \bibinfo {pages} {322–327} (\bibinfo {year} {2020})}\BibitemShut {NoStop}%
\bibitem [{\citenamefont {Gan}\ \emph {et~al.}(2022)\citenamefont {Gan},
  \citenamefont {Leykam},\ and\ \citenamefont {Angelakis}}]{Gan2022}%
  \BibitemOpen
  \bibfield  {author} {\bibinfo {author} {\bibfnamefont {B.~Y.}\ \bibnamefont
  {Gan}}, \bibinfo {author} {\bibfnamefont {D.}~\bibnamefont {Leykam}},\ and\
  \bibinfo {author} {\bibfnamefont {D.~G.}\ \bibnamefont {Angelakis}},\
  }\bibfield  {title} {\bibinfo {title} {Fock state-enhanced expressivity of
  quantum machine learning models},\ }\href
  {https://doi.org/10.1140/epjqt/s40507-022-00135-0} {\bibfield  {journal}
  {\bibinfo  {journal} {EPJ Quantum Technology}\ }\textbf {\bibinfo {volume}
  {9}},\  (\bibinfo {year} {2022})}\BibitemShut {NoStop}%
\bibitem [{\citenamefont {Maring}\ \emph {et~al.}(2024)\citenamefont {Maring},
  \citenamefont {Fyrillas}, \citenamefont {Pont}, \citenamefont {Ivanov},
  \citenamefont {Stepanov}, \citenamefont {Margaria}, \citenamefont {Hease},
  \citenamefont {Pishchagin}, \citenamefont {Lemaître}, \citenamefont
  {Sagnes}, \citenamefont {Au}, \citenamefont {Boissier}, \citenamefont
  {Bertasi}, \citenamefont {Baert}, \citenamefont {Valdivia}, \citenamefont
  {Billard}, \citenamefont {Acar}, \citenamefont {Brieussel}, \citenamefont
  {Mezher}, \citenamefont {Wein}, \citenamefont {Salavrakos}, \citenamefont
  {Sinnott}, \citenamefont {Fioretto}, \citenamefont {Emeriau}, \citenamefont
  {Belabas}, \citenamefont {Mansfield}, \citenamefont {Senellart},
  \citenamefont {Senellart},\ and\ \citenamefont {Somaschi}}]{Maring2024}%
  \BibitemOpen
  \bibfield  {author} {\bibinfo {author} {\bibfnamefont {N.}~\bibnamefont
  {Maring}}, \bibinfo {author} {\bibfnamefont {A.}~\bibnamefont {Fyrillas}},
  \bibinfo {author} {\bibfnamefont {M.}~\bibnamefont {Pont}}, \bibinfo {author}
  {\bibfnamefont {E.}~\bibnamefont {Ivanov}}, \bibinfo {author} {\bibfnamefont
  {P.}~\bibnamefont {Stepanov}}, \bibinfo {author} {\bibfnamefont
  {N.}~\bibnamefont {Margaria}}, \bibinfo {author} {\bibfnamefont
  {W.}~\bibnamefont {Hease}}, \bibinfo {author} {\bibfnamefont
  {A.}~\bibnamefont {Pishchagin}}, \bibinfo {author} {\bibfnamefont
  {A.}~\bibnamefont {Lemaître}}, \bibinfo {author} {\bibfnamefont
  {I.}~\bibnamefont {Sagnes}}, \bibinfo {author} {\bibfnamefont {T.~H.}\
  \bibnamefont {Au}}, \bibinfo {author} {\bibfnamefont {S.}~\bibnamefont
  {Boissier}}, \bibinfo {author} {\bibfnamefont {E.}~\bibnamefont {Bertasi}},
  \bibinfo {author} {\bibfnamefont {A.}~\bibnamefont {Baert}}, \bibinfo
  {author} {\bibfnamefont {M.}~\bibnamefont {Valdivia}}, \bibinfo {author}
  {\bibfnamefont {M.}~\bibnamefont {Billard}}, \bibinfo {author} {\bibfnamefont
  {O.}~\bibnamefont {Acar}}, \bibinfo {author} {\bibfnamefont {A.}~\bibnamefont
  {Brieussel}}, \bibinfo {author} {\bibfnamefont {R.}~\bibnamefont {Mezher}},
  \bibinfo {author} {\bibfnamefont {S.~C.}\ \bibnamefont {Wein}}, \bibinfo
  {author} {\bibfnamefont {A.}~\bibnamefont {Salavrakos}}, \bibinfo {author}
  {\bibfnamefont {P.}~\bibnamefont {Sinnott}}, \bibinfo {author} {\bibfnamefont
  {D.~A.}\ \bibnamefont {Fioretto}}, \bibinfo {author} {\bibfnamefont {P.-E.}\
  \bibnamefont {Emeriau}}, \bibinfo {author} {\bibfnamefont {N.}~\bibnamefont
  {Belabas}}, \bibinfo {author} {\bibfnamefont {S.}~\bibnamefont {Mansfield}},
  \bibinfo {author} {\bibfnamefont {P.}~\bibnamefont {Senellart}}, \bibinfo
  {author} {\bibfnamefont {J.}~\bibnamefont {Senellart}},\ and\ \bibinfo
  {author} {\bibfnamefont {N.}~\bibnamefont {Somaschi}},\ }\bibfield  {title}
  {\bibinfo {title} {A versatile single-photon-based quantum computing
  platform},\ }\href {https://doi.org/10.1038/s41566-024-01403-4} {\bibfield
  {journal} {\bibinfo  {journal} {Nature Photonics}\ }\textbf {\bibinfo
  {volume} {18}},\ \bibinfo {pages} {603–609} (\bibinfo {year}
  {2024})}\BibitemShut {NoStop}%
\bibitem [{\citenamefont {Wang}\ \emph {et~al.}(2025)\citenamefont {Wang},
  \citenamefont {Yin}, \citenamefont {Haug}, \citenamefont {Pentangelo},
  \citenamefont {Piacentini}, \citenamefont {Crespi}, \citenamefont
  {Ceccarelli}, \citenamefont {Osellame},\ and\ \citenamefont
  {Walther}}]{Wang2025}%
  \BibitemOpen
  \bibfield  {author} {\bibinfo {author} {\bibfnamefont {Y.}~\bibnamefont
  {Wang}}, \bibinfo {author} {\bibfnamefont {Z.}~\bibnamefont {Yin}}, \bibinfo
  {author} {\bibfnamefont {T.}~\bibnamefont {Haug}}, \bibinfo {author}
  {\bibfnamefont {C.}~\bibnamefont {Pentangelo}}, \bibinfo {author}
  {\bibfnamefont {S.}~\bibnamefont {Piacentini}}, \bibinfo {author}
  {\bibfnamefont {A.}~\bibnamefont {Crespi}}, \bibinfo {author} {\bibfnamefont
  {F.}~\bibnamefont {Ceccarelli}}, \bibinfo {author} {\bibfnamefont
  {R.}~\bibnamefont {Osellame}},\ and\ \bibinfo {author} {\bibfnamefont
  {P.}~\bibnamefont {Walther}},\ }\href
  {https://doi.org/10.48550/arXiv.2511.21951} {\bibinfo {title} {A scalable
  advantage in multi-photon quantum machine learning}} (\bibinfo {year}
  {2025}),\ \Eprint {https://arxiv.org/abs/2511.21951} {arXiv:2511.21951
  [quant-ph]} \BibitemShut {NoStop}%
\bibitem [{\citenamefont {Yalouz}\ \emph {et~al.}(2021)\citenamefont {Yalouz},
  \citenamefont {Senjean}, \citenamefont {Miatto},\ and\ \citenamefont
  {Dunjko}}]{Yalouz2021}%
  \BibitemOpen
  \bibfield  {author} {\bibinfo {author} {\bibfnamefont {S.}~\bibnamefont
  {Yalouz}}, \bibinfo {author} {\bibfnamefont {B.}~\bibnamefont {Senjean}},
  \bibinfo {author} {\bibfnamefont {F.}~\bibnamefont {Miatto}},\ and\ \bibinfo
  {author} {\bibfnamefont {V.}~\bibnamefont {Dunjko}},\ }\bibfield  {title}
  {\bibinfo {title} {Encoding strongly-correlated many-boson wavefunctions on a
  photonic quantum computer: application to the attractive bose-hubbard
  model},\ }\href {https://doi.org/10.22331/q-2021-11-08-572} {\bibfield
  {journal} {\bibinfo  {journal} {Quantum}\ }\textbf {\bibinfo {volume} {5}},\
  \bibinfo {pages} {572} (\bibinfo {year} {2021})}\BibitemShut {NoStop}%
\bibitem [{\citenamefont {Aniello}\ \emph
  {et~al.}(2006{\natexlab{b}})\citenamefont {Aniello}, \citenamefont {Lupo},\
  and\ \citenamefont {Napolitano}}]{Aniello-ExploringRepresentation-2006}%
  \BibitemOpen
  \bibfield  {author} {\bibinfo {author} {\bibfnamefont {P.}~\bibnamefont
  {Aniello}}, \bibinfo {author} {\bibfnamefont {C.}~\bibnamefont {Lupo}},\ and\
  \bibinfo {author} {\bibfnamefont {M.}~\bibnamefont {Napolitano}},\ }\bibfield
   {title} {\bibinfo {title} {Exploring {{Representation Theory}} of {{Unitary
  Groups}} via {{Linear Optical Passive Devices}}},\ }\href
  {https://doi.org/10.1007/s11080-006-9023-1} {\bibfield  {journal} {\bibinfo
  {journal} {Open Systems \& Information Dynamics}\ }\textbf {\bibinfo {volume}
  {13}},\ \bibinfo {pages} {415} (\bibinfo {year}
  {2006}{\natexlab{b}})}\BibitemShut {NoStop}%
\bibitem [{\citenamefont {{Garcia-Escartin}}\ \emph {et~al.}(2019)\citenamefont
  {{Garcia-Escartin}}, \citenamefont {Gimeno},\ and\ \citenamefont
  {{Moyano-Fern{\'a}ndez}}}]{Garcia-Escartin-MultiplePhoton-2019}%
  \BibitemOpen
  \bibfield  {author} {\bibinfo {author} {\bibfnamefont {J.~C.}\ \bibnamefont
  {{Garcia-Escartin}}}, \bibinfo {author} {\bibfnamefont {V.}~\bibnamefont
  {Gimeno}},\ and\ \bibinfo {author} {\bibfnamefont {J.~J.}\ \bibnamefont
  {{Moyano-Fern{\'a}ndez}}},\ }\bibfield  {title} {\bibinfo {title} {Multiple
  photon effective {{Hamiltonians}} in linear quantum optical networks},\
  }\href {https://doi.org/10.1016/j.optcom.2018.08.082} {\bibfield  {journal}
  {\bibinfo  {journal} {Optics Communications}\ }\textbf {\bibinfo {volume}
  {430}},\ \bibinfo {pages} {434} (\bibinfo {year} {2019})}\BibitemShut
  {NoStop}%
\bibitem [{\citenamefont {Hall}(2015)}]{Hall2015}%
  \BibitemOpen
  \bibfield  {author} {\bibinfo {author} {\bibfnamefont {B.~C.}\ \bibnamefont
  {Hall}},\ }\href {https://doi.org/10.1007/978-3-319-13467-3} {\emph {\bibinfo
  {title} {Lie Groups, Lie Algebras, and Representations: An Elementary
  Introduction}}}\ (\bibinfo  {publisher} {Springer International Publishing},\
  \bibinfo {year} {2015})\BibitemShut {NoStop}%
\bibitem [{\citenamefont {{Moyano-Fern{\'a}ndez}}\ and\ \citenamefont
  {{Garcia-Escartin}}(2017)}]{Moyano-Fernandez-LinearOptics-2017}%
  \BibitemOpen
  \bibfield  {author} {\bibinfo {author} {\bibfnamefont {J.~J.}\ \bibnamefont
  {{Moyano-Fern{\'a}ndez}}}\ and\ \bibinfo {author} {\bibfnamefont {J.~C.}\
  \bibnamefont {{Garcia-Escartin}}},\ }\bibfield  {title} {\bibinfo {title}
  {Linear optics only allows every possible quantum operation for one photon or
  one port},\ }\href {https://doi.org/10.1016/j.optcom.2016.07.085} {\bibfield
  {journal} {\bibinfo  {journal} {Optics Communications}\ }\textbf {\bibinfo
  {volume} {382}},\ \bibinfo {pages} {237} (\bibinfo {year}
  {2017})}\BibitemShut {NoStop}%
\bibitem [{\citenamefont {Boscain}\ \emph {et~al.}(2015)\citenamefont
  {Boscain}, \citenamefont {Gauthier}, \citenamefont {Rossi},\ and\
  \citenamefont {Sigalotti}}]{Boscain-Approximate-2015}%
  \BibitemOpen
  \bibfield  {author} {\bibinfo {author} {\bibfnamefont {U.}~\bibnamefont
  {Boscain}}, \bibinfo {author} {\bibfnamefont {J.-P.}\ \bibnamefont
  {Gauthier}}, \bibinfo {author} {\bibfnamefont {F.}~\bibnamefont {Rossi}},\
  and\ \bibinfo {author} {\bibfnamefont {M.}~\bibnamefont {Sigalotti}},\
  }\bibfield  {title} {\bibinfo {title} {Approximate controllability, exact
  controllability, and conical eigenvalue intersections for quantum mechanical
  systems},\ }\href {https://doi.org/10.1007/s00220-014-2195-6} {\bibfield
  {journal} {\bibinfo  {journal} {Comm. Math. Phys.}\ }\textbf {\bibinfo
  {volume} {333}},\ \bibinfo {pages} {1225} (\bibinfo {year}
  {2015})}\BibitemShut {NoStop}%
\bibitem [{\citenamefont {Sard}(1942)}]{Sard1942}%
  \BibitemOpen
  \bibfield  {author} {\bibinfo {author} {\bibfnamefont {A.}~\bibnamefont
  {Sard}},\ }\bibfield  {title} {\bibinfo {title} {The measure of the critical
  values of differentiable maps},\ }\href
  {https://doi.org/10.1090/s0002-9904-1942-07811-6} {\bibfield  {journal}
  {\bibinfo  {journal} {Bulletin of the American Mathematical Society}\
  }\textbf {\bibinfo {volume} {48}},\ \bibinfo {pages} {883–890} (\bibinfo
  {year} {1942})}\BibitemShut {NoStop}%
\bibitem [{\citenamefont {Ambainis}\ and\ \citenamefont
  {Emerson}(2007)}]{Ambainis2007}%
  \BibitemOpen
  \bibfield  {author} {\bibinfo {author} {\bibfnamefont {A.}~\bibnamefont
  {Ambainis}}\ and\ \bibinfo {author} {\bibfnamefont {J.}~\bibnamefont
  {Emerson}},\ }\bibfield  {title} {\bibinfo {title} {Quantum t-designs: t-wise
  independence in the quantum world},\ }in\ \href
  {https://doi.org/10.1109/ccc.2007.26} {\emph {\bibinfo {booktitle}
  {Twenty-Second Annual IEEE Conference on Computational Complexity (CCC07)}}}\
  (\bibinfo  {publisher} {IEEE},\ \bibinfo {year} {2007})\ p.\ \bibinfo {pages}
  {129–140}\BibitemShut {NoStop}%
\bibitem [{\citenamefont {Mele}(2024)}]{Mele2024}%
  \BibitemOpen
  \bibfield  {author} {\bibinfo {author} {\bibfnamefont {A.~A.}\ \bibnamefont
  {Mele}},\ }\bibfield  {title} {\bibinfo {title} {Introduction to haar measure
  tools in quantum information: A beginner's tutorial},\ }\href
  {https://doi.org/10.22331/q-2024-05-08-1340} {\bibfield  {journal} {\bibinfo
  {journal} {Quantum}\ }\textbf {\bibinfo {volume} {8}},\ \bibinfo {pages}
  {1340} (\bibinfo {year} {2024})}\BibitemShut {NoStop}%
\bibitem [{Hlo()}]{HloRepo}%
  \BibitemOpen
  \href@noop {} {}\bibinfo {howpublished}
  {\url{https://github.com/TommasoFrancalanci/hlo_jacobian_rank}}\BibitemShut
  {NoStop}%
\bibitem [{\citenamefont {Krantz}\ and\ \citenamefont
  {Parks}(2002)}]{Krantz2002}%
  \BibitemOpen
  \bibfield  {author} {\bibinfo {author} {\bibfnamefont {S.~G.}\ \bibnamefont
  {Krantz}}\ and\ \bibinfo {author} {\bibfnamefont {H.~R.}\ \bibnamefont
  {Parks}},\ }\href {https://doi.org/10.1007/978-0-8176-8134-0} {\emph
  {\bibinfo {title} {A Primer of Real Analytic Functions}}}\ (\bibinfo
  {publisher} {Birkh\"{a}user Boston},\ \bibinfo {year} {2002})\BibitemShut
  {NoStop}%
\bibitem [{\citenamefont {Bochnak}\ \emph {et~al.}(1998)\citenamefont
  {Bochnak}, \citenamefont {Coste},\ and\ \citenamefont
  {Roy}}]{BochnakCosteRoy1998}%
  \BibitemOpen
  \bibfield  {author} {\bibinfo {author} {\bibfnamefont {J.}~\bibnamefont
  {Bochnak}}, \bibinfo {author} {\bibfnamefont {M.}~\bibnamefont {Coste}},\
  and\ \bibinfo {author} {\bibfnamefont {M.-F.}\ \bibnamefont {Roy}},\ }\href
  {https://doi.org/10.1007/978-3-662-03718-8} {\emph {\bibinfo {title} {{Real
  Algebraic Geometry}}}},\ \bibinfo {series} {Ergebnisse der Mathematik und
  ihrer Grenzgebiete (3)}, Vol.~\bibinfo {volume} {36}\ (\bibinfo  {publisher}
  {Springer-Verlag},\ \bibinfo {address} {Berlin, Heidelberg},\ \bibinfo {year}
  {1998})\ pp.\ \bibinfo {pages} {x+430},\ \bibinfo {note} {translated from the
  1987 French original, revised by the authors}\BibitemShut {NoStop}%
\bibitem [{\citenamefont {Kato}(1995)}]{Kato1995}%
  \BibitemOpen
  \bibfield  {author} {\bibinfo {author} {\bibfnamefont {T.}~\bibnamefont
  {Kato}},\ }\href {https://doi.org/10.1007/978-3-642-66282-9} {\emph {\bibinfo
  {title} {Perturbation Theory for Linear Operators}}},\ \bibinfo {edition}
  {reprint of the 1980 edition}\ ed.,\ Classics in Mathematics\ (\bibinfo
  {publisher} {Springer-Verlag Berlin Heidelberg},\ \bibinfo {year} {1995})\
  pp.\ \bibinfo {pages} {xxii+619}\BibitemShut {NoStop}%
\bibitem [{\citenamefont {Guillemin}\ and\ \citenamefont
  {Pollack}(2010)}]{Guillemin2010}%
  \BibitemOpen
  \bibfield  {author} {\bibinfo {author} {\bibfnamefont {V.}~\bibnamefont
  {Guillemin}}\ and\ \bibinfo {author} {\bibfnamefont {A.}~\bibnamefont
  {Pollack}},\ }\href {https://doi.org/10.1090/chel/370} {\emph {\bibinfo
  {title} {Differential Topology}}}\ (\bibinfo  {publisher} {American
  Mathematical Society},\ \bibinfo {year} {2010})\BibitemShut {NoStop}%
\bibitem [{\citenamefont {Yao}\ \emph {et~al.}(2024)\citenamefont {Yao},
  \citenamefont {Miatto},\ and\ \citenamefont {Quesada}}]{Yao2024}%
  \BibitemOpen
  \bibfield  {author} {\bibinfo {author} {\bibfnamefont {Y.}~\bibnamefont
  {Yao}}, \bibinfo {author} {\bibfnamefont {F.}~\bibnamefont {Miatto}},\ and\
  \bibinfo {author} {\bibfnamefont {N.}~\bibnamefont {Quesada}},\ }\bibfield
  {title} {\bibinfo {title} {Riemannian optimization of photonic quantum
  circuits in phase and fock space},\ }\href
  {https://doi.org/10.21468/scipostphys.17.3.082} {\bibfield  {journal}
  {\bibinfo  {journal} {SciPost Physics}\ }\textbf {\bibinfo {volume} {17}},\
  (\bibinfo {year} {2024})}\BibitemShut {NoStop}%
\bibitem [{\citenamefont {Sim}\ \emph {et~al.}(2019)\citenamefont {Sim},
  \citenamefont {Johnson},\ and\ \citenamefont
  {{Aspuru-Guzik}}}]{Sim-ExpressibilityEntangling-2019}%
  \BibitemOpen
  \bibfield  {author} {\bibinfo {author} {\bibfnamefont {S.}~\bibnamefont
  {Sim}}, \bibinfo {author} {\bibfnamefont {P.~D.}\ \bibnamefont {Johnson}},\
  and\ \bibinfo {author} {\bibfnamefont {A.}~\bibnamefont {{Aspuru-Guzik}}},\
  }\bibfield  {title} {\bibinfo {title} {Expressibility and {{Entangling
  Capability}} of {{Parameterized Quantum Circuits}} for {{Hybrid
  Quantum-Classical Algorithms}}},\ }\href
  {https://doi.org/10.1002/qute.201900070} {\bibfield  {journal} {\bibinfo
  {journal} {Advanced Quantum Technologies}\ }\textbf {\bibinfo {volume} {2}},\
  \bibinfo {pages} {1900070} (\bibinfo {year} {2019})}\BibitemShut {NoStop}%
\bibitem [{\citenamefont {Nakaji}\ and\ \citenamefont
  {Yamamoto}(2021)}]{Nakaji2021expressibilityof}%
  \BibitemOpen
  \bibfield  {author} {\bibinfo {author} {\bibfnamefont {K.}~\bibnamefont
  {Nakaji}}\ and\ \bibinfo {author} {\bibfnamefont {N.}~\bibnamefont
  {Yamamoto}},\ }\bibfield  {title} {\bibinfo {title} {Expressibility of the
  alternating layered ansatz for quantum computation},\ }\href
  {https://doi.org/10.22331/q-2021-04-19-434} {\bibfield  {journal} {\bibinfo
  {journal} {{Quantum}}\ }\textbf {\bibinfo {volume} {5}},\ \bibinfo {pages}
  {434} (\bibinfo {year} {2021})}\BibitemShut {NoStop}%
\end{thebibliography}

\appendix

\section{Mapping other schemes to HLO}\label{app:other schemes}

\begin{figure*}[t]
\centering
\includegraphics[width=0.99\textwidth]{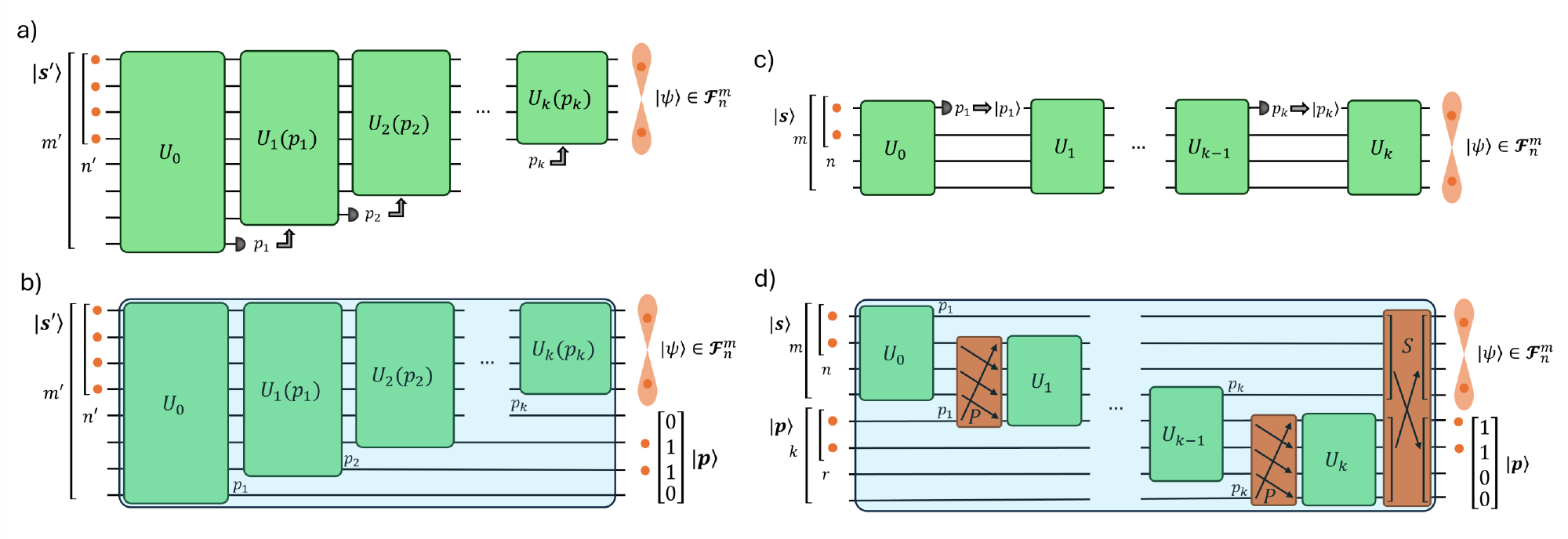}
\caption{Schematics of different settings:  
(a) Adaptive linear optics (ALO), where the evolution consists of a sequence of unitary layers $U_1(p_1), U_2(p_2), \dots, U_k(p_k)$ conditioned on intermediate measurement outcomes.  
(b) Equivalent heralded linear optics (HLO) description of a single adaptive branch, obtained by combining all conditional unitaries into a single global unitary $U' \in \mathcal{U}(m')$ followed by projection onto a fixed heralding outcome.  
(c) Single-mode State Injection (SI) protocol, where photon-counting measurements are interleaved with unitaries and followed by reinjection of the measured photons.  
(d) Equivalent HLO representation of the single-branch SI scheme, obtained by evolving data and auxiliary modes under a composite unitary $\tilde{U}$ and projecting onto the measurement pattern.}
\label{fig: Schemes2}
\end{figure*}

The focus of this work is the generation and controllability of pure photonic states conditioned on a fixed measurement outcome. Within this setting, linear optical architectures involving measurements and feed-forward can be naturally analyzed by restricting attention to a single measurement branch. As we show in this section, such single-branch evolutions can always be recast into an equivalent heralded linear optical (HLO) circuit, consisting of a global unitary followed by a projection onto a fixed heralding outcome. This mapping provides a unified framework to study local controllability across different architectures using the Jacobian of infinitesimal circuit variations.

We first consider adaptive linear optics (ALO)~\cite{Chabaud2021}, where the evolution consists of a sequence of unitaries $U_1(p_1), \dots, U_k(p_k)$, each conditioned on the outcome of an intermediate measurement, as shown in Fig.~\ref{fig: Schemes2} (a). The full measurement pattern is denoted by $\bm{p} = (p_1,\dots,p_k)$. From the heralded perspective, a single adaptive branch corresponding to a fixed pattern $\bm{p}$ can always be rewritten as a single global unitary
\begin{equation}
U' = \big(U_k(p_k)\otimes \mathbb{I}_k\big)\cdots \big(U_1(p_1)\otimes \mathbb{I}_1\big)\, U_0,
\end{equation}
acting on all modes, followed by projection onto the heralding state $\ket{\bm{h}}=\ket{\bm{p}}$, as shown in Fig.~\ref{fig: Schemes2} (b).

While genuine ALO protocols are deterministic in the sense that all measurement outcomes are retained and used for feed-forward control, the heralded representation isolates a single branch whose success probability generally decreases with system size. Nevertheless, at the level of pure-state generation in a fixed branch, adaptive and heralded schemes are operationally equivalent: heralded linear optics (HLO) is the correct framework to study the set of pure states accessible in a single branch of ABS, assuming that each of the adaptive unitaries have free parameters (e.g. programmable phase shifters), such that the total number of degrees of freedom in the adaptive chain matches or exceed those of the equivalent HLO circuit. A difference only arises when considering multiple outcomes, in particular in the mixed-state dynamics one obtains by considering the ensemble average, which is however beyond the scope of the present work.

A similar mapping applies to State Injection (SI) protocols~\cite{Monbroussou2025}.  
In its original formulation, SI is described as a completely positive trace-preserving (CPTP) map, where all measurement outcomes are retained and averaged over, and controllability is analyzed at the level of the resulting mixed output state.  
Here, in order to remain consistent with the heralded setting adopted throughout this work, we instead focus on the evolution conditioned on a fixed measurement pattern, corresponding to a single branch of the full CPTP dynamics. We consider the case of $k$ successive photon-counting measurements of a single mode interleaved with unitaries $U_0,\dots,U_k$ acting on $m$ modes, where the measured photons are re-injected before each subsequent unitary, as illustrated in Fig.~\ref{fig: Schemes2} (c). Focusing on a fixed measurement pattern $\bm{p}=(p_1,\dots,p_k)$, the evolution can again be recast into a heralded circuit by considering an extended input $\ket{\bm{s'}}=\ket{\bm{s}} \otimes\ket{\bm{p}}$ and evolving it under the composite unitary
\begin{equation}
\tilde{U} = S \, \prod_{l=1}^{k} \big( \mathbb{I}_l \otimes U_l P \otimes \mathbb{I}_{k-l} \big) \, (U_0 \otimes \mathbb{I}_k),
\label{eq:Utilde_SI}
\end{equation}
where $\tilde{U}\in\mathcal{U}(m')$, and $P$ and $S$ are fixed permutations. The output state is obtained by projection onto $\ket{\bm{h}}=\ket{\bm{p}}$, as shown in Fig.~\ref{fig: Schemes2} (d). Although $\tilde{U}$ formally belongs to $\mathcal{U}(m')$, it is generated by the $k+1$ local unitaries $\{U_0,\dots,U_k\}$ and therefore explores only a structured submanifold of the full unitary group. This restriction directly affects the dimension of the tangent space of reachable states.

Infinitesimal variations are generated by
\begin{equation}
\delta_{G,l} \ket{\psi} := 
\left. \frac{d}{dt}\, \ket{\psi(U_0,\dots, e^{tG}U_l,\dots,U_k)} \right|_{t=0},
\label{eq:deltaGl}
\end{equation}
with $G\in\mathfrak{u}(m)$ and $l=0,\dots,k$. The associated Jacobian rank $R_{\mathrm{SI}}(\bm{s},\bm{p})$ is therefore upper-bounded by the rank of the corresponding HLO configuration. More precisely, we obtain
\begin{equation}
2m-1 \ \leq\  R_{\mathrm{SI}}(\bm{s},\bm{p}) \ \leq\  R_{\mathrm{HLO}}(\bm{s}+\bm{p},\bm{p}),
\label{eq:rank_bound_SI}
\end{equation}
where the ``$+$" denotes concatenation. The lower bound corresponds to the worst-case scenario in which all photons are measured before the last reinjection, reducing the evolution to passive linear optics.
This shows that, at the level of single-branch pure-state generation, State Injection protocols cannot exceed the local controllability achievable by the corresponding heralded linear optical scheme. In contrast, enhanced controllability may arise when considering the fully averaged CPTP dynamics, as discussed in~\cite{Monbroussou2025}, which lies beyond the scope of the present analysis.

\section{Generic Jacobian rank}\label{app: rank}
To formally establish the existence of a generic rank and discuss the behavior of the Jacobian almost everywhere, we introduce some notions from real-analytic geometry. 

A real-analytic manifold is a topological manifold equipped with an atlas where the transition maps are real-analytic functions (i.e., functions that are locally given by a convergent power series in their arguments \cite{Krantz2002}). Similarly, a map $f: X \to \mathbb{R}^d$ on a real-analytic manifold $X$ is real-analytic if it can be locally represented as convergent power series using analytic coordinates on $X$. Unlike general smooth functions ($C^\infty$), which can be identically zero on a subset without being zero everywhere, a real-analytic function is rigidly constrained by its local behavior: if it vanishes on a 
set of positive measure, it must vanish on the entire connected component of its domain. We thus introduce the following lemma from real-analytic geometry.

\newtheorem{lemma}{Lemma}

\begin{lemma}

Let $X$ be a connected real-analytic manifold of dimension $\dim X$, $d\in \mathbb{N}$, and $f : X \to \mathbb{R}^d$, $g : X \to \mathbb{R}$ be real-analytic maps. Assume that $g$ is not identically zero. 
Denote by
\begin{equation}
J(x) = 
\frac{\partial (f_1/\sqrt{g},\dots,f_{d}/\sqrt{g})(x)}{\partial (x_1,\dots,x_{\dim X})}
\end{equation}
the real Jacobian matrix of $f/\sqrt{g}$ at $x \in \{\xi \in X : g(\xi)\ne 0\}$, and define the rank function
\begin{equation}
R(x) = \mathrm{rank}\,[J(x)].
\end{equation}
Let $R_{\max} = \max\{R(x) : g(x)\ne 0\}$ be the maximal rank attained on $X$.
Then there exists an 
open 
subset $U \subset X$ 
of full measure (and, in particular, dense)
on which $R(x)$ is constant and equal to $R_{\max}$. 

\end{lemma}

\begin{proof}

At a point $x$ at which $g(x)\ne 0$, $R(x)$ is equal to the dimension of 
\begin{equation}\label{eq:jacob}
    \mathcal{V}(x)=\mathrm{span}\{2 g(x) \nabla f_i(x)-f_i(x) \nabla g(x) : i=1,\dots,d\}.
    \end{equation}

For every $x\in X$ the dimension of $\mathcal{V}(\xi)$ is lower bounded by $\mathcal{V}(x)$ for $\xi$ in a sufficiently small neighborhood of $x$, as it follows by a simple continuity argument. In particular, since $g$ is analytic and does not vanish identically, 
for every nonempty open set $U$ in $X$ the maximum dimension of $\mathcal{V}(x)$ in $V$ is attained in 
$\{x \in X : g(x)\ne 0\}$.

Consider $x$ at which $R_{\max}$ is attained and a coordinate neighborhood $U_x$ of $x$. The linear space defined in Eq.~\eqref{eq:jacob} has dimension $R_{\max}$ in an open subset of $U_x$ of full measure, by analiticity of $f$ and $g$. 
Repeating the argument for every coordinate neighborhood intersecting $U_x$ and iterating up to covering the entire connected manifold $X$, we obtain the result. 
\end{proof}

In our context, we apply this lemma to the real-analytic map $\tilde{\psi}$.
The map of interest is
\begin{equation}
\tilde{\psi} : \mathcal{X} \subset \mathcal{U}(m') \longrightarrow S^{2D-1} \subset \mathbb{R}^{2D},
\end{equation}
where $\mathcal{X} = \mathcal{U}(m')  \setminus \mathcal{Z}$ is the open domain on which the success probability $p(U',\bm{s'},\bm{h})$ is positive, and
\begin{equation}
\tilde{\psi}(U') = (\mathrm{Re}\ket{\psi(U',\bm{s'},\bm{h})},\, \mathrm{Im}\ket{\psi(U',\bm{s'},\bm{h})}).
\end{equation}
Since $\mathcal{U}(m')$ is a connected real-analytic manifold
and $p(\cdot,\bm{s'},\bm{h})$ is not identically zero, the lemma applies and it follows that
\begin{equation}
R[J(U',\bm{s'},\bm{h})] = R_{\mathrm{HLO}}(\bm{s'},\bm{h})
\quad \text{for almost all } U' \in \mathcal{X},
\end{equation}
with
\begin{equation}
R_{\mathrm{HLO}}(\bm{s'},\bm{h}) := \max_{U'\in \mathcal{X}} R[J(U',\bm{s'},\bm{h})].
\end{equation}

\section{Proof of Theorem ~\ref{thm:upper_bound}}\label{app:submatrix}

In this appendix, we formally derive the upper bound on the HLO Jacobian rank. The proof proceeds in two steps: first, we establish a general lemma concerning the effective degrees of freedom of a submatrix extracted from a larger unitary matrix, subject to rescaling symmetries. Second, we map the physical parameters of the HLO heralded state onto this geometric framework.

Recall that a semialgebraic set is a subset of $\mathbb{R}^d$, $d\in\mathbb{N}$,
that is the finite union
of sets defined by polynomial equalities and inequalities. 
An important property of semi-algebraic sets is that they admit a \emph{stratification}, that is, they can be decomposed as a finite disjoint union of smooth manifolds \cite{BochnakCosteRoy1998}. We can, in particular, speak of the \emph{dimension} of a semialgebraic set (that is, the maximum of the dimensions of the smooth manifolds in which the set can be decomposed). 
\begin{lemma}[Degrees of freedom of constrained submatrices]\label{lem:submatrix}
Denote by $\mathcal{M}(p,q,N) \subset \mathbb{C}^{p \times q}$ the set of complex matrices that can be embedded as the upper-left block of an $N \times N$ unitary matrix.

Setting two subsets of indices $\tilde{\mathcal{I}} \subset \{1, \dots, p\}$ and $\tilde{\mathcal{J}} \subset \{1, \dots, q\}$, with cardinalities $|\tilde{\mathcal{I}}| = P_{\rm row} $ and $|\tilde{\mathcal{J}}| = P_{\rm col}$, let $\mathcal{D}_{\tilde{\mathcal{I}}} \subset \mathbb{C}^{p \times p}$ and $\mathcal{D}_{\tilde{\mathcal{J}}} \subset \mathbb{C}^{q \times q}$ be the sets of diagonal invertible matrices whose diagonal entries are equal to $1$ outside of the index sets $\tilde{\mathcal{I}}$ and $\tilde{\mathcal{J}}$, respectively.

We then define the following equivalence relation $\sim$ on $\mathcal{M}(p,q,N) \subset \mathbb{C}^{p \times q}$: 
$V\sim V'$ if and only if 
there exist $\Lambda_{\rm row} \in \mathcal{D}_{\tilde{\mathcal{I}}}$ and $\Lambda_{\rm col} \in \mathcal{D}_{\tilde{\mathcal{J}}}$ such that
\begin{equation}
V' = \Lambda_{\rm row} \, V \, \Lambda_{\rm col}.
\end{equation}
By construction, the group of rescalings $\mathcal{D}_{\tilde{\mathcal{I}}} \times \mathcal{D}_{\tilde{\mathcal{J}}}$ is parametrized by $P = P_{\rm row} + P_{\rm col}$ independent complex parameters.

Then, setting $\Omega = p - (N-q)$, the set of equivalence classes for the equivalence relation $\sim$ can be identified with a semialgebraic set of real dimension 
\begin{equation}
d_{\rm eff}(p,q,N) =
\begin{cases}
2pq - 2P, & \Omega \le 0, \\[2mm]
\min\!\big(2pq - 2P,\; 2pq - P - \Omega^2\big), & \Omega > 0.
\end{cases}
\end{equation}

\end{lemma}

\noindent
The proof, reported below, follows from the analysis of the interplay between orthogonality constraints imposed by the unitary embedding and the gauge redundancies captured by the equivalence relation $\sim$.

\begin{proof}
We first recall how to characterize the dimension of the  set
$\mathcal{M}(p,q,N)$, which is semialgebraic since the projection of any semialgebraic set is again semialgebraic (Tarski–Seidenberg principle \cite{BochnakCosteRoy1998}).

To characterize the dimension of $\mathcal{M}(p,q,N)$, consider the restriction of $U$ to its first $p$ rows:
\begin{equation}
U_{\mathrm{eff}} = \big[V^{p \times q} \ \big|\ A^{p \times (N-q)} \big],
\end{equation}
where $A$ collects the remaining columns. 
A matrix $V$ belongs to $\mathcal{M}(p,q,N)$ if and only if there exists $A$ for which the lines of $U_{\mathrm{eff}}$ are orthonormal, i.e.,
\begin{equation}
\label{eq:unitarity_constraint_app}
V V^\dagger + A A^\dagger = \mathbb{I}_p.
\end{equation}
This implies that the defect matrix 
\begin{equation}
K := \mathbb{I}_p - V V^\dagger,
\end{equation}
must be representable as $AA^{\dagger}$ for some $A \in \mathbb{C}^{p \times (N-q)}$. This imposes two conditions on $K$.

Since $K = AA^\dagger$, it must be Hermitian and positive semidefinite. This requirement is equivalent to $V$ being a contraction, i.e., its singular values must satisfy $0 \le \sigma_i \le 1$.

Moreover, since $A$ has only $N-q$ columns, the rank of $AA^\dagger$ cannot exceed this number. Thus, $K$ must satisfy:
\begin{equation}\label{eq:Vcompatible}
\mathrm{rank}(K) \le N-q.
\end{equation}

Since $K$ is a $p \times p$ matrix, its rank is always bounded by $p$. Thus, when $p \le N-q$, Eq.~\eqref{eq:Vcompatible} is trivially satisfied for any $V$ in the contractive set, $\{V\in \mathbb{C}^{p\times q} : \sigma_{\max}(V)<1\}$, with $\sigma_{\max}(\cdot)$ denoting the spectral norm. The set $\mathcal{M}(p,q,N)$
therefore has full dimension $2pq$.

When $p > N-q$, instead, the matrix $K$ must be a Hermitian matrix with at least $\Omega =  p - (N-q)$ zero eigenvalues (equivalently, $V$ has at least $\Omega$ singular
values equal to one).
This condition removes $\Omega^2$ real degrees of freedom, as can be explicitly shown by considering the spectral decomposition $K = U \Lambda U^\dagger$. Here, $\Lambda$ is a real diagonal matrix and $U$ is a unitary matrix defined modulo a phase rescaling of its columns. Since $\Omega$ eigenvalues are fixed to zero, $\Omega$ real degrees of freedom are removed from the diagonal matrix $\Lambda$. Moreover, the spectral representation is invariant under the transformation $U \to UB$, where $B$ is a block-diagonal unitary matrix acting as $B_{\Omega} \in \mathcal{U}(\Omega)$ on the null space of $K$ and as the identity elsewhere. Since $\Omega$ phase-scaling redundancies are already accounted for, the additional parameters required to specify $B$, acting as redundant degrees of freedom, are $\Omega^2 - \Omega$. Summing these contributions, we find that the total number of removed real degrees of freedom is $\Omega + (\Omega^2 - \Omega) = \Omega^2$.

The real dimension of $\mathcal{M}(p,q,N)$ for $\Omega \geq 0$ is then given by:
\begin{equation}
\dim_{\mathbb{R}} \mathcal{M}(p,q,N) = 2pq - \Omega^2.
\end{equation}
We now incorporate the effect of $P$ independent complex rescalings acting on rows and columns of $V$. Each rescaling contributes two real parameters (amplitude and phase).

\subsubsection*{Free regime ($\Omega \le 0$)}

In this case, taking $V$ in the interior
of $\mathcal{M}(p,q,N)$, all rescalings 
$V\to \Lambda_{\rm row}\, V \, \Lambda_{\rm col}$ with $\Lambda_{\rm col}$ and  $\Lambda_{\rm row}$ close to the identity
preserve the constraint~\eqref{eq:unitarity_constraint_app}. Therefore, all $2P$ real
parameters act as gauge redundancies, yielding
\begin{equation}
d_{\mathrm{eff}} = 2pq - 2P.
\end{equation}

\subsubsection*{Saturated regime ($\Omega > 0$)} 

In this regime, the condition for $V$ to be in $\mathcal{M}(p,q,N)$ is that $\Omega$ singular values are exactly equal to $1$. This is equivalent to requiring that the Gram matrix $G = V^\dagger V$ has an eigenvalue $\lambda = 1$ with multiplicity $\Omega$.

In particular, a phase rescaling $V' = \Lambda_\theta V \Lambda_{\phi}$, with unitary diagonal matrices $\Lambda_\theta, \Lambda_\phi$, induces a unitary transformation $G' = \Lambda_\phi^\dagger G \Lambda_\phi$ of the Gram matrix, preserving the spectrum. Thus, phase rescalings preserve $\mathcal{M}(p,q,N)$.

Consider now an amplitude rescaling coherent with the equivalence relation $\sim$:
\begin{equation}
    V(s) = (\mathbb{I} + s\Delta^{\rm{row}}) V (\mathbb{I} + s\Delta^{\rm{col}})
\end{equation}
depending on a parameter $s$, where $\Delta^{\rm{row}}$ 
and  $\Delta^{\rm{col}}$ are real-valued diagonal matrices. 
The matrices $\Delta^{\rm row}$ and $\Delta^{\rm col}$ are the infinitesimal generators of the real part of  $\mathcal{D}_{\tilde{\mathcal{I}}}$ and $\mathcal{D}_{\tilde{\mathcal{J}}}$, meaning their diagonal entries $\delta_{ii}$ are non-zero only for $i \in \tilde{\mathcal{I}}$ (resp. $j \in \tilde{\mathcal{J}}$).

The eigenpairs $(\lambda_k(s),w_k(s))$ of the Hermitian matrix $V(s)^\dagger V(s)$ can be chosen to depend analytically on $s$. Moreover,  
labeling so that $\lambda_k$ with $k=1,\dots,\Omega$ are the eigenvalues with $\lambda_k(0)=1$, we have  
\begin{equation}
    w_k^\dagger(0) \left.\frac{d}{ds}\right|_{s=0}(V(s)^\dagger V(s))w_j(0)=\delta_{kj}\left.\frac{d}{ds}\right|_{s=0}\lambda_k(s)
\end{equation}
for $k,j=1,\dots,\Omega$ \cite[Chapter~II, \S6]{Kato1995}.
Therefore, the condition that the singular values of $V(s)$ corresponding to $1$ are preserved at first order at $s=0$ can be rewritten as 
\begin{equation}\label{eq:MOmega=0}
M_\Omega=0,
\end{equation}
where $M_\Omega$ is the projection of $\left.\frac{d}{ds}\right|_{s=0}(V(s)^\dagger V(s))$ onto the space spanned by $w_1(0),\dots,w_\Omega(0)$, that is, 
\begin{equation}
   M_\Omega=W_\Omega^\dagger((\Delta^{\rm{row}}V+V\Delta^{\rm{col}})^\dagger V+V^\dagger(\Delta^{\rm{row}}V+V\Delta^{\rm{col}}) )W_\Omega, 
\end{equation}
where $W_\Omega$ is the $q \times \Omega$ matrix having columns $w_1(0),\dots,w_\Omega(0)$.

Defining $U_{\Omega}$ as the $p \times \Omega$ matrix whose columns are the left singular vectors of $V$ corresponding to the singluar values equal to $1$, we have the SVD relations
\begin{equation}\label{eq: SVD-m}
    V W_{\Omega} = U_{\Omega}, \qquad V^\dagger V W_\Omega=W_\Omega,
\end{equation} 
from which we get
\begin{equation}
M_\Omega =2(U_\Omega^\dagger \Delta^{\rm{row}} U_\Omega + W_\Omega^\dagger \Delta^{\rm{col}} W_\Omega).
\end{equation}

Notice that if we consider only a single row rescaling $\Delta^{\rm{row}}_i =  e_i e_i^T$, Eq.~\eqref{eq:MOmega=0} reduces to $\tilde{u}_i \tilde{u}_i^\dagger = 0$, where $\tilde{u}_i^\dagger$ is the $i$-th row of $U_\Omega$. For a generic $V$, this row is non-vanishing, implying that 
$\mathcal{M}(p,q,N)$ is not invariant under
individual rescalings. 

The effective number of amplitude gauge redundancies is determined by the dimension of the solution space of the matrix equation $M_{\Omega}(\Delta^{\rm{row}}, \Delta^{\rm{col}}) = 0$. Since $M_{\Omega}$ is Hermitian, this vanishing condition consists of $\Omega^2$ independent real equations for the $P$ amplitude parameters contained in the diagonal matrices $\Delta^{\rm{row}}$ and $\Delta^{\rm{col}}$. Two distinct regimes then further emerge.

\medskip

Overdetermined case ($\Omega^2 \geq P$): when the number of saturated singular values is large enough, the $\Omega^2$ constraints equal or exceed the available amplitude rescalings. For a generic point $V$, the only solution to the tangency condition is the trivial one ($\Delta^{\rm{row}} = 0, \Delta^{\rm{col}} = 0$). In this regime, any combination of amplitude rescalings necessarily shifts the saturated singular values away from unity, pushing the matrix out of the set $\mathcal{M}(p,q,N)$. 

\medskip

Underdetermined case ($\Omega^2 < P$): if the number of rescalings $P$ is sufficient to satisfy the spectral constraints, the system is underconstrained, and there exists a residual subspace of dimension $P - \Omega^2$ where row and column rescalings can be mutually balanced to keep the saturated singular values identically equal to $1$. In this subspace, the variations preserve $\mathcal{M}(p,q,N)$.

\medskip

Combining these validity constraints on amplitude symmetries with the always-valid $P$ phase symmetries, the total number of gauge directions that remain tangent to the manifold is $P + \max(0, P - \Omega^2)$. Subtracting these from the manifold dimension $2pq - \Omega^2$ yields
\begin{equation}
d_{\mathrm{eff}} = \min\big(2pq - 2P,\; 2pq - \Omega^2 - P\big).
\end{equation}

This establishes the expression for the effective number of independent real parameters of $V$ under unitary constraints and gauge redundancies, thus proving Lemma~\ref{lem:submatrix}.
\end{proof}

We now map this geometric result onto the HLO setting, thereby translating the effective degrees of freedom of the submatrix $V$ into a bound on the HLO Jacobian rank. This provides the proof of Theorem \ref{thm:upper_bound}, which we report below.

\begin{proof}
The heralded output state depends on the interferometer only through a submatrix
\begin{equation}
V = U'[\mathcal{I}, \mathcal{J}] \in \mathbb{C}^{p \times q},
\end{equation}
where $\mathcal{I}$ is the set of $n'$ occupied input modes and $\mathcal{J}$ is the set of allowed output modes, consisting of the $m$ unmeasured modes together with the $k_{\rm occ}$ occupied heralding modes. 
Eq.~\eqref{eq: heralded_state} can thus be rewritten as
\begin{equation}\label{eq:heralded_state_V}
\ket{\Phi} = \sum_{\substack{j_1,\dots,j_{n'} \in \mathcal{J} \\ \sum_{i \in \mathcal{I}} \delta_{j_i,l} = r_l \ \forall l\in \mathcal{H}}} 
\bigg( \prod_{i \in \mathcal{I}} V_{i,j_i} \bigg) 
\hat a_{j_1}^\dagger \dots \hat a_{j_{n'}}^\dagger \ket{0}.
\end{equation}
Therefore, the number of independent directions accessible through variations
of $U'$ is bounded by the number of effective degrees of freedom of $V$.

From Eq.~\eqref{eq:heralded_state_V}, the dependence of $\ket{\Phi}$
on $V$ is entirely through products of matrix elements of the form
$\prod_{i \in \mathcal{I}} V_{i,j_i}$.

Consider now a rescaling of the rows of $V$, $V_{i,j} \;\rightarrow\; \lambda_i \, V_{i,j}$, with $\lambda_i \in \mathbb{C}$. Each term in the sum acquires a factor
$\prod_{i \in \mathcal{I}} \lambda_i$, yielding $\ket{\Phi} \;\rightarrow\;
\left( \prod_{i \in \mathcal{I}} \lambda_i \right)\ket{\Phi}$.

Similarly, denoting with $\mathcal{H}_{\rm occ}$ the set of occupied heralding modes, a rescaling of the columns corresponding to these modes, $V_{i,j} \;\rightarrow\; \mu_j \, V_{i,j}$, with $j \in \mathcal{H}_{\rm occ}$, produces an overall factor $\ket{\Phi} \;\rightarrow\;
\left( \prod_{j \in \mathcal{H}_{\rm occ}} \mu_j^{\,r_j} \right)\ket{\Phi}$, since each heralded mode $j$ appears exactly $r_j$ times in every allowed configuration.

Therefore, all row and column rescalings act on $\ket{\Phi}$ only through an overall complex multiplicative factor. In particular, the $n'$ row rescalings and the $k_{\rm occ}$ column rescalings define a total of
\begin{equation}
P = n' + k_{\rm occ}
\end{equation}
independent complex gauge parameters.

However, these $2P$ real directions do not generate independent variations of the physical state: they collapse to a two-dimensional real subspace corresponding to global amplitude and phase transformations of $\ket{\Phi}$. Upon normalization, the global amplitude is removed, while the global phase remains as the only residual direction.

We can therefore apply Lemma~\ref{lem:submatrix} to the submatrix $V \in \mathbb{C}^{p \times q}$, where in the HLO setting, the embedding dimension is $N = m'$. While normalization removes the global amplitude, the global phase remains a valid direction for the Jacobian of the state vector, thus $R = d_{\rm eff} + 1$. This yields the expressions for $R_{\rm free}$ and $R_{\rm sat}$.

Finally, the rank cannot exceed the total number of independent real parameters of the output state, which lies on the unit sphere of the $n$-photon Hilbert space over $m$ modes and has real dimension $2D-1$. Combining these bounds yields the statement of the theorem.
\end{proof}

\section{Proof of Theorem ~\ref{thm:lower bounds}}\label{app:lower bounds}

We provide here the proof of the lower bounds on the generic HLO Jacobian rank in the regime $\Omega \le 0$. In the $\Omega \leq 0$ regime, there are sufficiently many modes to satisfy the global unitarity constraints without restricting the local degrees of freedom of the submatrix $V \in \mathbb{C}^{n \times m}$. In particular, if the spectral norm satisfies $\sigma_{\max}(V) < 1$, the matrix $V$ can be embedded into a larger unitary $U' \in \mathcal{U}(m')$. The set of such contractive matrices is an open subset of $\mathbb{C}^{n \times m}$, so that the entries of $V$ provide valid local coordinates for the map. The lower bounds for the considered regimes are obtained by considering the map $V \mapsto \ket{\psi}$ and its Jacobian rank at a suitably chosen point $V^*$.

\subsection{Case $r=0$}

We consider first the case $r=0$, and we consider input states of the form $\ket{s'} = \ket{\bm{s}}_{\mathcal{D}}\otimes \ket{\bm{0}}_{\mathcal{H}}$, where $\ket{\bm{s}} = |1^{\otimes n} 0^{\otimes (m-n)}\rangle$, with $n \leq m$. The heralded state factorizes as $\ket{\Phi}=\ket{\phi}_{\mathcal{D}} \otimes \ket{\bm{0}}_{\mathcal{H}}$ and depends only on the data-sector matrix $V \in \mathbb{C}^{n \times m}$.

\begin{proof} 
The unnormalized state is
\begin{equation}
\ket{\phi} =
\sum_{j_1,\dots,j_n=1}^m
\bigg{(} \prod_{i=1}^n V_{i,j_i} \bigg{)}
\hat a_{j_1}^\dagger \cdots \hat a_{j_n}^\dagger \ket{0}.
\end{equation}
The derivative with respect to $V_{k,l}$ reads
\begin{equation}\label{eq: derivative r0}
\frac{\partial \ket{\phi}}{\partial V_{k,l}} =
\sum_{j_1,\dots,j_n=1}^m
\bigg{(} \delta_{j_k,l} \prod_{i \neq k} V_{i,j_i} \bigg{)}
\hat a_{j_1}^\dagger \cdots \hat a_{j_n}^\dagger \ket{0}.
\end{equation}
Each derivative is thus identified with a vector in $\mathbb{C}^D$, corresponding to a column of the complex Jacobian.

To establish the lower bound, we evaluate the Jacobian at the point
$ V^* = (\alpha  \,\mathbb{I} \ |\ 0)$, with $ 0 < \alpha < 1$.
At $V^*$, Eq.~\eqref{eq: derivative r0} reduces to
\begin{align}
    \left. \frac{\partial \ket{\phi}}{\partial V_{k,l}} \right|_{V^*}
&=   \sum_{j_1,\dots,j_n=1}^m
\bigg( \delta_{j_k,l} \prod_{i\neq k} \alpha\delta_{i,j_i}\bigg)
\hat a_{j_1}^\dagger \cdots \hat a_{j_n}^\dagger \ket{0} \nonumber\\&=\alpha^{n-1} \hat{a}_l^\dagger \hat{a}_k \ket{\bm{s}}.
\end{align}
For $l \neq k$, this is proportional to the Fock state, obtained by moving one photon from mode $k$ to $l$. Different pairs $(k,l)$ yield orthogonal Fock states, hence linearly independent directions. This provides $2n(m-1)$ real directions.

For $l = k$, all derivatives are proportional to $\ket{\bm{s}}$, yielding a single complex direction. Thus, the total real rank for the unnormalized state is $2n(m-1)+2$. After normalization, the radial (norm) direction is removed, while the global phase remains. Letting $U'^*\in \mathcal{U}(m')$ be such that $V^*$ is its upper left $n \times m$ block, we find
\begin{equation}
\mathrm{rank}\ J(U'^*,\bm{s'},\bm{0})= 2n(m-1) + 1.
\end{equation}
The generic value $R_{\rm{HLO}}(\bm{s'},\bm{0})$ is necessary greater than or equal to the value $\mathrm{rank}\ J(U'^*,\bm{s'},\bm{0})$ as it is the maximal value, hence the bound is proven.
\end{proof}

\subsection{Case $r>0$, $n>2$, $n+r \le m$}

We now consider heralding configurations with $r>0$, 
considering heralding patterns of the form $\ket{\bm{h}} = |1^{\otimes r} 0^{\otimes (k-r)}\rangle$. In this regime, the scattering submatrix is $V \in \mathbb{C}^{(n+r) \times (m+r)}$, where the first $m$ columns correspond to the data modes and the remaining $r$ columns correspond to the occupied heralding modes. 

\begin{proof}
The unnormalized heralded state is given by
\begin{equation}
\ket{\Phi} =
\sum_{\substack{j_1,\dots,j_{n+r}=1 \\ \text{heralding constraints}}}^{m+r}
\bigg( \prod_{i=1}^{n+r} V_{i,j_i} \bigg)
\hat a_{j_1}^\dagger \cdots \hat a_{j_{n+r}}^\dagger \ket{0}.
\end{equation}
Taking the derivative with respect to a matrix element $V_{k,l}$ fixes the $k$-th photon to the output mode $l$. The remaining $n+r-1$ photons are then distributed such that either $r$ of them satisfy the heralding requirement (if $l \le m$) or $r-1$ satisfy it (if $l > m$). The derivative expression is:
\begin{equation}
\frac{\partial \ket{\Phi}}{\partial V_{k,l}}  = \sum_{\substack{j_1,\dots,j_{n+r}=1 \\ \text{heralding constraints}}}^{m+r} \delta_{j_k,l}
\bigg( \prod_{i\neq k} V_{i,j_i} \bigg) 
\hat a_{j_1}^\dagger \cdots \hat a_{j_{n+r}}^\dagger \ket{0},
\end{equation}
which can be factorized as 
$\partial \ket{\Phi} /\partial V_{k,l} = \partial \ket{\phi} /\partial V_{k,l} \otimes \ket{\bm{h}}_{\mathcal{H}}$. All non-trivial variations of the state are thus entirely captured by the data-sector component $\ket{\phi}$, and in the following we restrict the analysis to this subsystem. Once again the derivative of the data-sector state $\partial \ket{\phi} /\partial V_{k,l}$ can be represented as a vector in $\mathbb{C}^D$, corresponding to a column of the complex Jacobian.

To establish the lower bound we restrict to the configurations where $n>2$ and $n+r\leq m$. Within this setting we can consider input states of the form $\ket{\bm{s}'}=\ket{\bm{\tilde{s}}}_{\mathcal{D}}\otimes\ket{\bm{0}}_{\mathcal{H}}$, where $\ket{\bm{\tilde{s}}}= |1^{\otimes (n+r)} 0^{\otimes (m-(n+r))}\rangle$. We then evaluate the derivatives at a reference point of the form $V^* = (\alpha \,\mathbb{I} \ |\ H)$
where $0 < \alpha < 1$, while the heralding block $H \in \mathbb{C}^{(n+r)\times r}$ is filled with strictly positive real entries chosen so that $\sigma_{\max}(V^*) < 1$.

We first consider derivatives $\partial \ket{\phi}/\partial V_{k,l}$  with respect to matrix elements in the data sector ($l \le m$). Evaluating them at $V^*$, the structure of the diagonal block forces all photons $i \neq k$ that are not assigned to the heralding modes to occupy their corresponding modes $i$. The remaining $r$ photons are distributed among the heralding modes according to the pattern $\bm{h}$. As a result, for $l = k$, the diagonal derivatives take the form: 
\begin{align}
\left. \frac{\partial \ket{\phi}}{\partial V_{k,k}} \right|_{V^*}
&= \alpha^{n-1} \hat{a}_k^\dagger 
\sum_{ \bm{p}_k}   \mathrm{Perm} (H_{ \bm{p}_k} )
\bigg( \prod_{i \notin \bm{p}_k,\, i \neq k} \hat{a}_i^\dagger \bigg) \ket{0}
\nonumber\\&= \alpha^{n-1}
\sum_{ \bm{p}_k}   \mathrm{Perm} (H_{ \bm{p}_k} )
\bigg( \prod_{i \in \bm{p}_k} \hat{a}_i \bigg) \ket{\tilde{\bm{s}}}_{\mathcal{D}}, 
\end{align}
where the sum is over all the subsets $ \bm{p}_k$ of $r$ indices chosen from $\{1,\dots,n+r\}\setminus\{k\}$, and 
$H_{\bm{p}_k}$
is the corresponding $r\times r$ submatrix of $H$.

This sector has support on Fock states obtained by removing exactly $r$ photons from the reference state $\ket{\tilde{\bm{s}}}$. In particular, for each index $k$, the vector $\partial \ket{\phi}/\partial V_{k,k}$ has support on all configurations obtained by removing $r$ photons from $\ket{\tilde{\bm{s}}}$, with the exception of those where the photon in mode $k$ is among those removed. Equivalently, its support is identified by the set of all removal patterns $\bm{p}$ of size $|\bm{p}|=r$ such that $k \notin \bm{p}$. To establish linear independence, consider a linear combination $\sum_{k=1}^{n+r} \lambda_k \, \frac{\partial \ket{\phi}}{\partial V_{k,k}} = 0$. Projecting this equation onto each Fock basis state $\ket{\mathbf{n}}_{\bm{p}}$ labeled by the removal pattern $\bm{p}$ yields:
\begin{equation}
\text{Perm}(H_{\bm{p}}) 
\sum_{k \notin \bm{p}} \lambda_k  
= 0.
\end{equation}
Given that all coefficients $\text{Perm}(H_{\bm{p}})$ are strictly positive, this means that the vector $(\lambda_1,\dots,\lambda_{n+r})$ is orthogonal to each vector of $\mathbb{R}^{n+r}$ with $r$ components equal to zero and the other equal to 1 (each corresponding to a different $\bm{p}$). Since the family of all such vectors spans $\mathbb{R}^{n+r}$,  
this guarantees that the only solution is the trivial one, $\lambda_k = 0$ for all $k$. Consequently, the $n+r$ diagonal derivatives provide $2(n+r)$ independent real directions.

For $l \neq k$, instead, the derivative $\partial \ket{\phi}/\partial V_{k,l}$ takes the form
\begin{align}
   \left. \frac{\partial \ket{\phi}}{\partial V_{k,l}} \right|_{V^*}
&= \alpha^{\,n-1}\, \hat a_l^\dagger 
\sum_{ \bm{p}_k}\mathrm{Perm} (H_{ \bm{p}_k} )\,
\bigg( \prod_{i \notin \bm{p}_k,\, i \neq k} \hat a_i^\dagger \bigg)
\ket{0}\\
&=\alpha^{n-1}\sum_{ \bm{p}_k}
\mathrm{Perm} (H_{ \bm{p}_k} ) \ \hat a_l^\dagger \hat{a}_k \bigg( \prod_{i \in \bm{p}_k}\hat{a}_i \bigg)\ket{\tilde{\bm{s}}}_{\mathcal{D}}, 
\end{align}
where the sum is over all possible subsets $ \bm{p}_k$ of $r$ indices chosen from $\{1, \dots, n+r\} \setminus \{k\}$, and $H_{ \bm{p}_k}$ is the corresponding $r\times r$ submatrix of $H$.

At the reference point $V^*$, the $k$-th photon is uniquely reassigned to the target mode $l$. To establish the independence of the off-diagonal sector, we focus on the Fock subspaces $\{\mathcal{F}_{l}\}$ characterized by an excess occupation in mode $l$ ($n_l = \tilde{\bm{s}}_l+ 1$). These subspaces are mutually orthogonal for different $l$, and the diagonal derivatives have null projection onto them, as they populate exclusively states with no excess occupation. Within each $\mathcal{F}_l$, the derivatives have support on the set of states obtained by the addition of one photon in mode $l$ and the removal of exactly $r+1$ photons from the originally occupied modes, excluding mode $l$ itself. The set of available modes for this removal contains $n+r-1$ indices if $l \le n+r$, and $n+r$ indices if $l > n+r$. Each such state is identified by a removal pattern $\bm{P} \subset \{1, \dots, n+r\}\setminus\{l\}$ of size $|\bm{P}|=r+1$. A derivative $\partial \ket{\phi}/\partial V_{k,l}$ has non-zero support on the state labelled by $\bm{P}$ if and only if $k \in \bm{P}$, as the photon originating from mode $k$ must be among those removed from the data sector. Projecting a linear combination $\sum_k \lambda_k \frac{\partial \ket{\phi}}{\partial V_{k,l}} = 0$ onto each basis state $\ket{\mathbf{n}}_{\bm{P}}$ yields:
\begin{equation}\label{eq:Hpmk}
\sum_{k \in \bm{P}} \lambda_k \, \mathrm{Perm} (H_{\bm{P} \setminus \{k\}} ) = 0.
\end{equation}
Since $H$ has strictly positive entries, all coefficients $\mathrm{Perm} (H_{\bm{P} \setminus \{k\}} )$ are non-vanishing. 
Under the additional assumption that all entries of $H$ have the same value $\beta$, with $0<\beta<1$, 
one has that $\mathrm{Perm} (H_{\bm{P} \setminus \{k\}} )=r!\beta^r$ is independent of $\bm{P}$ and of $k$.
The collection of equalities \eqref{eq:Hpmk} obtained by varying $\bm{P}$ can be seen as the orthogonality of the vector $\lambda=(\lambda_1,\dots,\lambda_{l-1},\lambda_{l+1},\dots,\lambda_{n+r})$ (if $l\le n+r$) or $\lambda=(\lambda_1,\dots,\lambda_{n+r})$ (if $l> n+r$) with a family of vectors $v_{\bm{P}}$ each 
having coefficient $1$ for the coordinates corresponding to $\bm{P}$ and zero otherwise. Hence $\lambda$ must vanish.
 
This confirms the linear independence of the derivatives for each data mode $l$. Notice that the previous argument relies on the assumption $n>2$ since, for $n=2$, the condition for $l \le n+r$ fails as the system collapses to a single equation, consistently with possible rank deficiency expected in the $n=2$ regime. 

Since different output modes $l$ populate disjoint subspaces $\mathcal{F}_l$, the off-diagonal data sector provides $2(m-1)(n+r)$ independent real directions, all of which are linearly independent from the diagonal sector.

Normalization removes at most one real direction. Collecting the contributions, we obtain the lower bound
\begin{align}
\mathrm{rank} \ J(U'^*,\bm{s'},\bm{h}) &\ge 2(m-1)(n+r) + 2(n+r) - 1 \nonumber \\
& = 2m(n+r) - 1,
\end{align}
where we recall that 
$U'^*$ is such that 
$V^*$ is its upper left $(n+r) \times (m+r)$ block and the subtraction of one accounts for the removal of the radial direction upon normalization. The generic value $R_{\rm{HLO}}(\bm{s'},\bm{h})$ is necessary greater than or equal to the value $\mathrm{rank} \ J(U'^*,\bm{s'},\bm{h})$ as it is the maximal value by definition, hence the bound is proven.
\end{proof}

\emph{Remark.} 
It is worth noting that the derivatives with respect to the Heralding sector ($m+1 \leq l \leq m+r$) do not necessarily provide additional independent directions beyond those already identified. Specifically, at the reference point $V^*$, these derivatives take the form:
\begin{align}
    \left. \frac{\partial \ket{\phi}}{\partial V_{k,l}} \right|_{V^*}
&= \alpha^{n-1} \ \sum_{ \bm{q}_k} \mathrm{Perm} (H_{ \bm{q}_k,l} )
\bigg( \prod_{i \notin  \bm{q}_k,\, i \neq k} \hat{a}_i^\dagger \bigg) \ket{0} \nonumber \\&=\alpha^{n-1} \sum_{ \bm{q}_k} \mathrm{Perm} (H_{ \bm{q}_k,l} ) \ \hat{a}_k 
\bigg( \prod_{i \in  \bm{q}_k}\hat{a}_i \bigg) \ket{\tilde{\bm{s}}}_{\mathcal{D}}  ,
\end{align}
where the sum is over all the subsets $ \bm{q}_k$ of $r-1$ indices chosen from $\{1,\dots,n+r\}\setminus\{k\}$, and $H_{ \bm{q}_k,l}$ is the  $(r-1)\times(r-1)$ submatrix of $H$, obtained by using the corresponding rows and excluding column $l$. Since these variations have support on the same set of Fock states as the diagonal sector, they are not guaranteed to increase the lower bound. This observation clarifies why the current proof saturates at $2m(n+r) - 1$ and suggests that further improvement of the bound would require a more refined analysis of the dependencies between the diagonal and heralding sectors, possibly through a specific non-uniform choice of the matrix $H$. 

\section{Details on global reachability}\label{app: global}
\begin{figure*}[t]
    \centering
    \includegraphics[width=0.8\textwidth]{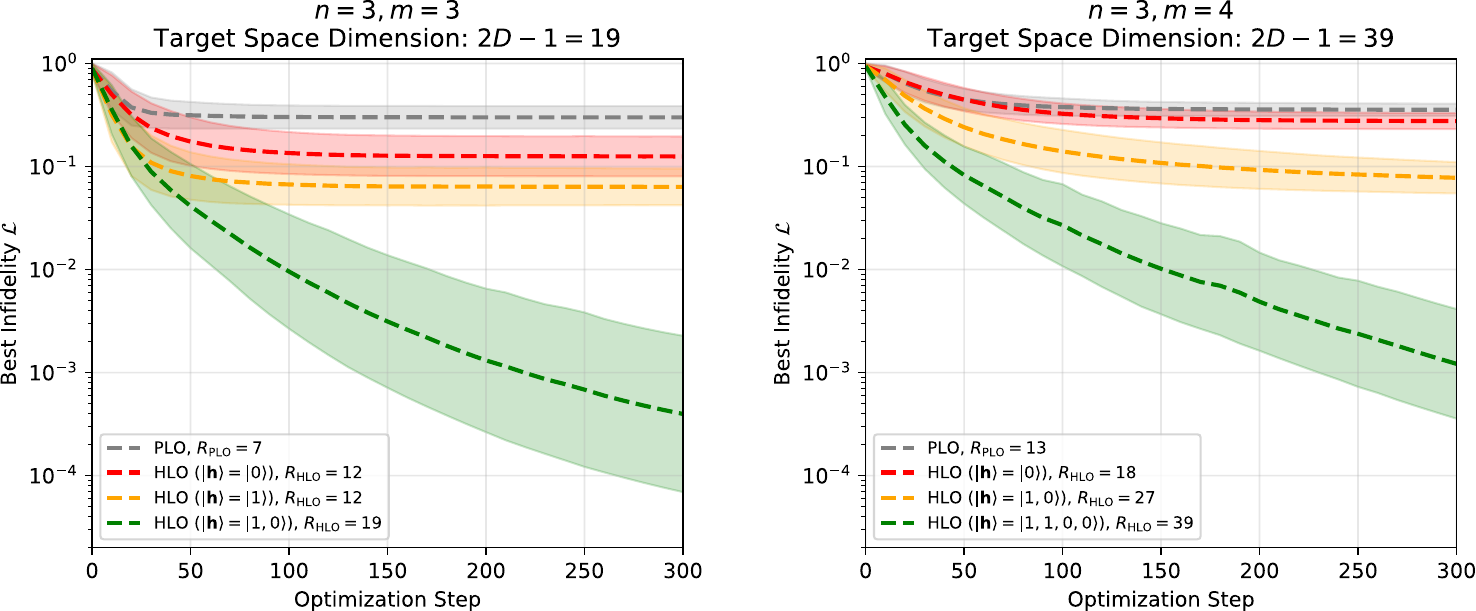}
    \caption{Evolution of the best infidelity $\mathcal{L}$ as a function of the optimization steps for output space with either $n=3$ photons in $m=3$ modes (target space dimension $2D-1 = 19$) or $n=3$ photons in $m=4$ modes (target space dimension $2D-1 = 39$). The curves represent the geometric mean over $M=200$ Haar-random target states, with shaded areas indicating the geometric standard deviation. The gray dashed line corresponds to standard Passive Linear Optics (PLO), which quickly plateaus due to insufficient degrees of freedom. HLO configurations with suboptimal ancillary resources show improved performance but still exhibit a saturation of the loss. In contrast, the HLO configuration with maximal Jacobian rank (green curve) shows an sustained decay, reaching infidelities orders of magnitude lower within $T=300$ steps. Each data point is the result of the best among $10$ independent restarts with different initial unitaries.}
    \label{fig:opt_infidelity}
\end{figure*}
We consider again the normalized HLO map on the data modes,
\begin{equation}
\begin{aligned}
&\tilde{\psi} : \mathcal{X} \longrightarrow S^{2D-1} \subset \mathbb{R}^{2D},\\
&\tilde{\psi}(U') = \bigl( \mathrm{Re}\ket{\psi(U',\bm{s'},\bm{h})}, \mathrm{Im}\ket{\psi(U',\bm{s'},\bm{h})} \bigr),
\end{aligned}
\end{equation}
where $\mathcal{X}\subset \mathcal{U}(m')$ is the open domain where the success probability is strictly positive,
\begin{equation}
\mathcal{X} = \{U'\in\mathcal{U}(m') :  p(U',\bm{s'},\bm{h})>0\}.
\end{equation}
Hence, $\tilde\psi$ is well defined on $\mathcal{X}$, and the map is smooth and real analytic on its domain.

The complement
\begin{equation}
\mathcal{Z}= \mathcal{U}(m')\setminus \mathcal{X}
= \{U'\in\mathcal{U}(m') :  p(U',\bm{s'},\bm{h})=0\}
\end{equation}
is the \emph{singular locus}, consisting of unitaries for which the heralding probability vanishes.
It is a closed subset of $\mathcal{U}(m')$ and has measure zero.

Inside $\mathcal{X}$, the subset
\begin{equation}
\mathcal{C} = \{U'\in\mathcal{X} : \mathrm{rank}[J(U')]<R_{\mathrm{HLO}}(\bm{s'},\bm{h})\},
\end{equation}
is the \emph{critical locus}, where the Jacobian rank drops below its generic value.
Here $\mathcal{C}$ is a real-analytic subset of $\mathcal{X}$ and thus has Lebesgue measure zero.

We assume now that the generic rank attains its maximal possible value,
\begin{equation}
R_{\mathrm{HLO}}(\bm{s'},\bm{h}) = 2D-1,
\end{equation}
so that $\tilde\psi$ is a submersion almost everywhere on $\mathcal{X}$, i.e. on $\mathcal{X}\setminus \mathcal{C}$.

Before addressing the specific HLO case, we recall a general topological result concerning globally surjective submersions.

\begin{theorem}[Submersion surjectivity]\label{thm:submersion_surj}
Let \(X\) be a compact smooth manifold without boundary and \(Y\) a connected smooth manifold.
Let \(f:X\to Y\) be a smooth map that is a submersion at every point of \(X\) (i.e. \(df_x\) is surjective for all \(x\in X)\).
Then \(f\) is surjective.
\end{theorem}

\begin{proof}
By the Local Submersion Theorem~\cite[Ch.~1, Sec.~4]{Guillemin2010}, a submersion is locally equivalent to a projection and hence maps small neighborhoods to open sets in \(Y\). Since \(X\) is a manifold without boundary, every point in \(X\) has a neighborhood fully contained in \(X\), so \(X\) can be covered by open sets. It follows that the image \(\mathcal{A} = f(X)\) is a union of open sets in \(Y\), and is therefore open. Since \(X\) is compact and \(f\) continuous, \(\mathcal{A}\) is compact in \(Y\). Because manifolds are Hausdorff, compact subsets are closed, so \(\mathcal{A}\) is closed in \(Y\). Therefore \(\mathcal{A}\) is both open and closed in the connected manifold \(Y\).

As \(f\) is continuous and \(X\) nonempty, \(\mathcal{A}\) is nonempty. The only nonempty subset of a connected space that is both open and closed is the whole space, hence \(\mathcal{A}=Y\).
Thus \(f\) is surjective.
\end{proof}

In the case of the HLO map, if both the singular locus and the critical locus were empty (i.e., \(\mathcal{Z} = \mathcal{C} = \emptyset\)), then \(\tilde{\psi}\) would be a smooth submersion on the entire compact manifold \(\mathcal{U}(m')\). By Theorem~\ref{thm:submersion_surj}, it would follow that
\begin{equation}
\tilde{\psi}(\mathcal{U}(m')) = S^{2D-1},
\end{equation}
and thus full global reachability would hold.

In practice, however, while the singular locus \(\mathcal{Z}\) may or may not be empty depending on the setting, the critical locus \(\mathcal{C}\) is known to be nonempty, as the rank drops below its maximal values for interferometer configurations which do not couple the output space with the heralding register.

The presence of \(\mathcal{Z}\) can prevent the image \(\mathcal{Y} = \tilde{\psi}(\mathcal{X})\) from being closed, as sequences in \(\mathcal{X}\) approaching \(\mathcal{Z}\) may accumulate at the open boundary of \(\mathcal{Y}\).

Conversely, if \(\mathcal{Z} = \emptyset\) but \(\mathcal{C} \neq \emptyset\), the domain remains compact, so \(\mathcal{Y}\) is closed by continuity; however, even though Sard’s theorem~\cite{Sard1942} ensures that the set of critical values \(\tilde{\psi}(\mathcal{C})\) has measure zero in \(S^{2D-1}\), the existence of critical points may obstruct the global openness of the map, producing effective boundaries in \(\mathcal{Y}\).

In summary, although \(\tilde{\psi}\) is a local submersion almost everywhere, the coexistence of singular and critical loci can obstruct global surjectivity. Whether their combined effect can be ruled out in the HLO setting remains an open question and for this reason global reachability of the HLO map should be regarded as a conjecture.

\section{Numerical evidence for global reachability}\label{app: global numerical}

Since an analytical proof of the density of the HLO accessible set $\mathcal{A}$ in $S^{2D-1}$ is lacking for general configurations, we complement our rank-based analysis with numerical experiments. We employ Riemannian optimization \cite{Yao2024} on the unitary manifold $\mathcal{U}(m')$ to test whether the available degrees of freedom allow the system to reach arbitrary target states.

For a given output configuration $(n, m)$, we consider varying input states of the form $\ket{\bm{s'}} = |1^{\otimes n'} 0^{\otimes (m'-n')}\rangle$, with $m'=m+k$ and $n'=n+r$, projecting on heralding patterns of the form $\ket{\bm{h}} = |1^{\otimes r} 0^{\otimes (k-r)}\rangle$. The distance between the heralded state $\ket{\psi(U')}$ and a target $\ket{\psi_{\rm target}}$ in the output space is quantified by the infidelity:
\begin{equation}
\mathcal{L}(U') = 1 - |\langle \psi_{\rm target} \mid \psi(U') \rangle|^2.
\end{equation}
To minimize $\mathcal{L}(U')$, we implement an iterative Riemannian gradient descent algorithm directly on the unitary manifold $\mathcal{U}(m')$. At each step $t$, the unitary $U'_t$ is updated via the exponential map:
\begin{equation}
U'_{t+1} = e^{-\eta_t \Omega_t}U'_t,
\end{equation}
where $\Omega_t$ is an anti-Hermitian operator belonging to the Lie algebra $\mathfrak{u}(m')$ and $\eta_t$ is the learning rate. Specifically, $\Omega_t$ is the analytical Riemannian gradient of the loss function, determined by projecting the Euclidean gradient onto the tangent space of the manifold at the current point $U'_t$. This approach ensures that each update step remains intrinsically within the unitary group.

\begin{figure*}[t]
    \centering
    \includegraphics[width=0.9\textwidth]{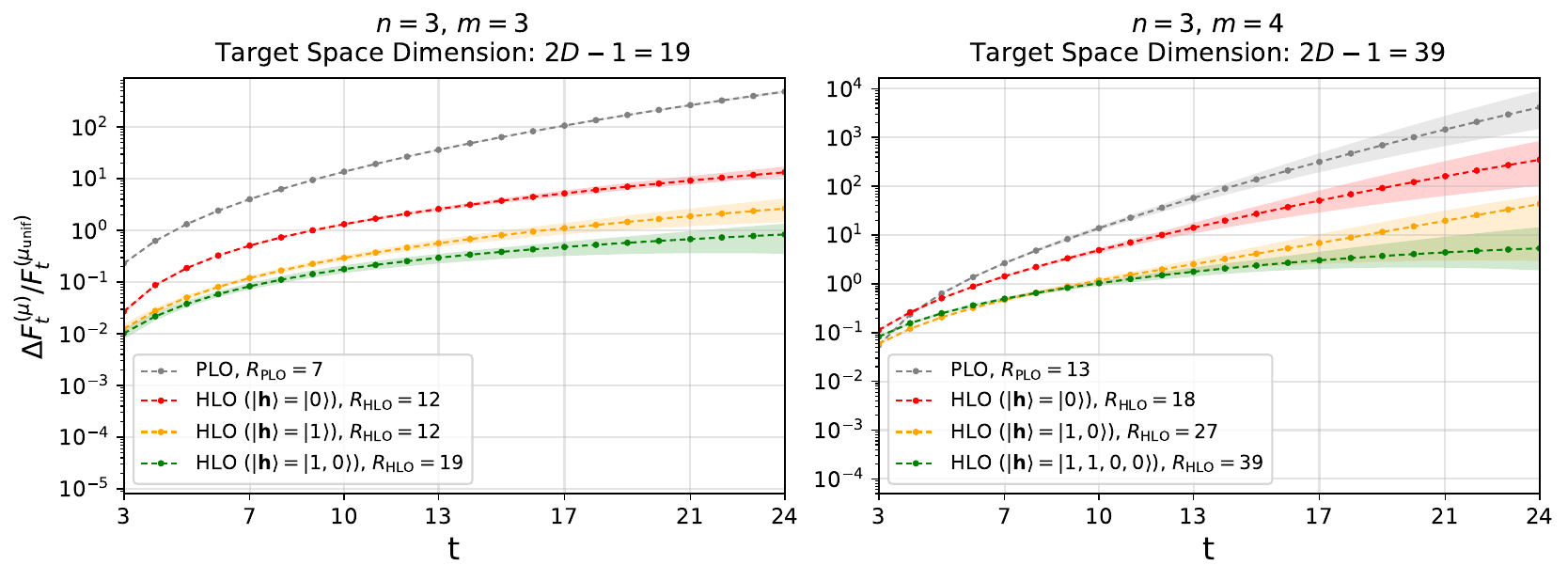}
    \caption{Relative difference between the frame potentials of output states $\ket{\psi(U',\bm{s'},\bm{h})}$ of PLO/HLO circuits, and the ones of the uniform measure on all pure states. The configurations considered are exactly the same as those in Figure~\ref{fig:opt_infidelity}. For each configuration, $N_{\rm samples}=3000$ Haar random matrices  $U' \in \mathcal{U}(m')$ were drawn, and the normalized output states $\ket{\psi}=\ket{\psi(U',\bm{s'},\bm{h})}$ computed. Estimates of the $t$-th frame potential $F_t^{(\mu)}$ (solid circles) are then calculated as the empirical average of $\left| \bra{\psi_1}\ket{\psi_2} \right|^{2t}$ over all distinct pairs $(\ket{\psi_1},\ket{\psi_2})$ of states obtained. The relative differences with the uniform measure's frame potential are then computed using Eq.~\eqref{eq: frame-potential-uniform}. Lastly, this whole process is repeated $N_{\rm repeat}=30$ times, producing $N_{\rm repeat}$ new estimates of each frame potential $F_t^{(\mu)}$, according to which the $5^{\rm{th}}$-$95^{\rm{th}}$ percentile range is evaluated (shaded regions).}
    \label{fig:frame-potentials-relative-difference}
\end{figure*}

To ensure the statistical robustness of our results and to mitigate the impact of local minima, we adopt the following protocol. We draw $M=200$ target states $\ket{\psi_{\rm target}}$ sampled uniformly according to the Haar measure on the data Fock space. For each target state, we perform $10$ independent optimization runs starting from different Haar-random initial unitaries $U'_0$. For each target, only the run achieving the lowest final infidelity is retained. Each run consists of $T=300$ iterations. We use an exponentially decaying learning rate $\eta_t$ to ensure fast initial descent and stable convergence in the late stages of the optimization.

Figure~\ref{fig:opt_infidelity} illustrates the optimization dynamics for different resource configurations in low-dimensional settings, with $n=3$ output photons and $m=\{3,4\}$ output modes. We plot the evolution of the infidelity $\mathcal{L}$ as a function of the optimization steps. The dashed lines represent the geometric mean of the best infidelity across the $100$ target states, while the shaded regions indicate the geometric standard deviation, capturing the performance variability across the Hilbert space.

The numerical results highlight a clear divergence in the convergence regimes. When the auxiliary resources are insufficient to saturate the Jacobian rank (e.g., standard PLO or HLO with minimal auxiliary resources), the infidelity quickly reaches a plateau. This saturation suggests that the target state lies at a finite distance from the restricted accessible manifold $\mathcal{A}$. In contrast, when the auxiliary configuration $\ket{\bm{h}}$ is chosen such that the Jacobian rank attains the maximal value $2D-1$ (highlighted in green), the loss exhibits a sustained descent, reaching values orders of magnitude lower within the same optimization budget. In the cases where the rank is maximized, the infidelity reaches values several orders of magnitude lower than the PLO baseline within the allocated optimization budget. The absence of an early plateau in these configurations provides strong numerical evidence in support of the Global Reachability Conjecture, suggesting that saturating the Jacobian rank may indeed be sufficient to render the accessible set dense in the target Hilbert space.

\section{Numerical exploration of increased expressivity through frame potentials}\label{app: frame potential}

In this appendix, we complement our optimization-based analysis of global reachability by investigating the global expressivity of HLO architectures from an ensemble perspective. Specifically, we evaluate how closely the ensemble of generated heralded states mimics a truly uniform distribution over the state space. 

Consider the ensemble of normalized output states $\ket{\psi(U',\bm{s'},\bm{h})}$ obtained by drawing $U$ from the Haar measure on $\mathcal{U}(m')$ (with fixed input and heralding patterns $\bm{s'},\bm{h}$). Denote by $\mu = \mu_{\bm{s'},\bm{h}}$ the induced probability measure on the complex projective space $\mathbb{CP}^{D-1}$ (the data sphere with points identified if they differ by a complex global phase). As the resources $\bm{s'}$ and $\bm{h}$ scale, one may expect that the measure $\mu_{\bm{s'},\bm{h}}$ gets ``closer" to the uniform measure $\mu_{\rm unif}$ on $\mathbb{CP}^{D-1}$. A common way to quantify how a measure $\mu$ on $\mathbb{CP}^{D-1}$  resembles the uniform measure, is by comparing their so-called \textit{frame potentials}. We now recall the necessary notions regarding frame potentials, and refer to \cite{Ambainis2007,Mele2024} for more background.
A measure $\mu$ on $\mathbb{CP}^{D-1}$ is said to be a (complex projective) \textit{$t$-design} if all of its moments up to order $t$ coincide with those of $\mu_{\rm unif}$.
The \textit{$t$-th frame potential} of a measure $\mu$ on $\mathbb{CP}^{D-1}$ is defined as the quantity
\begin{equation}
F_t^{(\mu)} = \mathbb{E}_{\psi_1,\psi_2 \sim \mu} \left[ \left| \bra{\psi_1}\ket{\psi_2} \right|^{2t} \right].
\end{equation}
It holds that (i) $F_t^{(\mu)} \geq F_t^{(\mu_{\rm unif})}$, (ii) 
there is equality if and only if  $\mu$ is a $t$-design, (iii) and that 
\begin{equation}\label{eq: frame-potential-uniform}
F_t^{(\mu_{\rm unif})} = {D + t -1 \choose t}^{-1}.
\end{equation}
Therefore, the difference
\begin{equation}
\Delta F_t^{(\mu)} = F_t^{(\mu)} - F_t^{(\mu_{\rm unif})}
\end{equation}
serves as a natural quantifier of a measure $\mu$'s closeness to being a $t$-design.
This metric has been explored as a measure of expressivity of variational quantum circuits \cite{Sim-ExpressibilityEntangling-2019,Nakaji2021expressibilityof}.

In Figure~\ref{fig:frame-potentials-relative-difference}, we report numerical estimations of the relative frame potential differences $\Delta F_t^{(\mu)}/F_t^{(\mu_{\rm unif})}$, for the same PLO and HLO configurations as those studied in Figure \ref{fig:opt_infidelity}. Note that, in both PLO settings considered (gray curves), from $t\geq 10$ onwards the frame potentials differ (in relative difference) from the uniform values by more than an order of magnitude.
On the other hand, for both HLO settings considered that fulfill the local controllability criterion (green curves), we observe that the frame potentials can get closer to uniform (in relative difference) by several orders of magnitude; for instance, they remain below an order of magnitude throughout ($t\leq 24$).
The specific amounts, though, appear to depend on the configurations: in the $(n=3, m=3)$ setting (left), the relative difference remains below $1$ throughout ($t\leq 24$), while in the $(n=3, m=4)$ setting (right) this only holds for $t\leq 9$.
It is interesting that for some ``weaker" heralding resources, i.e. settings where our local controllability criterion is not fulfilled (red and yellow curves), the frame potentials still get closer to uniform (compared to without heralding), sometimes by more than an order of magnitude (in relative difference). This suggests that in settings with experimentally limited ancillary/heralding resources, which are insufficient to achieve local controllability of the data sphere, HLO may still offer a considerable increase of expressivity, this time measured by the output state ensemble's closeness to being a $t$-design.

\end{document}